\pdfoutput=1
\documentclass[11pt]{article}
\usepackage[margin=1in]{geometry}
\usepackage[T1]{fontenc}
\usepackage{lmodern}
\usepackage{amsmath,amssymb,bm}
\usepackage{amsthm}
\usepackage{graphicx,xcolor,booktabs,enumitem,authblk,adjustbox}
\usepackage{algorithm,algpseudocode}
\usepackage[font=small,labelfont=bf,labelsep=colon]{caption}
\usepackage[hyphens]{url}
\definecolor{linkblue}{RGB}{25,25,112}
\usepackage[colorlinks=true,linkcolor=linkblue,citecolor=linkblue,urlcolor=black]{hyperref}
\theoremstyle{definition}
\newtheorem{theorem}{Theorem}
\newtheorem{corollary}{Corollary}
\newtheorem{proposition}{Proposition}
\algrenewcommand\algorithmicrequire{\textbf{Given}}
\renewcommand{\topfraction}{0.9}\renewcommand{\bottomfraction}{0.8}
\renewcommand{\textfraction}{0.07}\renewcommand{\floatpagefraction}{0.8}
\newcommand{\dropcap}[1]{#1}
\newcommand{\Endparasplit}{}
\newcommand{\JournalTitle}[1]{#1}
\newcommand{\addtabletext}[1]{\par\smallskip{\footnotesize #1\par}}

\newcommand{\SfourCellsExt}{432}

\newcommand{\SfourMpair}{50,000}
\newcommand{\SfourMref}{400,000}

\newcommand{\SfourRefUcb}{0.04749}

\newcommand{\SoneBeNparMax}{0.0486}
\newcommand{\SoneBeParMax}{0.0504}
\newcommand{\SoneBeRatio}{17}
\newcommand{\SoneBeWorst}{0.0587}

\newcommand{\SoneMaxDev}{5.4}
\newcommand{\SoneNrepA}{2,000,000}
\newcommand{\SoneNrepB}{400,000}
\newcommand{\SoneTabRQAdvTh}{0.50000}
\newcommand{\SoneTabRQAdv}{0.50029}
\newcommand{\SoneTabRQIndTh}{0.40126}
\newcommand{\SoneTabRQInd}{0.40141}
\newcommand{\SoneTabSAdvTh}{0.10000}
\newcommand{\SoneTabSAdv}{0.09999}
\newcommand{\SoneTabSIndTh}{0.09750}
\newcommand{\SoneTabSInd}{0.09757}
\newcommand{\SoneTabTAdvTh}{0.15000}
\newcommand{\SoneTabTAdv}{0.14992}
\newcommand{\SoneTabTIndTh}{0.14263}
\newcommand{\SoneTabTInd}{0.14253}
\newcommand{\SoneTabUAdvTh}{0.20000}
\newcommand{\SoneTabUAdv}{0.19985}
\newcommand{\SoneTabUIndTh}{0.18549}
\newcommand{\SoneTabUInd}{0.18523}
\newcommand{\SoneTabVAdvTh}{0.25000}
\newcommand{\SoneTabVAdv}{0.25054}
\newcommand{\SoneTabVIndTh}{0.22622}
\newcommand{\SoneTabVInd}{0.22659}
\newcommand{\SoneTabYAdvTh}{0.40000}
\newcommand{\SoneTabYAdv}{0.39973}
\newcommand{\SoneTabYIndTh}{0.33658}
\newcommand{\SoneTabYInd}{0.33657}

\newcommand{\SthreeAtC}{0.28}
\newcommand{\SthreeFactor}{2.4}

\newcommand{\SthreeM}{200,000}
\newcommand{\SthreeSwitch}{0.2936}
\newcommand{\SthreeTheta}{0.8926}
\newcommand{\SthreeWorstA}{0.1036}
\newcommand{\SthreeWorstB}{0.1141}
\newcommand{\SthreeWorstC}{0.1204}
\newcommand{\StwoAlphaZero}{0.0452}
\newcommand{\StwoCells}{252}

\newcommand{\StwoGamma}{0.0048}
\newcommand{\StwoMaxTie}{0.05159}
\newcommand{\StwoMaxUcbRes}{0.04933}
\newcommand{\StwoMaxUcb}{0.05411}
\newcommand{\StwoM}{100,000}
\newcommand{\StwoPowerLossMax}{2.88}
\newcommand{\StwoPowerLossMean}{1.27}
\newcommand{\StwoResARQQQ}{0.0016}

\newcommand{\StwoResAVQ}{0.0066}

\newcommand{\EthreeM}{3}
\newcommand{\EthreeK}{5}
\newcommand{\EthreeAlphaZero}{0.045}
\newcommand{\EthreeCobf}{2.077}
\newcommand{\EthreeRefUcb}{0.0466}
\newcommand{\EthreePeekTie}{0.142}
\newcommand{\EthreePeekBound}{0.1449}
\newcommand{\EthreePeekPower}{0.848}
\newcommand{\EthreeMetricTie}{0.144}
\newcommand{\EthreeMetricBound}{0.1466}
\newcommand{\EthreeMetricPower}{0.827}
\newcommand{\EthreeBothTie}{0.370}
\newcommand{\EthreeBothBound}{0.3737}
\newcommand{\EthreeBothPower}{0.897}
\newcommand{\EthreeObfTie}{0.046}
\newcommand{\EthreeObfBound}{0.0481}
\newcommand{\EthreeObfPower}{0.776}
\newcommand{\EthreeBonfTie}{0.045}
\newcommand{\EthreeBonfBound}{0.0482}
\newcommand{\EthreeBonfPower}{0.654}
\newcommand{\EthreeRefPower}{0.785}
\newcommand{\EthreeObfInfo}{81}

\newcommand{\EtwoAPe}{95.09}
\newcommand{\EtwoALo}{90.76}
\newcommand{\EtwoAHi}{99.62}
\newcommand{\EtwoAPubpe}{95.09}
\newcommand{\EtwoAPublo}{90.76}
\newcommand{\EtwoAPubhi}{99.62}
\newcommand{\EtwoAN}{18}
\newcommand{\EtwoBPe}{71.10}
\newcommand{\EtwoBLo}{51.45}
\newcommand{\EtwoBHi}{98.26}
\newcommand{\EtwoBPubpe}{71.10}
\newcommand{\EtwoBPublo}{51.45}
\newcommand{\EtwoBPubhi}{98.26}
\newcommand{\EtwoBN}{18}
\newcommand{\EtwoCPe}{58.56}
\newcommand{\EtwoCLo}{39.41}
\newcommand{\EtwoCHi}{87.03}
\newcommand{\EtwoCPubpe}{58.56}
\newcommand{\EtwoCPublo}{39.41}
\newcommand{\EtwoCPubhi}{87.03}
\newcommand{\EtwoCN}{13}
\newcommand{\EtwoDPe}{71.10}
\newcommand{\EtwoDLo}{51.45}
\newcommand{\EtwoDHi}{98.26}
\newcommand{\EtwoDPubpe}{71.10}
\newcommand{\EtwoDPublo}{51.45}
\newcommand{\EtwoDPubhi}{98.26}
\newcommand{\EtwoDN}{18}
\newcommand{\EtwoEPe}{91.83}
\newcommand{\EtwoELo}{55.71}
\newcommand{\EtwoEHi}{151.37}
\newcommand{\EtwoEPubpe}{91.83}
\newcommand{\EtwoEPublo}{55.71}
\newcommand{\EtwoEPubhi}{151.37}
\newcommand{\EtwoEN}{18}
\newcommand{\EtwoFPe}{99.89}
\newcommand{\EtwoFLo}{93.37}
\newcommand{\EtwoFHi}{106.86}
\newcommand{\EtwoFPubpe}{99.89}
\newcommand{\EtwoFPublo}{93.37}
\newcommand{\EtwoFPubhi}{106.86}
\newcommand{\EtwoFN}{100}
\newcommand{\EtwoGPe}{92.15}
\newcommand{\EtwoGLo}{88.46}
\newcommand{\EtwoGHi}{95.99}
\newcommand{\EtwoGPubpe}{92.15}
\newcommand{\EtwoGPublo}{88.46}
\newcommand{\EtwoGPubhi}{95.99}
\newcommand{\EtwoGN}{1000}
\newcommand{\EtwoHPe}{93.42}
\newcommand{\EtwoHLo}{86.81}
\newcommand{\EtwoHHi}{100.55}
\newcommand{\EtwoHPubpe}{93.42}
\newcommand{\EtwoHPublo}{86.81}
\newcommand{\EtwoHPubhi}{100.55}
\newcommand{\EtwoHN}{717}

\newcommand{\EtwoNsets}{8}
\newcommand{\EtwoDSoneCVwR}{46.96}
\newcommand{\EtwoDSonePubCVwR}{47.0}
\newcommand{\EtwoDStwosixPE}{151.29}
\newcommand{\EtwoDStwosixPubPE}{151.3}

\newcommand{\EfourRounds}{8}
\newcommand{\EfourAdopted}{5}

\newcommand{\EfourEscStart}{0.643}
\newcommand{\EfourEscFinal}{0.119}
\newcommand{\EfourPowerStart}{0.255}
\newcommand{\EfourPowerFinal}{0.566}
\newcommand{\EfourMaxBound}{0.0478}
\newcommand{\EfourCostUsd}{0.20}
\newcommand{\EfourDeltaFirst}{0.0152}
\newcommand{\EfourDeltaLast}{0.00024}
\newcommand{\EfourRecomputeWall}{25 seconds}

\newcommand{\EfourGRChange}{dropout\_route $\to$ completers}

\newcommand{\EfourGRBound}{0.0474}
\newcommand{\EfourGSChange}{qc\_remediation $\to$ declared\_mapping}

\newcommand{\EfourGSBound}{0.0477}
\newcommand{\EfourGTChange}{sensitivity\_role $\to$ promote\_if\_supports}

\newcommand{\EfourGTBound}{0.0626}
\newcommand{\EfourGUChange}{dropout\_route $\to$ information\_recovery}

\newcommand{\EfourGUBound}{0.0618}
\newcommand{\EfourGVChange}{high\_cv\_escalate\_above $\to$ 0.5}

\newcommand{\EfourGVBound}{0.0477}
\newcommand{\EfourGWChange}{exclusion\_threshold $\to$ 3.0}

\newcommand{\EfourGWBound}{0.0501}
\newcommand{\EfourGXChange}{exclusion\_threshold $\to$ 3.5}

\newcommand{\EfourGXBound}{0.0478}
\newcommand{\EfourGYChange}{sensitivity\_role $\to$ escalate\_on\_disagreement}

\newcommand{\EfourGYBound}{0.0477}

\newcommand{\EoneSoneCQNamedEsc}{0}

\newcommand{\EoneSoneCQNamedN}{550}

\newcommand{\EoneSoneCSNamedEsc}{65}

\newcommand{\EoneSoneCTNamedEsc}{42}

\newcommand{\EoneSoneCUNoneEsc}{241}

\newcommand{\EwvHkSubX}{1}
\newcommand{\EwvHkSubN}{402}

\newcommand{\EwvSnSubX}{0}
\newcommand{\EwvSnSubN}{750}

\newcommand{\EwvOpSubX}{0}
\newcommand{\EwvOpSubN}{402}

\newcommand{\EtruthGR}{0.0392}
\newcommand{\EtruthGS}{0.0427}
\newcommand{\EtruthGT}{0.0564}
\newcommand{\EtruthGU}{0.0548}
\newcommand{\EtruthGV}{0.0427}
\newcommand{\EtruthGW}{0.0450}
\newcommand{\EtruthGX}{0.0429}
\newcommand{\EtruthGY}{0.0417}
\newcommand{\EtruthAdmMax}{0.0429}

\newcommand{\ENoGateMax}{0.0658}
\newcommand{\ENoGateInvalidGens}{5}
\newcommand{\ErepN}{25}
\newcommand{\ErepAdoptLo}{3}
\newcommand{\ErepAdoptHi}{6}
\newcommand{\ErepEscLo}{0.076}
\newcommand{\ErepEscHi}{0.223}
\newcommand{\ErepPowLo}{0.541}
\newcommand{\ErepPowHi}{0.580}

\newcommand{\ErepProposals}{191}
\newcommand{\ErepAdoptMed}{5}
\newcommand{\ErepEscMed}{0.115}
\newcommand{\ErepPowMed}{0.569}
\newcommand{\ErepInvalidProposed}{47}
\newcommand{\ErepConservRefusals}{14}
\newcommand{\ErepUniqueAdm}{33}
\newcommand{\ErepAdmTruthMax}{0.0450}

\newcommand{\EscrRounds}{12}
\newcommand{\EscrRefused}{8}
\newcommand{\EscrAdopted}{4}
\newcommand{\EscrPow}{0.579}
\newcommand{\EscrEsc}{0.076}
\newcommand{\EscrMaxAdm}{0.0478}

\newcommand{\ErepRoundsPlanned}{200}
\newcommand{\ErepStops}{7}
\newcommand{\ErepInterfaceRej}{0}

\newcommand{\EfiveCells}{12}
\newcommand{\EfiveM}{50{,}000}
\newcommand{\EfiveCoreUcb}{0.0448}
\newcommand{\EfivePickWorst}{0.0684}
\newcommand{\EfivePickBound}{0.0735}
\newcommand{\EfiveTauWorst}{0.0843}
\newcommand{\EfiveTauBound}{0.0924}
\newcommand{\EfiveRtWorst}{0.0554}
\newcommand{\EfiveRtBound}{0.0603}
\newcommand{\EfiveCbWorst}{0.0343}
\newcommand{\EfiveCbBound}{0.0411}
\newcommand{\EfivePowCore}{0.432}
\newcommand{\EfivePowCombo}{0.627}
\newcommand{\EfivePowRt}{0.220}
\newcommand{\EfivePowPick}{0.475}
\newcommand{\EnvFlagWorst}{0.0239}
\newcommand{\EnvCoinWorst}{0.0220}
\newcommand{\EnvAdvWorst}{0.0343}
\newcommand{\EnvFlagPow}{0.542}
\newcommand{\EnvCoinPow}{0.484}
\newcommand{\EnvAdvPow}{0.627}
\newcommand{\EnvFhShareNull}{10\%}
\newcommand{\EnvFhShareAlt}{54\%}
\newcommand{\EfiveRcalAlpha}{0.034}
\newcommand{\EfiveRcalWorst}{0.0393}
\newcommand{\EfiveRcalBound}{0.0450}
\newcommand{\EfivePowRcal}{0.600}
\newcommand{\EfiveAlphaZero}{0.040}

\newcommand{\FoneGHalf}{14}
\newcommand{\FoneGHalfStrict}{139}
\newcommand{\FonePairedFirst}{1{,}414}
\newcommand{\FoneScratchFirst}{28{,}008}

\newcommand{\FonePairedLast}{4{,}050}

\newcommand{\FoneCostPairedA}{1{,}414}
\newcommand{\FoneCostScratchA}{28{,}008}
\newcommand{\FoneCostFactorA}{19.8}
\newcommand{\FoneCostPairedB}{3{,}700}
\newcommand{\FoneCostScratchB}{197{,}187}
\newcommand{\FoneCostFactorB}{53}
\newcommand{\FoneCostPairedC}{7{,}402}
\newcommand{\FoneCostScratchC}{796{,}052}
\newcommand{\FoneCostFactorC}{108}
\newcommand{\FoneNestedPairedA}{4{,}479}
\newcommand{\FoneNestedScratchA}{50{,}407}
\newcommand{\FoneNestedFactorA}{11.3}
\newcommand{\FoneNestedPairedB}{11{,}928}
\newcommand{\FoneNestedScratchB}{114{,}862}
\newcommand{\FoneNestedFactorB}{9.6}
\newcommand{\FoneNestedPairedC}{55{,}248}
\newcommand{\FoneNestedScratchC}{465{,}016}
\newcommand{\FoneNestedFactorC}{8.4}
\newcommand{\FoneDiscCombinedA}{5{,}781}
\newcommand{\FoneDiscFactorA}{4.8}
\newcommand{\FoneDiscCombinedB}{14{,}250}
\newcommand{\FoneDiscFactorB}{2.0}
\newcommand{\FoneDiscCombinedC}{22{,}875}
\newcommand{\FoneDiscFactorC}{1.2}
\newcommand{\FoneReserveLossLow}{2.3}
\newcommand{\FoneReserveLossHigh}{4.4}
\newcommand{\FoneReservePowerLow}{74}
\newcommand{\FoneReservePowerHigh}{90}

\newcommand{\FoneRandKFirst}{580}
\newcommand{\FoneRandKLast}{1{,}962}

\newcommand{\FoneRandM}{64{,}845}

\providecommand{\NumRuns}{twenty}
\providecommand{\NumRounds}{ten}
\providecommand{\SimBE}{500{,}000}
\providecommand{\SimGB}{500{,}000}
\providecommand{\EvalSims}{$10^{6}$}

\providecommand{\BEnoCheckMax}{0.098}   %
\providecommand{\BEfixedEver}{20}    %
\providecommand{\BEltPowerShare}{52} %
\providecommand{\BEcvMax}{0.045}        %
\providecommand{\BEcvPowerShare}{82} %
\providecommand{\BEoffGridMax}{0.058}   %
\providecommand{\BErandPowerShare}{79}   %
\providecommand{\GBnoCheckMax}{0.183}

\providecommand{\GBgateEver}{95}
\providecommand{\GBltPowerShare}{80}

\providecommand{\GBcvMax}{0.044}
\providecommand{\GBcvPowerShare}{96}
\providecommand{\GBoffGridMax}{0.044}
\providecommand{\GBallEver}{60} %
\providecommand{\GBgateMax}{0.817} %

\providecommand{\GBcvJudged}{200}
\providecommand{\GBcvDeployedN}{170}
\providecommand{\GBcvRejectedN}{30}
\providecommand{\GBcvRejAbove}{28}

\providecommand{\GBcvRejValid}{2}
\providecommand{\GBrefSeconds}{296}

\providecommand{\BEcvDistinct}{112}
\providecommand{\BEcvCodeShare}{32}
\providecommand{\BEcvRunsFinalCode}{18}
\providecommand{\BEltCode}{37}
\providecommand{\BEcvJudged}{198}
\providecommand{\BEcvDeployedN}{117}
\providecommand{\BEcvRejectedN}{81}
\providecommand{\BEcvRejAbove}{62}
\providecommand{\BEcvRejIndet}{19}

\providecommand{\BEcvPowerShareSmall}{74}
\providecommand{\BEcvRejValidSmall}{17}
\providecommand{\BEfixedEverPrime}{30}
\providecommand{\BEcvSecondsPerVal}{1101}
\providecommand{\BEcvDatasetsPerVal}{447}
\providecommand{\BErefSeconds}{668}

\providecommand{\GBcvDistinct}{62}
\providecommand{\GBcvCodeShare}{21}
\providecommand{\GBcvRunsFinalCode}{15}
\providecommand{\GBltCode}{18}

\providecommand{\GBcvSecondsPerVal}{329}
\providecommand{\GBcvDatasetsPerVal}{29}
\providecommand{\OWqwenMalformed}{7}
\providecommand{\OWdeepseekMalformed}{21}

\usepackage{mathtools,array,placeins}
\theoremstyle{definition}
\newtheorem{SItheorem}{Theorem}

\newtheorem{SIproposition}[SItheorem]{Proposition}
\newtheorem{SIcorollary}[SItheorem]{Corollary}
\newtheorem{lemma}[SItheorem]{Lemma}
\newtheorem{definition}[SItheorem]{Definition}

\newtheorem{remark}[SItheorem]{Remark}
\newtheorem{assumption}{Assumption}

\newcommand{\E}{\mathbb{E}}

\newcommand{\SIref}[1]{\ref{#1}}

\newcommand{\SIpart}{%
  \clearpage
  \setcounter{section}{0}\setcounter{subsection}{0}%
  \renewcommand{\thesection}{A\arabic{section}}%
  \renewcommand{\thesubsection}{\thesection.\arabic{subsection}}%
  \def\theHsection{SI.\arabic{section}}%
  \def\theHsubsection{SI.\arabic{section}.\arabic{subsection}}%
  \setcounter{equation}{0}\renewcommand{\theequation}{A\arabic{equation}}%
  \def\theHequation{SI.\arabic{equation}}%
  \setcounter{figure}{0}\renewcommand{\thefigure}{A\arabic{figure}}%
  \def\theHfigure{SI.\arabic{figure}}%
  \setcounter{table}{0}\renewcommand{\thetable}{A\arabic{table}}%
  \def\theHtable{SI.\arabic{table}}%
  \setcounter{algorithm}{0}\renewcommand{\thealgorithm}{A\arabic{algorithm}}%
  \def\theHalgorithm{SI.\arabic{algorithm}}%
  \let\theorem\SItheorem\let\endtheorem\endSItheorem
  \let\proposition\SIproposition\let\endproposition\endSIproposition
  \let\corollary\SIcorollary\let\endcorollary\endSIcorollary
  \renewcommand{\topfraction}{0.9}\renewcommand{\bottomfraction}{0.8}%
  \renewcommand{\textfraction}{0.07}\renewcommand{\floatpagefraction}{0.85}%
  \setcounter{topnumber}{3}\setcounter{totalnumber}{4}%
}

\newcommand{\figpending}[2]{\fbox{\parbox[c][#2][c]{0.96\linewidth}{\centering\sffamily\small #1}}}

\newcommand{\rev}[1]{#1}

\title{Lifetime-valid decisions by continually evolving AI analysts}
\author[1]{Xiaojing Sun}
\author[1]{Yuhan Zeng}
\author[1]{Xiao Wang}
\affil[1]{Department of Statistics, Purdue University}
\date{}

\begin{document}
\maketitle

\begin{abstract}
AI agents increasingly act as analysts whose decision can often be cast as hypothesis tests.
Classical tests control the Type I error rate, the probability of rejecting a true null hypothesis, for fixed procedures. However, AI analysts introduce internal randomness and can evolve continuously over time.
As AI analysts evolve, the risk of approving versions with Type I error rates above the target can accumulate, and adaptive use of evaluation feedback can lead to overfitting.
We introduce continual validation, which evaluates
each proposed AI analyst version before deployment
using simulated data under the null hypothesis,
with scenarios that can be randomized and expanded over time.
We prove that, regardless of how updates are chosen
or how many are proposed, the lifetime probability
of ever deploying a version whose Type I error rate
exceeds the allowed level remains bounded.
This guarantees that every subsequent decision
controls the Type I error rate at the target level across these scenarios.
At the same time, updates with sufficiently low
Type I error rates pass validation with high probability,
allowing the analyst to continuously improve its
performance without sacrificing rigor.
In two examples involving evolving AI analysts for
bioequivalence studies and product A/B tests, Type I error rates reached \BEnoCheckMax\ and
\GBnoCheckMax\ without continual validation, compared with at most \BEcvMax\ and
\GBcvMax\ with it, while retaining 82\% and 96\% of the power from optimally selected valid updates.
\end{abstract}

\section{Introduction}
\dropcap{A}s large language models and the systems
built around them advance, AI agents are increasingly
used to analyze data.
They plan analyses, write and execute code, and report findings
in tasks such as developing machine learning methods,
studying cellular processes, and discovering new drugs
\cite{lu2024aiscientist,swanson2025virtuallab,alber2026cellvoyager,zhang2026virtualbiotech}.
Recent benchmarks assess these capabilities
across scientific and data science tasks
\cite{chen2025scienceagentbench,majumder2025discoverybench,chan2025mlebench,nie2026dsgym}.
We refer to agents performing such analyses as AI analysts.

Although these analyses may involve complex
workflows, including multiple stages of data processing,
model selection, diagnostic evaluation, and iterative
refinement, they often lead to decisions such as reporting
a discovery, claiming bioequivalence of a drug product, or launching a product.
Such decisions can often be formulated as hypothesis tests,
which AI agents are increasingly trained to perform
\cite{huang2025popper,miao2026fisherr1}.
For example, in bioequivalence studies \cite{schuirmann1987comparison,ich2024m13a}, the null hypothesis
states that the test product is not bioequivalent to
the reference product.
A Type I error occurs when a true null hypothesis is rejected.
Similarly, when deciding whether to launch a new product
feature, the null hypothesis states that the feature provides
no improvement, while the alternative states that it does.
Classical statistical methods control the probability of this error
at a prespecified level $\alpha$, typically for procedures fixed
before the data are observed, where randomness comes primarily
from data sampling \cite{neyman1933problem,lehmann2022testing}.

AI analysts, however, pose two additional challenges to
Type I error control.
First, even a fixed version can make different decisions on the same
data. 
A version fixes its model, prompts, and available tools,
but sampling from the underlying language model may change
which analyses and tools are selected
\cite{bertran2026manyai,miao2026agentic}.
Each version therefore defines a randomized test, whose Type I error
rate accounts for both data sampling and the analyst's internal randomness \cite{tocher1950extension,lehmann2022testing}.
\Endparasplit

Second, the analyst itself can change frequently in many ways.
Developers edit prompts, add skills, redesign workflows, and upgrade
models \cite{chen2024chatgptdrift}.
Self-improving agents can also repeatedly revise themselves to evolve
\cite{schmidhuber2007godel,zelikman2024stop,zhang2026darwin,shinn2023reflexion,yuksekgonul2025textgrad,zhang2026ace}.
Together, these changes can produce a rapid and potentially
unlimited sequence of versions, each modifying multiple parts
of the analyst.
We call each proposed change an update and the resulting
configuration a new version of the analyst.
Updates are often retained when they appear to improve
the analyst's performance, but without appropriate validation,
an increase in their Type I error rate may go undetected.

This adaptive selection makes the problem different from controlling
a single randomized test.
Updates can be selected based on feedback from earlier
analyses and evaluations, potentially causing the analyst
to overfit past results and inflate its Type I error rate \cite{simmons2011false,silberzahn2018many}.
Even a single update can substantially change the Type I error rate.
For example, at a target Type I error level of 0.10,
replacing Claude Sonnet 3.5 with Claude
Haiku 3.5 in a hypothesis-testing agent increased
its estimated Type I error rate from 0.082 to 0.780
on one benchmark \cite{huang2025popper}.

Existing methods provide important statistical guarantees,
but they address different problems.
Prediction-powered inference enables valid statistical
inference using machine learning predictions
\cite{angelopoulos2023ppi}, while conformal methods such
as TxConformal control false discoveries when AI predictions
are used to select therapeutic candidates
\cite{jin2026txconformal}.
Agentic falsification controls Type I error for a specified
agent performing sequential hypothesis testing
\cite{huang2025popper}.
Other methods address adaptive analysis choices and repeated testing.
The reusable holdout allows the same dataset to be used
repeatedly to evaluate adaptively chosen analyses
\cite{dwork2015reusable}, while selective inference
accounts for data-dependent analysis choices
\cite{benjamini2020selective}.
Online multiple testing controls false rejections across
a sequence of hypotheses \cite{tian2021online},
and adaptive clinical trial designs allow prespecified
modifications within a trial \cite{bauer1994evaluation}.

Other work examines how learning systems can be updated safely.
High-confidence safety frameworks evaluate
candidate algorithms against specified constraints
\cite{thomas2019preventing}, while approval policies for
medical software control the risk of harmful updates and
gradual deterioration in performance
\cite{feng2021approval}.
Recent methods for self-improving agents provide statistical
guarantees for accepting beneficial modifications
\cite{sengupta2026selfevolving,sun2026reuse},
while Learn then Test calibrates a prespecified set of
configurations to satisfy statistical risk constraints
\cite{angelopoulos2025learn}.

However, these methods do not directly guarantee Type I error control for AI analysts that continually evolve through adaptive updates to their models, tools, and decision rules. Our goal is to provide a lifetime guarantee by controlling the risk of ever deploying a version whose Type I error exceeds the required level. This allows the analyst to continue updating while maintaining Type I error control in future studies. We also aim to preserve statistical power at a reasonable computation cost.

We introduce continual validation, shown in
Fig.~\ref{fig:concept}\textit{A}, which allows AI analysts
to keep improving while maintaining Type I error control
at modest cost.
Developers or the analyst itself can propose any updates.
Before deployment, each proposed version is fixed and tested
on newly simulated datasets where the null hypothesis is true.
We set a validation level $\alpha'$ slightly below $\alpha$.
We prove that, regardless of how updates are proposed,
the probability of ever deploying a version whose Type I
error rate exceeds $\alpha'$ under the simulated null
conditions is at most a prespecified $\delta$.
This guarantees Type I error control at the prespecified level
$\alpha$ for every subsequent decision.
At the same time, updates with sufficiently low Type I
error rates can still be approved with high probability.
We apply it to two AI analysts, one for bioequivalence
studies of test products and another for product launching
experiments.
We compare it with four alternative deployment strategies,
and Fig.~\ref{fig:results} summarizes the comparison.

\begin{figure}[tp]
\centering
\includegraphics[width=\textwidth]{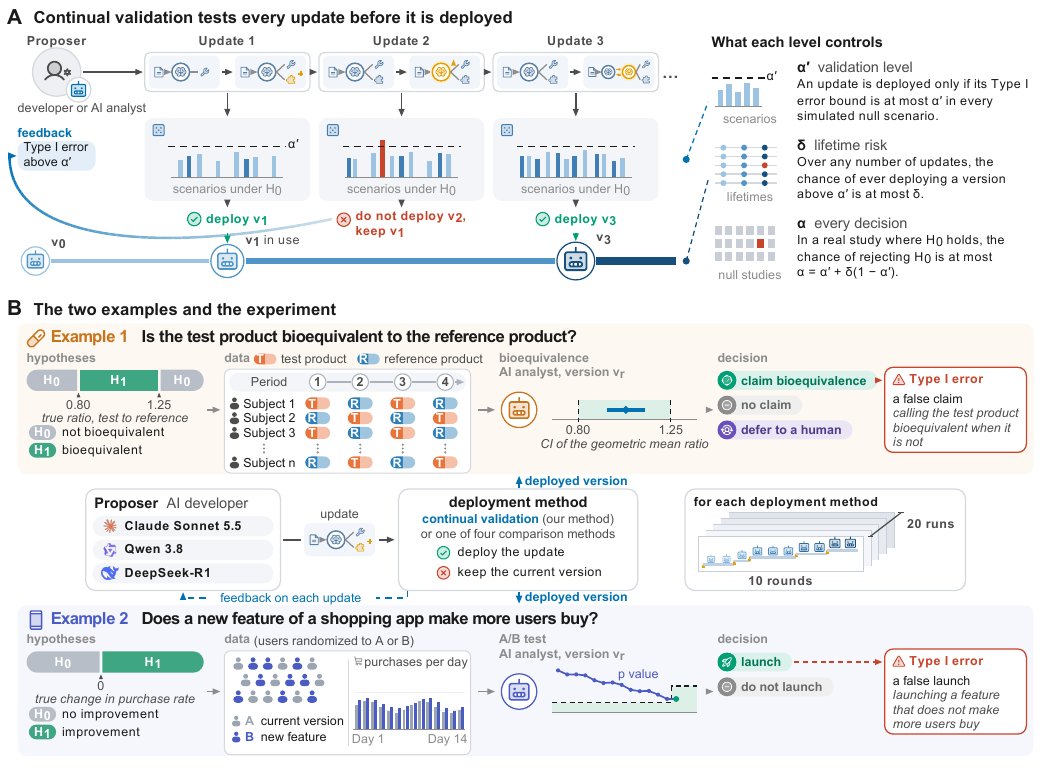}
\caption{Continual validation and the two examples. (\textit{A}) A human developer or the AI analyst itself proposes updates. The validator freezes the updated version and runs it alongside a reference analysis on datasets newly simulated under the null hypothesis. The update is deployed only when the upper confidence bound on its Type I error rate is at most the validation level $\alpha'$ across all simulated scenarios. The bound for update 2 exceeds $\alpha'$ in one scenario so it is rejected and $v_1$ stays in use. The proposer receives this result as feedback. Over any number of updates the chance of ever deploying a version above $\alpha'$ remains at most the lifetime risk $\delta$. In every real study where the data follow a simulated null scenario a true $H_0$ is rejected with probability at most $\alpha=\alpha'+\delta(1-\alpha')$. 
The examples use $\alpha=0.05$ and $\delta=0.005$ which gives $\alpha'=0.0452$. (\textit{B}) The two examples and the experiment setup. In Example 1 the agent decides whether a test product and a reference product are bioequivalent under the standard average bioequivalence framework. The true ratio is the ratio of their geometric mean exposures for a measure such as the peak concentration. The null hypothesis states that this ratio lies outside the regulatory acceptance limits of 0.80 and 1.25. The data come from replicate crossover studies where subjects receive the formulations in successive periods and the reference product is given at least twice. The agent claims bioequivalence and rejects $H_0$ when the computed confidence interval lies entirely within these limits. The agent can also defer a study to a human. In Example 2 users of a shopping application are randomized to the current version or to a new feature over 14 days. The true change represents the difference in the share of users who make a purchase. 
The agent uses the statistics engine of the
GrowthBook platform and makes launch decisions
according to the rule specified by each version.
A false launch constitutes a Type I error. In the experiment Claude Sonnet 5.5 proposes one update per round. The update either changes a setting or introduces a new decision rule. 
Each method is run 20 times for 10 rounds. Example 1 also tests the open-weight models Qwen 3.8 and DeepSeek-R1 as developers under three of the methods.}
\label{fig:concept}
\end{figure}

\section{Results}

\subsection{Why Common Validation Strategies Can Fail to Control Type I Error for Evolving AI Analysts}

Three commonly used statistical strategies can each fail to control the Type I error rate of an evolving AI analyst for different reasons.

First, validating each statistical tool separately does not validate an analyst that combines them.
In a bioequivalence example, two tests each control the Type I error rate near 0.05, but applying the second whenever the first fails to show bioequivalence raises it to \SoneBeWorst, as shown in
Fig.~\SIref{fig:si-concept}\textit{A} in the Appendix.
Theorem~\SIref{thm:nonclosure} in the Appendix
provides a general result.

Second, evaluating every proposed update at a fixed confidence
level \cite{thomas2019preventing} does not control the risk over many updates.
Even when each update is evaluated using new simulated data,
the probability of approving a version whose Type I error
rate is too high can accumulate.
If every proposed version has a Type I error rate above the target
and independently passes validation with probability $\eta$,
the probability that at least one such version passes within
$G$ updates is $1-(1-\eta)^G$, which approaches one as $G$
grows.
Fig.~\SIref{fig:si-concept}\textit{B} in the Appendix
shows this effect.

Third, repeatedly evaluating updates on the same
benchmark \cite{sun2026reuse} can lead to overfitting.
As updates are developed using the same test cases
or their feedback, they may perform well on the
benchmark without controlling Type I error on new data.
In an extreme example, an updated version can
reproduce the decisions of a valid reference analysis
on every stored test case but reject the null
hypothesis on all other datasets.
It would pass the benchmark evaluation while
having a Type I error rate of one on new data from
a continuous distribution\rev{, as Proposition~\SIref{prop:reuse} in the Appendix shows}.

\subsection{Continual Validation Maintains Type I Error Control at Modest Cost}

Continual validation, outlined in Fig.~\ref{fig:concept}\textit{A}
and Algorithm~\ref{alg:validation}, addresses these  problems
simultaneously.

Evaluating an evolving AI analyst requires controlling the Type I error rate across diverse data distributions and operational settings. For example, in a bioequivalence study, the null hypothesis holds when the true drug exposure ratio is at or beyond the acceptance limits of 0.80 and 1.25 for most drugs, and the error rate can vary with sample size, within-subject variability, and dropout rate. We therefore evaluate an analyst under different {null scenarios}, each specifying a data-generating process under which the null hypothesis is true. For instance, one scenario consists of 36 volunteers, 10 percent dropout, and a true exposure ratio at the upper acceptance limit, with volunteers resampled from a published study. For a version $v$ and a null scenario $\theta$, we denote its Type I error rate by $e(v,\theta)$.

Given the proposed version, the validator simulates $m_r$
datasets from each null scenario for update $r$, and runs the
updated version alongside a reference analyst version $v_0$,
where the reference analysis can be a standard statistical procedure run at
level $\alpha_0$, with a known upper bound $b_\theta$ on its
Type I error rate in scenario $\theta$.
Using these simulated datasets, the validator computes an exact
upper confidence bound $B_\theta$ on the updated version's
Type I error rate.
An update is deployed only if $B_\theta$ is at most the
validation level $\alpha'$ in every scenario.
Here, $\alpha$ is the target Type I error rate for each decision,
while $\delta$ bounds the lifetime probability of ever deploying
a version whose Type I error rate exceeds $\alpha'$.
We set $\alpha'=(\alpha-\delta)/(1-\delta)$ so that
$\alpha'+(1-\alpha')\delta=\alpha$.
For a target Type I error level such as $\alpha=0.05$ and lifetime
risk $\delta=0.005$, this gives $\alpha'=0.0452$.
For update $r$, the validator constructs each bound
$B_\theta$ at confidence level $1-\delta_r$, where
$\delta_r=6(\delta-\delta_0)/(\pi^2r^2)$.
These error budgets sum to $\delta-\delta_0$ over all
updates, with the remaining $\delta_0$ reserved for
the reference bounds $b_\theta$.
The initial set $\mathcal{S}_0$ of null scenarios is specified
in advance and expanded when an update uses information not varied
by the existing scenarios.

\begin{algorithm}[t]
\caption{Continual validation}
\label{alg:validation}
\begin{algorithmic}[1]
\Require Type I error level $\alpha$, lifetime risk $\delta$,
reference risk $\delta_0$, initial null scenarios $\mathcal{S}_0$,
reference analysis $v_0$ with bounds $b_\theta$
\State Set $\alpha'\gets(\alpha-\delta)/(1-\delta)$
and $\mathcal{S}\gets\mathcal{S}_0$
\State Deploy $v_0$, or a version that rejects only where $v_0$ does,
if $b_\theta\le\alpha'$ for every $\theta\in\mathcal{S}_0$.
Otherwise, start with a version that defers every study.
\For{each proposed update $r=1,2,\ldots$}
    \State Freeze the updated version $v_r$
    \State Add scenarios to $\mathcal{S}$ for any new inputs
    used by $v_r$, each with a reference bound $b_\theta$
    \State Set $\mathcal{S}_r\gets\mathcal{S}$
    \State Set $\delta_r\gets6(\delta-\delta_0)/(\pi^2r^2)$
    and fix the number $m_r$ of datasets per scenario
    \State Simulate $m_r$ new datasets from each
    $\theta\in\mathcal{S}_r$
    \State Run $v_r$ and $v_0$ on the same datasets
    \State For each $\theta\in\mathcal{S}_r$, count rejections
    $X_\theta$ of $v_r$, added rejections $N_\theta^{+}$,
    and removed rejections $N_\theta^{-}$
    \State Let $U(x;m,\eta)$ and $L(x;m,\eta)$ be the exact
    one-sided binomial bounds for $x$ events in $m$ trials
    at confidence level $1-\eta$ \cite{clopper1934use}
    \State $B^{\mathrm{pair}}_\theta\gets
    b_\theta+U(N_\theta^{+};m_r,\delta_r/4)
    -L(N_\theta^{-};m_r,\delta_r/4)$
    \State $B^{\mathrm{dir}}_\theta\gets
    U(X_\theta;m_r,\delta_r/2)$
    \State $B_\theta\gets
    \min(B^{\mathrm{pair}}_\theta,B^{\mathrm{dir}}_\theta)$
    \State If $B_\theta\le\alpha'$ for every
    $\theta\in\mathcal{S}_r$, deploy $v_r$.
    Otherwise, retain the current version.
    \State Report the deployment decision to the proposer
\EndFor
\end{algorithmic}
\end{algorithm}

The following theorem establishes that the lifetime risk
remains controlled, regardless of how many updates are
proposed or how they are selected.

\begin{theorem}[Lifetime risk control]
Suppose Algorithm~\ref{alg:validation} is followed and the reference bounds
$b_\theta$ hold simultaneously except with probability
at most $\delta_0$.
Then, regardless of how updates are proposed, we have
\begin{equation*}
\begin{aligned}
&\Pr\Bigl\{\exists\,r:\ v_r\text{ passes},\
\max_{\theta\in\mathcal{S}_r}e(v_r,\theta)>\alpha'\Bigr\}\\
&\le \delta_0+\sum_{r\ge1}\delta_r=\delta,
\end{aligned}
\end{equation*}
where $v_r$ is the version proposed in update $r$ and
$\mathcal{S}_r$ is its set of null scenarios.
\end{theorem}

Theorem~1 guarantees that the probability of ever deploying
a version whose Type I error rate exceeds $\alpha'$ in any
tested null scenario is at most $\delta$, regardless of
the number or selection of updates.
The following corollary establishes Type I error control
for every subsequent decision made by the AI analyst.

\begin{corollary}[Type I Error Control for Every Decision]
Let the analyst follow Algorithm~\ref{alg:validation}.
Consider any future study $s$ whose data follow a null scenario
covered by the version deployed before the study data arrive.
The probability of incorrectly rejecting 
$H_0$ satisfies
\[
\Pr_\theta\{\text{reject }H_0\text{ in study }s\}
\le \alpha'+(1-\alpha')\delta=\alpha.
\]
\end{corollary}

Corollary~1 establishes that each decision in a validated null scenario controls Type I
error at the target level $\alpha$, even after an
unlimited number of adaptive updates.
We call this property {lifetime validity}.

Controlling Type I error alone would be trivial if no updates
were ever approved. We therefore also show that continual
validation allows updates whose Type I error rate is sufficiently below $\alpha'$ to pass with high probability.

\begin{proposition}[Valid updates pass]
Consider update $r$ with $J_r=|\mathcal{S}_r|$ null scenarios.
Suppose the proposed version $v_r$ has Type I error rate
at most $\alpha'-\gamma$ in every scenario, where
$0<\gamma<\alpha'$ is the margin below the validation level.
For any $0<\zeta<1$, the update passes with probability
at least $1-\zeta$ if
\(
m_r\ge
\frac{8\alpha'}{\gamma^2}
\log\max\left\{
{2}/{\delta_r},{J_r}/{\zeta}
\right\}.
\)
\end{proposition}

This proposition shows that updates with Type I error rates
sufficiently below $\alpha'$ can be approved reliably.
For a fixed number of scenarios, the required number of
simulated datasets grows only logarithmically with the update
number.
For updates that never add rejections relative to the reference
analysis, the guarantee is even stronger.
If $b_\theta\le\alpha'-\gamma$ in every scenario,
such updates are guaranteed to pass validation whenever
$m_r\ge\log(4/\delta_r)/\gamma$.

Although continual validation requires repeated simulations,
its computational cost can be reduced in two ways.
First, it can avoid expensive language model calls.
When an analyst follows a fixed set of allowed analysis paths,
each ending in a decision computed by fixed code, the validator
can evaluate these decisions directly without running the
language model. The following proposition justifies this approach.

\begin{proposition}[Allowed paths]
Suppose every run of an AI analyst follows one of a fixed
set of allowed paths, each producing a decision that is
a function of the data.
Then its Type I error rate in every scenario is at most
that of the {any-path rule}, which rejects whenever
any allowed path rejects, regardless of which language
model is used or how it selects among the paths.
\end{proposition}

The choice of analysis path can vary substantially
across models, even for the same data and analysis
options\rev{, as Table~\SIref{tab:e1} in the Appendix shows}.
The any-path rule in Proposition~2 makes the guarantee
independent of which path the model selects.
Thus, validating the any-path rule avoids language model
calls, and subsequent model replacements, prompt edits,
or provider updates that preserve the allowed paths
require no additional simulation.

Second, for updates that require validation, comparing
their decisions with the reference analysis on the same
simulated datasets can substantially reduce the number
of simulations needed, as Fig.~\SIref{fig:si-concept}\textit{C}
in the Appendix shows.

\subsection{Two Evolving AI Analysts}
We studied continual validation in two AI analysts that kept changing, an agent that decides whether a test product is bioequivalent to its reference product and an analyst that decides whether to launch a new product feature. Fig.~\ref{fig:concept}\textit{B} outlines both. In each, the language model Claude Sonnet 5.5 acted as the developer of the analyst. In each of \NumRounds\ rounds, it proposed one update after reviewing a performance report and the validation feedback from its earlier proposals. 

In both experiments, while a language model plans the workflow, deterministic tools execute the computations. This architecture allows us to bypass expensive model calls and efficiently evaluate the agents at scale through their underlying rules\rev{, as section~\SIref{sec:exp-be-agent} of the Appendix details}. Each run started with a conservative baseline test. The developer was then tasked with safely increasing statistical power and automation.

We compared five deployment methods.
\begin{itemize}\setlength{\itemsep}{0pt}
\item No validation deploys every update without any check.
\item The static benchmark evaluates update on the same fixed set of simulated studies \cite{dwork2015reusable}.
\item The improvement test deploys an update when a lower confidence bound on its gain in power over the current version is positive \cite{sengupta2026selfevolving,sun2026reuse}.
\item Learn then Test evaluates a predefined set of configurations once before the first update \cite{angelopoulos2025learn}.
\item Continual validation runs with $\alpha=0.05$ and $\delta=0.005$.
\end{itemize}

Each deployment method ran \NumRuns\ times from the same initial version. 
The Type I error rate and power of every deployed version were estimated afterwards from independent simulations, on the validation scenarios and on off-grid scenarios that no method simulated. 
Power is reported relative to an oracle which, in its own runs with the same developer, deploys an update only if its estimated Type I error rate
is at most $\alpha'$ in every validation scenario and its estimated
power exceeds that of the current version.

\subsection*{Example 1, an Agent That Analyzes Bioequivalence Studies}

A test product is approved on several kinds of evidence, and one of them is a
bioequivalence study, which tests whether the test product delivers the active substance to the body like the reference product
\cite{ema2010bioequivalence,ich2024m13a,fda2026statistical,davit2009comparing}.
Bioequivalence assessment spans many designs, endpoints, and acceptance rules
\cite{chow2008design,hauschke2007bioequivalence,fda2026statistical}, and our
example covers one common part of it: average bioequivalence from
pharmacokinetic endpoints in replicate crossover studies. 
Healthy volunteers
take the test product and the reference product in successive periods, the reference at
least twice, and the exposure in each period is summarized by the area under
the concentration curve and the peak concentration. 
Evaluating bioequivalence can be formulated as performing the following hypothesis test:
$$
\begin{aligned}
H_0 &: \text{the two products are not bioequivalent} \\
H_1 &: \text{the two products are bioequivalent}
\end{aligned}
$$
In practice, the test is based on whether the true ratio of their geometric mean exposures lies outside the 0.80 to 1.25 interval ($H_0$), or within it ($H_1$).
An analyst claims bioequivalence when the 90 percent confidence interval for the ratio lies within this range, which corresponds to the two one-sided tests at level 0.05 \cite{schuirmann1987comparison}.
For highly variable drugs, whose reference product has a within-subject coefficient of variation above 30 percent, regulatory agencies widen these limits based on observed variability \cite{tothfalusi2009evaluation,davit2012implementation,karalis2012bioequivalence}.

A statistician performing
this analysis checks the protocol, prepares the data, fits
the model, runs diagnostics, and computes the interval,
with judgment at each step.
Our AI analyst is an agent
(\url{https://be.statrust.ai/}) that carries out these steps with
tools under the direction of language
models, writes reports, and claims bioequivalence, retains
the null hypothesis, or defers the study to a human.
Our experiment uses an early, simplified configuration
of this broader agent, which continues to evolve
through updates evaluated by continual validation.
To evaluate proposed updates to this workflow, we used realistic null scenarios resampled from published studies \cite{schutz2022replicatebe,schutz2020reference}. 
This exposed the analyst to complex clinical edge cases, including the 30 percent variability threshold where regulatory limits shift and errors easily inflate \cite{labes2016inflation,munoz2016consumer}.

Fig.~\ref{fig:results}\textit{A} to \textit{D} summarize the results. We highlight three key findings.

First, continual validation successfully controls the Type I error rate while enabling the analyst to safely evolve and achieve higher statistical power.
Continual validation deployed no version above $\alpha'$ in any of the 20 runs, and no deployed version had an estimated Type I error rate above \BEcvMax\ on the validation scenarios as shown in Fig.~\ref{fig:results}\textit{B} and \textit{C}.
The final version of the analyst achieved a power of 0.79 while maintaining a true Type I error rate of 0.043. Across all 200 rounds, the analyst proposed \BEcvDistinct\ distinct versions. Furthermore, \BEcvCodeShare\ percent of these updates introduced entirely new decision rules, creating the unique evolutionary trajectories shown in Fig.~\ref{fig:results}\textit{A}.

Second, comparison methods fail to balance capability evolution and Type I error control\rev{, as Fig.~\ref{fig:results}\textit{B} to \textit{D} show}. 
Both no validation and the improvement test failed to control the Type I error rate. They deployed versions above $\alpha=0.05$ in every run and reached error rates up to \BEnoCheckMax.
The static benchmark failed to detect when the Type I error rate exceeded $\alpha$ and deployed versions above this level in \BEfixedEver\ percent of runs. Learn then Test strictly controls the Type I error rate but only evaluates predefined options. Because it cannot evaluate novel decision rules proposed by the analyst, it achieved only \BEltPowerShare\ percent of the oracle power. Continual validation evaluates these novel rules while controlling the Type I error rate. It never exceeded the adjusted bound $\alpha'$, reached a maximum Type I error rate of \BEcvMax, and attained the highest power among all methods that controlled the Type I error rate.
Continual validation reached a mean final power of  \BEcvPowerShare\ percent of the oracle's and the highest among the methods that deployed no version above $\alpha$.

Third, continual validation is robust.
Although continual validation requires constructing null scenarios, its deployed versions maintained estimated Type I error rates at or below 0.05 on at least 666 of 690 off-grid scenarios (96.5\%).
The few exceedances, reaching at most \BEoffGridMax, occurred only under response-dependent dropout in two datasets not used for validation.
By drawing random scenarios, continual validation drives the analyst to evolve toward population-level performance rather than overfitting to a fixed testing grid.
In an additional experiment \rev{shown in Fig.~\ref{fig:robust}}, this randomization successfully rejected twelve updates that inflated the Type I error rate at off-grid sample sizes.
Under this setting, the method retained \BErandPowerShare\ percent of the oracle power.
Finally, these results hold across different language models.
Substituting Qwen or DeepSeek produced identical patterns.
The comparison methods failed to control the Type I error rate, while continual validation maintained strict control\rev{, as section~\SIref{sec:exp-be-results} of the Appendix details}.

\subsection*{Example 2, an Analyst of Product Experiments}

Online companies evaluate new product features through controlled experiments, commonly known as A/B tests \cite{kohavi2020trustworthy,larsen2024statistical}. Our example focuses on a two-arm experiment over 14 days. The primary metric is conversion, with revenue and engagement as secondary goals. Evaluating the new feature can be formulated as performing the following hypothesis test
$$
\begin{aligned}
H_0 &: \text{the new feature does not improve conversion} \\
H_1 &: \text{the new feature improves conversion}
\end{aligned}
$$
Rejecting $H_0$ means launching the feature, and launching an ineffective feature constitutes a Type I error. Common analytical choices, such as evaluating data every day or selectively choosing metrics after seeing the results, severely inflate this error \cite{johari2017peeking,simmons2011false}. We built the analyst on GrowthBook, an open-source experimentation platform \cite{growthbook2026}. 
The initial version uses daily purchase data to
compare conversion rates and bases its launch
decision on the day-14 $p$-value.
The developer was instructed to safely launch effective features more often and earlier. To evaluate proposed updates, we simulated null experiments featuring heavy-tailed revenue, correlated metrics, and multiple user segments, alongside dynamic scenarios targeting complex updates like daily data evaluations and new decision rules.

The results \rev{in Fig.~\ref{fig:results}\textit{E} to \textit{H}}
show that continual validation deployed no version
with a Type I error rate above $\alpha'=0.0452$ in any run.
Of the 30 rejected updates, 28 had Type I error rates
above $\alpha'$, while the remaining two had rates of
0.0435 and 0.0434, slightly below $\alpha'$.
In one run, the analyst proposed checking the data daily,
which raised the Type I error rate to 0.17.
Continual validation rejected this update.
In response to this feedback, the analyst autonomously
developed a new decision rule that allowed early launches
only for very small $p$-values and used an adjusted
threshold on day 14.
The final version achieved a power of 0.21
with a true Type I error rate of 0.043.
Across all runs, the analyst proposed \GBcvDistinct\ distinct versions, and \GBcvCodeShare\ percent of its updates were new decision rules. 

Comparison methods failed to consistently balance statistical power and Type I error control. 
No validation and the improvement test deployed versions above $\alpha=0.05$ in \GBallEver\ and \GBgateEver\ percent of runs, reaching error rates up to \GBnoCheckMax\ and \GBgateMax. 
The static benchmark deployed no version above $\alpha'$
and reached 98\% of the oracle's power in Example 2.
This was because all six updates that would have
exceeded $\alpha=0.05$ did so in scenarios covered
by its fixed benchmark.
In Example 1, by contrast, harmful updates exceeded
$\alpha$ in scenarios outside the fixed benchmark.
Learn then Test deployed no version above $\alpha$ but could not
deploy new decision rules, reaching only 80 percent of the
oracle power.
Continual validation safely evaluated these novel rules, never exceeded the adjusted bound $\alpha'$, and attained \GBcvPowerShare\ percent of the oracle power. Furthermore, On 116 off-grid scenarios, its deployed
versions had estimated Type I error rates of at most
0.044.

\begin{figure}[tp]
\centering
\IfFileExists{fig2.pdf}{\includegraphics[width=\textwidth]{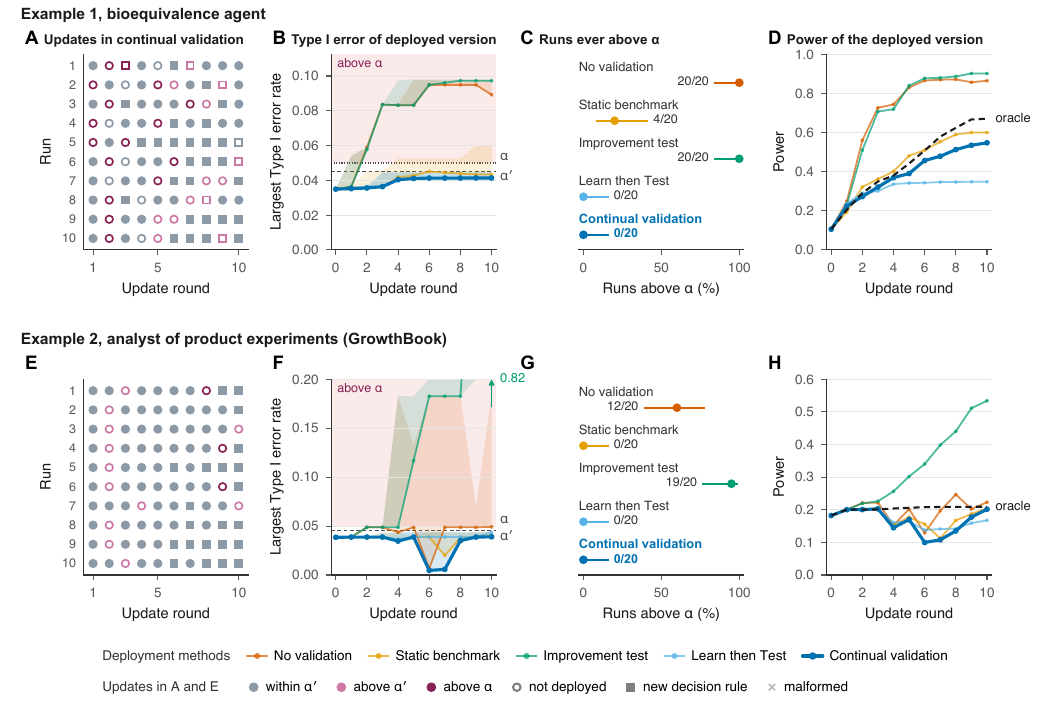}}{\figpending{Figure 2 shows the results of the two examples and is generated once they have been run.}{5cm}}
\caption{Continual validation and four comparison methods evaluated across two evolving AI analysts. Each method was tested in 20 independent runs of 10 rounds from the same initial version. In each round, Claude Sonnet 5.5 acted as the developer and proposed one update based on feedback from earlier rounds. (\textit{A} to \textit{D}) Bioequivalence analyst. (\textit{E} to \textit{H}) Product experiment analyst. (\textit{A} and \textit{E}) Updates proposed under continual validation in the first ten runs, with one row per run. Circles denote setting changes, squares denote new decision rules, and crosses denote malformed updates. Filled markers indicate deployed updates. Colors indicate the true Type I error rate of the proposed version, showing rates within $\alpha'$ (grey), above $\alpha'$ (pink), and above $\alpha$ (dark red). (\textit{B} and \textit{F}) Largest true Type I error rate of the deployed version over the full validation set after each round. Lines show the median over runs and bands extend to the worst run for no validation (vermilion), the static benchmark (orange), the improvement test (bluish green), Learn then Test (sky blue), and continual validation (blue). The dashed line marks the validation level $\alpha'=0.0452$, the dotted line marks the target level $\alpha=0.05$, and the shaded region lies above $\alpha$. In \textit{F}, the axis ends at 0.2 and an arrow indicates the largest rate in the worst run of the improvement test. (\textit{C} and \textit{G}) Share of runs that ever deployed a version above $\alpha$, showing Wilson 95 percent intervals and the raw count out of 20 runs. (\textit{D} and \textit{H}) Mean true statistical power of the deployed version after each round. The dashed line represents an oracle that deploys only valid updates that strictly increase power.}
\label{fig:results}
\end{figure}

\begin{figure}[tbp]
\centering
\includegraphics[width=0.6\linewidth]{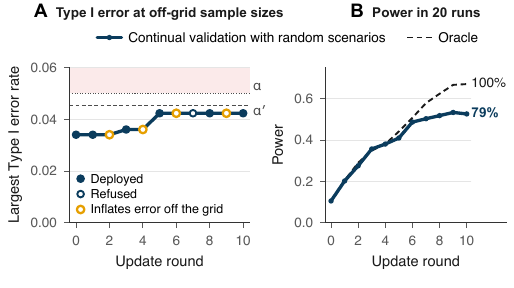}

\caption{Continual validation with random scenarios in the
bioequivalence analyst.
(\textit{A}) Largest true Type I error rate of the deployed
version at sample sizes of 30, 42, and 60, outside the
original validation grid of 24, 36, and 48.
Filled circles indicate deployed versions, while open circles
indicate rejected updates.
Orange circles highlight four rejected updates that would
increase Type I error outside the original grid while leaving
the analysis unchanged on the grid.
The dashed and dotted lines mark the validation level
$\alpha'=0.0452$ and the target level $\alpha=0.05$,
respectively. The shaded region indicates rates above $\alpha$.
(\textit{B}) Mean true power of the deployed version across
20 runs with Claude Sonnet 5.5 proposing updates.
The dashed line shows the oracle in Fig.~\ref{fig:results},
and the percentages indicate power relative to the oracle
after ten rounds.}
\label{fig:robust}
\end{figure}

\section{Discussion}

Type I error control need not end when a decision procedure starts to change. Continual validation extends this guarantee from a fixed procedure to an adaptively evolving, potentially unlimited sequence of AI analyst versions. Simulating datasets under the null hypothesis accounts for data sampling, while evaluating each frozen version accounts for the analyst's internal randomization. Freezing each update before simulating its validation data and controlling the lifetime probability of approving an invalid version prevents overfitting during the adaptive selection of the analyst. Together, these steps ensure that the probability of ever approving an invalid version remains bounded and that each deployed decision strictly controls the Type I error.

Several features distinguish continual validation from existing methods for evaluating and updating AI systems \cite{sengupta2026selfevolving,sun2026reuse,feng2021approval}. First, it directly controls the Type I error of statistical decisions, whereas methods that check only improvements in performance do not. A performance-based deployment criterion can still be added without affecting the statistical guarantee. Second, simulating new data after each update is frozen allows updates to be proposed adaptively, including by the analyst itself. Third, comparing every update with the same fixed reference analysis reduces simulation costs and avoids accumulating differences relative to successive versions. Finally, validating the allowed analysis paths can cover model replacements and prompt changes without additional simulation. Unlike Learn then Test, which calibrates a finite menu of configurations fixed in advance \cite{angelopoulos2025learn}, continual validation also accommodates previously unanticipated updates, including new decision rules written by the analyst. Regulators already use simulation over nuisance parameter grids for adaptive clinical trials \cite{fda2019adaptive} and provide for predetermined change control plans for AI-enabled medical software \cite{fda2024pccp}. Continual validation could provide a statistical framework for evaluating updates under such plans.

A limitation is that the theoretical guarantee strictly applies  to the null scenarios covered by validation. It does not automatically provide a mathematical bound for unobserved data generating patterns. However, our experiments demonstrate that continual validation remains highly robust in practice. The deployed versions successfully maintained Type I error control even on unseen off-grid scenarios. Furthermore, developers can easily expand the validation scenarios as the AI system gains new capabilities. 
Moreover, drawing randomly sampled scenarios naturally extends the theoretical safety guarantee to the broader data population. This prevents the AI analyst from overfitting to a fixed testing grid and ensures it safely evolves toward true population-level performance. The method is applicable wherever relevant null scenarios can be simulated and the AI system can be evaluated at sufficient scale.

Although we focus on Type I error, the framework requires only valid upper confidence bounds on the risk being controlled. It can therefore extend to other statistical error measures, such as the false discovery rate of a screening agent, when suitable confidence bounds are available. Overall, continual validation provides a practical deployment framework for AI systems. It allows AI analysts to continuously improve their statistical power and adapt to new tasks while maintaining rigorous safety guarantees throughout their working lives.

\section{Methods}
\subsection{Decisions, Versions, and Type I Error}
An AI analyst maps the data of a study, including its protocol, to one of three
decisions, to reject the null hypothesis, to retain it, or to defer to a human. A
deferral never counts as a rejection, so the guarantee concerns the decisions the
analyst makes on its own. 
A version $v$ is a complete specification of everything
that can change the decision, including the language model
and its sampling settings, the prompts, the allowed tools
and skills, the workflow, the code, and decision thresholds.
A deployed version decides each study using that
study's data and its own internal randomness,
without updating decision-relevant memory across studies.
Information from earlier studies may inform future
updates, which are validated before deployment.
For a null scenario
$\theta$, the Type I error rate $e(v,\theta)$ is the probability that $v$
rejects, over the data and the internal randomness. We assume that a real study whose true
state is a validation scenario produces data with the distribution from which the
validator simulates, and that a version cannot tell simulated studies from real
ones, which the Appendix states as Assumption~\SIref{ass:faithful}.

\subsection{Validation of an Update}
The bounds $U(x;m,\eta)$ and $L(x;m,\eta)$ of Algorithm~\ref{alg:validation} are
the exact one-sided Clopper and Pearson bounds \cite{clopper1934use}. Their levels $\delta_r/4$, $\delta_r/4$,
and $\delta_r/2$ add up to $\delta_r$, so by a union bound all three bounds of a
scenario hold together with probability at least $1-\delta_r$. Because an update
passes only if every scenario passes, the intersection-union argument
\cite{berger1982multiparameter} of Theorem~\SIref{thm:iut} in the Appendix shows that no division of
$\delta_r$ among the scenarios is needed. Random scenarios drawn after the
freeze enter a round with the analytic bound of the reference analysis where one
exists and $b_\theta=1$ otherwise, as described in
section~\SIref{sec:random} of the Appendix.

\subsection{Reference Analysis and Lifetime Risk}
The reference analysis is the standard analysis of the field, run at a reference
level $\alpha_0$ below $\alpha'$. For the two one-sided tests on Gaussian
scenarios with dropout that does not depend on the responses, its Type I error
rate is at most $\alpha_0$ \cite{schuirmann1987comparison}, which
Theorem~\SIref{thm:cont-ref} in the Appendix confirms for every design
used here, so that $b_\theta=\alpha_0$ at no simulation cost. On scenarios built from real studies or with non-Gaussian
data, the reference analysis receives a one-time validation that holds
simultaneously over these scenarios, including those that later updates may
add, and uses $\delta_0$ of the lifetime risk. Each round simulates all $m$
datasets in every scenario before the deployment decision, because optional stopping would
invalidate the exact bounds. We use $\alpha=0.05$, $\delta=0.005$ and
$\delta_0=0.001$, with $\alpha_0=0.030$ in Example 1 and $\alpha_0=0.040$ in
Example 2, and section~\SIref{sec:lifetime} of the Appendix
gives the proofs.

\subsection{Conditions for Deployment}
The validator must be able to run each version on simulated studies, either
through its allowed paths or in full. An analyst that learns from each study keeps
its memory frozen while a version is in force. The random seeds of a round are
drawn after its freeze. Each study is decided by one binding run, because
rerunning a stochastic version until it rejects is a different procedure. A
version validated in full pins its hosted model to a fixed snapshot, whereas a
version validated through its allowed paths is protected against silent model
updates by Proposition~2.

\section*{Acknowledgments}
The research of Xiao Wang was supported in part by the National Science Foundation under
Grant SES-2316428, by Anthropic through its AI for Science program, and by Lilly Purdue Innovation Institute.

\SIpart
\begin{center}{\LARGE\bfseries Appendix\par}\end{center}
\vskip6pt

\section{Setting and notation}
\label{sec:framework}

The theory concerns an AI analyst that decides one study after another while
its decision procedure keeps changing. We define the objects first, then the
Type I error rate that every guarantee controls, and finally the property of
lifetime validity. Table~\ref{tab:notation} collects the symbols used here and
in later sections.

\subsection{Studies and decisions}

The analyst reads the data $D$ of a study, which include its protocol, and
returns one of three decisions,
\[
  Y\in\{\text{reject},\ \text{retain},\ \text{defer}\}.
\]
Rejecting the null hypothesis is the consequential action that the evidence
must support. Such actions include reporting a discovery, launching a product
feature and declaring a test product bioequivalent to its reference product.
Retaining the null hypothesis declines the action, and a deferral hands the
study to a human expert. A deferral is never counted as a rejection, and the
guarantees below concern the rejections that the analyst makes on its own.

The data-generating distribution of a study depends on the true state of the
world. A scenario $\theta$ specifies this state completely, and $\Pr_\theta$
denotes probability when the data of a study are drawn from scenario $\theta$.
In a bioequivalence study a scenario fixes the true ratio of geometric mean
exposures, which the null hypothesis concerns, and the within-subject and
between-subject variability, which are nuisance parameters. For highly variable
drugs the within-reference variability also sets the limits that define the
null hypothesis. The scenario also fixes the sample size, the allocation to
sequences, the dropout pattern and the protocol the analyst reads.

Let $\Theta_0$ be the set of null scenarios, those in which the null
hypothesis holds and a rejection is wrong. The null hypothesis is composite,
since it holds in many scenarios. In bioequivalence the null scenarios are
those whose true ratio lies on or outside the regulatory limits for the true
variability. These limits are 80.00 and 125.00 percent when the true
within-reference coefficient of variation is at most 30 percent. Otherwise
they are the widened limits $\exp(\pm0.760\,\sigma_{wR})$ of the European
Medicines Agency (EMA) \cite{si:ema2010bioequivalence}, capped at 69.84 and
143.19 percent and evaluated at the true within-reference standard deviation
$\sigma_{wR}$ on the log scale, whereas the EMA's rule itself evaluates them at
the observed standard deviation $s_{wR}$ (section~\ref{sec:exp-be-scenarios}).
With the
limits of 80.00 and 125.00 percent the null hypothesis is $|\mu|\ge\Delta$ for
the log-ratio $\mu$ and $\Delta=\log 1.25$. Scenarios may be parametric, such
as Gaussian models of a crossover study. They may also be built from real
studies by resampling their subjects and moving the true effect to the
boundary of the null hypothesis or beyond it. The validator can draw studies
from every scenario it uses.

\subsection{Versions}

A version $v$ of the AI analyst is one fixed configuration of everything that
turns the data of a study into a decision. It includes the language model with
its identifier and sampling settings, the instructions, and the tools and
skills with their code. It also includes the statistical analyses and their
thresholds, the rules that send a study to an analysis or to a human, and the
computing environment. The analyst decides each study by running its current
version.

A version determines its decision as a fixed function of the data and of the
analyst's internal randomness,
\begin{equation}
  Y_v=y_v(D,\xi),
  \label{eq:decision}
\end{equation}
where $\xi$ collects every random choice the analyst makes. The randomness
$\xi$, which includes the sampling of a language model, is drawn independently
of $D$ from a distribution that the version fixes. Each version is therefore a
randomized test, a decision rule that may use randomness independent of the
data, and two runs of one version on one study have the same distribution of
decisions.

A version carries no memory from one study to the next. An analyst whose
context grows with every study it decides becomes a new version after each
study. Its memory must therefore stay frozen while a version is in force and
change only through updates. Each study is also decided by one binding run.
Rerunning a stochastic version until it rejects is a different procedure, and
$K$ runs of a version whose Type I error rate is $a$ can reject with
probability up to $1-(1-a)^K$ (the example after
Proposition~\ref{prop:routes}).

Every change to any component produces a new version, and the theory treats
all such versions alike. An edited prompt, a new skill, an upgraded tool, a
swapped model, a changed threshold and a revision written by the analyst
itself are all new versions. So is a hosted model that its provider updates
behind an unchanged name.

The validator of section~\ref{sec:certifying} runs a version exactly as it
runs in deployment, and simulated studies reach the version through the same
input as real studies. One assumption links the scenarios to the studies the
analyst decides in the field.

\begin{assumption}[Faithful simulation]
\label{ass:faithful}
A real study whose true state is $\theta$ produces data with the law
$\Pr_\theta$ from which the validator draws. A version's decision depends on a
study only through its data $D$ and its internal randomness $\xi$. In
particular, a version cannot tell a simulated study from a real one.
\end{assumption}

Under Assumption~\ref{ass:faithful} the rejection rate of a version in a
simulated scenario is its Type I error rate on real studies in that state. The
guarantees below are therefore statements about decisions in the field. The
assumption asks more of scenarios resampled from real studies than of
parametric ones, because resampling can leave traces that real studies lack,
such as repeated subjects. Section~\ref{sec:lifetime} shows how far the
guarantee reaches beyond the laws of the validation scenarios.

\subsection{Type I error rate}

For a version $v$ and a scenario $\theta\in\Theta_0$ the Type I error rate is
\begin{equation}
  e(v,\theta)=\Pr\nolimits_\theta\{Y_v=\text{reject}\},
  \label{eq:falsego}
\end{equation}
the probability over the data of the study and the internal randomness of the
analyst. For a set $\mathcal{S}\subseteq\Theta_0$ we write
$e(v,\mathcal{S})=\sup_{\theta\in\mathcal{S}}e(v,\theta)$, which is the size
of $v$ on $\mathcal{S}$. We call $v$ valid at level $a$ on $\mathcal{S}$ if
$e(v,\mathcal{S})\le a$. Power is the probability of a rejection in a scenario
where the null hypothesis is false. Power plays no part in the guarantees, and
most updates try to raise it.

Throughout, every decision is held to the level $\alpha=0.05$, and the
lifetime risk is $\delta=0.005$. The lifetime risk is the largest acceptable
probability that validation ever passes an update whose Type I error rate
exceeds the validation level in one of its validation scenarios. Each
validation uses the validation level $\alpha'=(\alpha-\delta)/(1-\delta)=0.0452$,
the reference analysis introduced below runs at the reference level
$\alpha_0$ of the example (0.030 in Example 1 and 0.040 in Example 2), and $\delta_0$ is the part of $\delta$ spent on the reference
bounds (section~\ref{sec:lifetime}).

\subsection{Scenario sets}

The validator tests versions on finite sets of scenarios, the validation
scenarios. They begin with the initial scenario set
$\mathcal{S}_0\subseteq\Theta_0$ of $J$ scenarios, which the validator fixes
before the analyst's working life begins. The initial scenario set varies the
factors that the validator expects to change decisions, such as variability,
sample size and dropout. It also contains null scenarios both at the boundary
of the null hypothesis and inside it. When an updated version reads a kind of
information that no scenario in the current set varies, the validator adds
scenarios that vary it (section~\ref{sec:extension}). Scenarios are added and
never removed, so the set $\mathcal{S}_r$ in force at update $r$ always
contains $\mathcal{S}_0$. The validator may also draw scenarios at random in
each round, after the updated version is frozen, and use them in that round
only. These extend the guarantee for each deployed version to the space
between the fixed scenarios, except a set of small probability
(Proposition~\ref{prop:random}). The proposer has no say in which scenarios
validation uses.

\subsection{An analyst that keeps changing}

Over its working life the analyst receives updates $r=1,2,\dots$, and update
$r$ offers a new version $v_r$, the updated version. An update may come from
anyone, and we call its source the proposer. The proposer may be the analyst
revising itself after a failure or an engineer editing its instructions. It
may also be an optimizer that maximizes power or an adversary searching for a
version that rejects too often. Each updated version passes or fails
validation. A version that passes may be deployed, and a deployed version
decides every study whose data arrive afterward. An update that fails changes
nothing but still counts, because updates are counted by attempt.

The validator holds a reference analysis $v_0$, a fixed decision procedure
with reference bounds $b_\theta$ on its Type I error rates. These bounds have
simultaneous coverage, $b_\theta\ge e(v_0,\theta)$ for every scenario ever
used, except on an event of probability at most $\delta_0$
(section~\ref{sec:lifetime}). When $b_\theta\le\alpha'$ for every
$\theta\in\mathcal{S}_0$, the analyst starts its life with the reference
analysis in force, or with a version that rejects only where the reference
analysis rejects (Corollary~\ref{cor:domination}). Otherwise it starts with the
version that defers every study. In bioequivalence the reference analysis is
the two one-sided tests procedure at the reference level $\alpha_0$
applied to the subjects with complete data (sections~\ref{sec:analytic}
and~\ref{sec:exp-be-settings}).

Studies $s=1,2,\dots$ arrive over time, each with a true scenario $\theta_s$
and a decision $Y_s$. We write $v_{(s)}$ for the version in force when the
data of study $s$ arrive.

\subsection{Lifetime validity}

\begin{definition}[Lifetime validity]
\label{def:lifetime}
An analyst that keeps changing is lifetime valid at level $\alpha$ on a set of
scenarios $\mathcal{S}\subseteq\Theta_0$ if every study $s$ that it ever
decides, and whose true scenario $\theta_s$ lies in $\mathcal{S}$ almost
surely, satisfies
\begin{equation}
  \Pr\{Y_s=\text{reject}\}\le\alpha ,
  \label{eq:lifetime-def}
\end{equation}
where the probability is taken over everything that happens up to and
including the decision. This history consists of the updates, the
validations, the earlier studies, the data of study $s$ and the internal
randomness of the analyst.
\end{definition}

Lifetime validity is a guarantee about each decision, under whichever version is
in force, after any number of changes chosen in any way.
Corollary~\ref{cor:every-decision} shows that continual validation delivers it
on every scenario set on which all versions in force are validated, and
Assumption~\ref{ass:faithful} carries it to the real studies in those states.

\begin{table}[h]
\centering\footnotesize
\caption{Notation.}
\label{tab:notation}
\renewcommand{\arraystretch}{1.12}
\begin{adjustbox}{max width=\textwidth}\begin{tabular}{@{}p{0.30\textwidth}p{0.64\textwidth}@{}}
\toprule
symbol & meaning \\
\midrule
$D$, $\xi$ & data of a study, internal randomness of the analyst \\
$Y_v\in\{\text{reject},\text{retain},\text{defer}\}$ & decision of version $v$ \\
$Y_s$, $\theta_s$ & decision on study $s$, true scenario of study $s$ \\
$\theta$, $\Theta_0$ & scenario, set of null scenarios \\
$e(v,\theta)$, $e(v,\mathcal{S})$ & Type I error rate in $\theta$, worst case over $\mathcal{S}$ \\
$a$ & a level of validity, the target of a single validation \\
$\mu$, $\Delta$ & log-ratio of the geometric means, equivalence limit $\log 1.25$ \\
$\sigma_{wR}$, $s_{wR}$ & true and observed within-reference standard deviation on the log scale \\
$\beta_i$, $f(i,j)$ & effect of sequence $i$ in the crossover, formulation given in period $j$ of sequence $i$ \\
$T_\mu$, $T_\kappa$ & $t$ statistics of the log-ratio $\mu$ and of the sequence effect $\kappa=\beta_1-\beta_2$ \\
$\mathcal{S}_0$, $\mathcal{S}_r$ & initial scenario set, scenario set of update $r$ \\
$J$, $J_r$ & number of scenarios in $\mathcal{S}_0$, number of scenarios in round $r$ \\
$\mathcal{S}(v)$ & scenario set on which $v$ passed validation or whose validity it inherits \\
$v_0$, $v_r$, $v_{(s)}$ & reference analysis, updated version of update $r$, version in force for study $s$ \\
$\mathcal{R}$, $v_{\mathcal{R}}$ & set of allowed paths, any-path rule \\
$\phi$, $\phi_\lambda$, $\Lambda$ & analysis, member of a routing family, index set of the family \\
$I$ & routing information that selects an analysis \\
$\alpha$, $\alpha'$, $\alpha_0$ & Type I error level of every decision, validation level, reference level \\
$\delta$, $\delta_0$, $\delta_r$ & lifetime risk, reference risk (the part of $\delta$ spent on the reference bounds), share of update $r$ \\
$\delta_{0,\theta}$, $\delta_0'$, $\delta_{0,k}$ & level of the reference bound of $\theta\in\mathcal{S}_0$, their sum, level of the $k$th added scenario \\
$b_\theta$ & reference bound, a bound on the Type I error rate of the reference analysis in $\theta$ \\
$E_0$, $\mathcal{E}$ & event that every reference bound holds, event that in addition no update above the validation level passes \\
$A_r$ & event that $v_r$ passes validation with $e(v_r,\mathcal{S}_r)>\alpha'$ \\
$m$, $m_r$, $m_0$ & simulated studies per scenario in one round, in round $r$, in the validation of the reference analysis \\
$X_\theta$, $N^{+}_\theta$, $N^{-}_\theta$ & rejections of the version under validation, rejections it adds to those of the reference analysis, rejections it removes \\
$q^{+}_\theta$, $q^{-}_\theta$ & probabilities of an added and of a removed rejection \\
$U(x;m,\eta)$, $L(x;m,\eta)$ & exact one-sided upper and lower binomial bounds at level $\eta$ \\
$\eta$, $\zeta$ & level of a bound, allowed probability that a valid update fails \\
$B^{\mathrm{pair}}_\theta$, $B^{\mathrm{dir}}_\theta$, $B_\theta$, $B_{r,\theta}$ & paired, direct and combined bounds of scenario $\theta$, combined bound of $\theta$ in round $r$ \\
$\gamma$ & margin below the validation level \\
$\mathrm{kl}(x,y)$, $\ell$, $\ell_r$ & relative entropy between Bernoulli distributions, logarithmic terms in the numbers of simulated studies of Theorem~\ref{thm:cost-rel} and Proposition~\ref{prop:power} \\
$\mathcal{G}_r$, $\mathcal{H}_s$ & information at the freeze of round $r$, at the arrival of study $s$ \\
$\Pi$, $\mathcal{Q}_r$, $K_r$, $\varepsilon_\Pi$ & law, set and number of the random scenarios of round $r$, $\Pi$-probability left uncovered \\
$V_r$, $\Gamma_r$ & set where $v_r$ exceeds $\alpha'$, event that $\Pi(V_r)>\varepsilon_\Pi$ and no random scenario of round $r$ lies in $V_r$ \\
$F_k$, $\Psi_k$ & number of false rejections among the first $k$ studies of Corollary~\ref{cor:frequency}, $F_k$ minus the sum of their conditional probabilities \\
\bottomrule
\end{tabular}\end{adjustbox}
\end{table}

\FloatBarrier

\section{Why common checks fail}
\label{sec:failures}

\begin{figure}[htbp]
\centering
\begin{minipage}[t]{0.32\textwidth}{\sffamily\bfseries A}\\[1pt]\includegraphics[width=\linewidth]{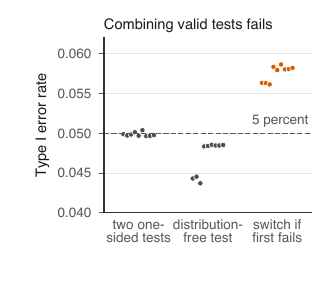}\end{minipage}\hfill
\begin{minipage}[t]{0.32\textwidth}{\sffamily\bfseries B}\\[1pt]\includegraphics[width=\linewidth]{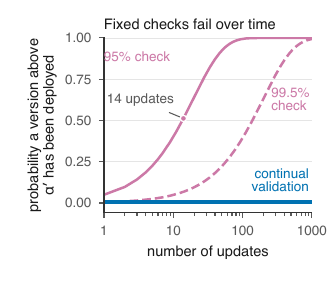}\end{minipage}\hfill
\begin{minipage}[t]{0.32\textwidth}{\sffamily\bfseries C}\\[1pt]\includegraphics[width=\linewidth]{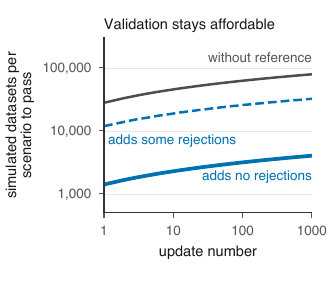}\end{minipage}
\caption{Why common checks fail and what validation costs. (\textit{A}) Combining
individually valid tests can inflate Type I error. Two bioequivalence tests each
control Type I error near 5 percent, but using the second whenever the first
fails exceeds this level. Dots show nine designs, each evaluated using
\SoneNrepB\ simulated studies. (\textit{B}) Repeated validation at a fixed
confidence level accumulates risk. The curves show the probability of ever
deploying a version above $\alpha'$ when every proposed version lies just above
this threshold. Fixed 95 percent and 99.5 percent checks accumulate risk, whereas
continual validation keeps it below the lifetime level $\delta$. (\textit{C})
Simulated datasets per scenario required for an update to pass validation as the
update number increases. Without reference comparison, a version with Type I
error rate $\alpha_0$ passes with probability 0.9 (gray). With paired comparison,
a version that adds no rejections passes with certainty (solid blue), while one
that adds rejections at rate $(\alpha'-\alpha_0)/2$ passes with probability 0.9
(dashed blue). Panels \textit{B} and \textit{C} use $\alpha=0.05$, $\delta=0.005$,
$\delta_0=0.001$, and $\alpha_0=0.040$, giving $\alpha'=0.0452$.}
\label{fig:si-concept}
\end{figure}

An analyst that changes freely can reject far more often than any check of a
single version suggests. Three mechanisms are at work, and each has an exact
statement below. Valid analyses do not compose, a check that reuses stored
simulations can be defeated, and a check that is valid for each version
separately still fails over a long life.

\subsection{Valid analyses do not compose}
\label{sec:composition}

An analysis is any decision rule $\phi$ with values reject and retain. It is valid
at level $\alpha$ if $e(\phi,\theta)\le\alpha$ for every $\theta\in\Theta_0$.

\begin{theorem}[Valid analyses do not compose]
\label{thm:nonclosure}
Fix $\alpha\in(0,1)$ and an integer $K\ge2$.
\begin{enumerate}[label=(\roman*),leftmargin=2.2em,itemsep=0pt]
\item Let $\phi_1,\dots,\phi_K$ be analyses valid at level $\alpha$. Let $v$
  be any version that rejects on a study only when at least one of
  $\phi_1,\dots,\phi_K$ rejects on it, however $v$ chooses among them. Then
  $e(v,\theta)\le\min(1,K\alpha)$ for every $\theta\in\Theta_0$.
\item The bound is attained. There are analyses $\phi_1,\dots,\phi_K$ with
  $e(\phi_k,\theta)=\alpha$ for every $k$ and every $\theta\in\Theta_0$, and a
  version $v$ of the kind in (i), such that $e(v,\theta)=\min(1,K\alpha)$ for
  every $\theta\in\Theta_0$.
\end{enumerate}
\end{theorem}

\begin{proof}
(i) The rejection event of $v$ lies in the union of the rejection events of the analyses, so
$e(v,\theta)\le\sum_k e(\phi_k,\theta)\le K\alpha$, and a probability is at
most one.

(ii) Let $W$ be uniform on $[0,1)$ and independent of the data, as part of the
internal randomness of the analyst. Set $c=\min(\alpha,1/K)$ and let $\phi_k$
reject when $(W-(k-1)c)\bmod 1<\alpha$. Each $(W-(k-1)c)\bmod 1$ is uniform on
$[0,1)$, so $e(\phi_k,\theta)=\alpha$ in every scenario. The rejection event of
$\phi_k$ is the arc of length $\alpha$ that starts at $(k-1)c$ on the circle
$[0,1)$. If $K\alpha\le1$, then $c=\alpha$, the $K$ arcs are disjoint and their
union has length $K\alpha$. If $K\alpha>1$, then $c=1/K<\alpha$, consecutive
arcs overlap and their union is the whole circle. The version that rejects
when any $\phi_k$ does therefore has $e(v,\theta)=\min(1,K\alpha)$.
\end{proof}

Choosing among valid analyses after seeing the data can multiply the Type I error
rate by the number of analyses, and nothing in the individual analyses reveals
it. The bound in (i) is Boole's inequality, and the construction in (ii)
ignores the data so that the arithmetic is exact. Real analyses behave the
same way whenever they reject on different studies. In bioequivalence, falling
back to a distribution-free interval whenever the two one-sided tests fail
raises the Type I error rate at the equivalence boundary to \SoneBeWorst\
(section~\ref{sec:exp-nonclosure}). The reference-scaled rule of the US Food
and Drug Administration (FDA) for highly variable drugs, which switches between
two valid analyses according to the observed variability, reaches about
\SthreeFactor\ times its nominal level when each criterion is judged where the
switch applies it (section~\ref{sec:exp-rsabe}).

\subsection{Selection through repeated updates}
\label{sec:selection}

A proposer that maximizes power keeps every change that adds rejections. If a
version $v'$ rejects on every run on which $v$ rejects, then
$e(v',\theta)\ge e(v,\theta)$ in every null scenario. In that case $v'$ also
has at least the power of $v$ in every scenario where the null hypothesis is
false. Under such changes power and the Type I error rate rise together.
Without a check, a lineage of $K$ changes that each add a valid analysis is a
version of the kind in Theorem~\ref{thm:nonclosure} and can reach
$\min(1,K\alpha)$. The changes that raise power fastest are often of exactly
this kind. Optional stopping, which tests at every interim look and stops at
the first rejection, launching when any of several metrics improves, and
widening the equivalence limits when the observed variability is high all add
rejections.

Checking each version is not enough, because the proposer selects among
checked versions. If the check reuses the same simulated studies, the proposer
can tune versions to them. If the check draws new studies but gives every
version the same confidence level, the probability that some version above the
validation level passes grows with every update. The next two
results make both failures exact, the first for reused simulations and the
second for a fixed confidence level per version.

\subsection{A reused set of simulations can be defeated}
\label{sec:reuse}

\begin{proposition}[A reused suite can be defeated]
\label{prop:reuse}
Let the reference analysis $v_0$ decide by a fixed function of the data. Let
$\mathcal{D}$ be a stored suite of simulated studies, $m$ per scenario, that
the proposer can read. Then there is a version $\tilde v$ that makes the
decision of the reference analysis on every study in $\mathcal{D}$. A
validation computed from $\mathcal{D}$ therefore gives $\tilde v$ exactly the
bounds it gives the reference analysis itself, and $\tilde v$ passes whenever
the reference analysis would pass in its place. Its Type I error rate satisfies
$e(\tilde v,\theta)\ge1-\Pr_\theta\{D\in\mathcal{D}\}$ in every scenario, and
$e(\tilde v,\theta)=1$ in every scenario whose data have a distribution
without atoms, such as any distribution with a density.
\end{proposition}

\begin{proof}
Let $\tilde v$ return the decision of the reference analysis on any study in
$\mathcal{D}$ and reject on any other study. On $\mathcal{D}$ the two
procedures agree study by study, so the rejection count of $\tilde v$ equals that of
the reference analysis, and $\tilde v$ adds and removes no rejections. These
are exactly the counts that the reference analysis produces against itself. In scenario
$\theta$ the version $\tilde v$ rejects whenever $D\notin\mathcal{D}$, which
gives the bound on $e(\tilde v,\theta)$, and a finite suite has probability
zero under a distribution without atoms.
\end{proof}

A suite of simulations that the proposer can see is a benchmark that it can
overfit. The construction is extreme, and its milder form is selection on
noise. A proposer that tunes a continuous threshold against
stored studies keeps the most permissive version whose estimate on those
studies stays below the validation level. The kept version is selected for a
low estimate, and its true Type I error rate can exceed that level. Conditions (a)
and (b) of Theorem~\ref{thm:lifetime} remove both effects for the simulated
studies, because the studies of every round are drawn after the version is
fixed and independently of everything the proposer has seen. The same two
effects can act on a fixed set of scenarios, and random scenarios drawn after
each freeze confine them to a set of small probability
(section~\ref{sec:random}).

\subsection{A fixed confidence level per version fails over a lifetime}
\label{sec:flat}

Call an update invalid if its Type I error rate exceeds the validation level
in some of its validation scenarios.

\begin{lemma}[A fixed confidence level per version]
\label{lem:flat}
Suppose that every update is validated on newly simulated studies at the
same confidence level $1-\eta$. In particular, given everything that precedes
the draws of round $r$, an invalid update passes with probability at most
$\eta$. Then for every proposer the probability that some invalid version
passes among the first $G$ updates is at most $1-(1-\eta)^G$. If every
update is invalid and passes with conditional probability exactly
$\eta$, the probability equals $1-(1-\eta)^G$, which tends to one as $G$
grows.
\end{lemma}

\begin{proof}
Let $A_r$ be the event that update $r$ is invalid and passes, and let
$\mathcal{G}_r$ be the information available before the draws of round $r$.
Let $N_r$ be the event that none of $A_1,\dots,A_r$ occurs. The event
$N_{r-1}$ is known before the draws of round $r$, so
\[
  \Pr(N_r)=\E\bigl[\mathbf{1}_{N_{r-1}}\Pr(A_r^{c}\mid\mathcal{G}_r)\bigr]
  \ge(1-\eta)\Pr(N_{r-1}),
\]
and induction gives $\Pr(N_G)\ge(1-\eta)^G$. Under the second hypothesis
$\Pr(A_r^{c}\mid\mathcal{G}_r)=1-\eta$ on $N_{r-1}$, and every inequality
becomes an equality.
\end{proof}

A confidence level per version is a guarantee per version, and over $G$
updates the risk compounds. In the worst case of the lemma, the familywise
error rate over the updates, the probability that some invalid update passes,
exceeds one half after \FoneGHalf\ updates at the conventional 95 percent
confidence, and after \FoneGHalfStrict{} updates at 99.5 percent.

Continual validation divides the lifetime risk $\delta$ among all updates, and
update $r$ receives a share $\delta_r$ with $\sum_r\delta_r\le\delta$. By a
union bound over the updates, the familywise error rate is then at most
$\delta$ however many updates arrive (Theorem~\ref{thm:lifetime}).

\FloatBarrier

\section{Validating one version}
\label{sec:certifying}

This section shows how a single frozen version is validated. The version $v$
is fixed, and the validator draws $m$ new simulated studies in each scenario
of a finite set $\mathcal{S}$. The question is whether $e(v,\theta)$ is at most
a target $a$ in every scenario of $\mathcal{S}$. In continual validation the
target is the validation level $a=\alpha'$, and Section~\ref{sec:lifetime}
chains such validations over a lifetime.

\subsection{Exact bounds in one scenario}

Let $X$ be binomial with $m$ trials and success probability $p$, and write
$\Pr_p$ for probabilities under this distribution. The exact one-sided upper bound of Clopper and Pearson
\cite{si:clopper1934use} at level $\eta$ is $U(x;m,\eta)=1$ for $x=m$, and
otherwise the $(1-\eta)$ quantile of the $\mathrm{Beta}(x+1,m-x)$
distribution. The lower bound is $L(x;m,\eta)=0$ for $x=0$, and otherwise the
$\eta$ quantile of the $\mathrm{Beta}(x,m-x+1)$ distribution.

\begin{lemma}[Exact binomial bounds]
\label{lem:cp}
For every $m\ge1$, every $\eta\in(0,1)$ and every $p\in[0,1]$, $\Pr_p\{U(X;m,\eta)<p\}\le\eta$ and
$\Pr_p\{L(X;m,\eta)>p\}\le\eta$. Moreover $\Pr_p\{U(X;m,\eta)\le p\}\le\eta$
for every $p<1$, and $\Pr_p\{L(X;m,\eta)\ge p\}\le\eta$ for every $p>0$.
\end{lemma}

\begin{proof}
For $x<m$ the probability $\Pr_q\{X\le x\}$ is continuous and strictly
decreasing in $q$, and it equals $\Pr\{\mathrm{Beta}(x+1,m-x)>q\}$. Hence
$U(x;m,\eta)$ is the value of $q$ at which $\Pr_q\{X\le x\}=\eta$, and $U$
increases in $x$. Fix $p<1$. If no $x$ has $U(x;m,\eta)\le p$, the event
$\{U(X;m,\eta)\le p\}$ is empty. Otherwise let $x_p$ be the largest such $x$.
Then $x_p<m$ because $U(m;m,\eta)=1$, and
\[
  \Pr\nolimits_p\{U(X;m,\eta)\le p\}=\Pr\nolimits_p\{X\le x_p\}
  \le\Pr\nolimits_{U(x_p;m,\eta)}\{X\le x_p\}=\eta ,
\]
because $p\ge U(x_p;m,\eta)$ and $\Pr_q\{X\le x_p\}$ decreases in $q$. The
strict statement follows for $p<1$, and it holds for $p=1$ because then $X=m$
and $U(m;m,\eta)=1$. The statements for $L$ follow from
$L(x;m,\eta)=1-U(m-x;m,\eta)$ applied to $m-X$, which is binomial with success
probability $1-p$.
\end{proof}

The bounds are exact for every number of simulated studies, with no
asymptotic approximation. The second statement adds that the upper bound
falls at or below the true rate with probability at most $\eta$. This property
implies that a version whose Type I error rate equals the target cannot be
validated (Proposition~\ref{prop:slack}).

\subsection{Every scenario must pass}

A validation of $v$ on $\mathcal{S}$ computes one bound $B_\theta$ per
scenario from the simulated studies. The version $v$ passes validation if
$B_\theta\le a$ for every scenario $\theta\in\mathcal{S}$. The bound
$B_\theta$ is valid at level $\eta$ for $v$ if
$\Pr\{e(v,\theta)>B_\theta\}\le\eta$.

\begin{theorem}[Every scenario must pass]
\label{thm:iut}
Let $v$ be a fixed version and $\mathcal{S}$ a finite set of scenarios. For
each $\theta\in\mathcal{S}$, let $B_\theta$ be a bound that is valid at level
$\eta$ for $v$. Then
\[
  \Pr\bigl\{v\text{ passes validation and }e(v,\mathcal{S})>a\bigr\}\le\eta ,
\]
and $\max_{\theta\in\mathcal{S}}B_\theta$ is an upper confidence bound for
$e(v,\mathcal{S})$ at level $\eta$. The bounds may be computed from shared or
from separate simulated studies.
\end{theorem}

\begin{proof}
Fix an enumeration of $\mathcal{S}$, and let $\theta^*$ be the first scenario
in it at which $e(v,\cdot)$ is largest. Because $v$ and $\mathcal{S}$ are
fixed, $\theta^*$ is fixed as well. If
$v$ passes validation and $e(v,\mathcal{S})>a$, then $B_{\theta^*}\le a<e(v,\theta^*)$.
If $e(v,\mathcal{S})>\max_\theta B_\theta$, then $e(v,\theta^*)>B_{\theta^*}$.
Both events lie in $\{e(v,\theta^*)>B_{\theta^*}\}$, whose probability is at
most $\eta$.
\end{proof}

A version above the target exceeds it in some particular scenario, and passing
requires the bound of that scenario to come out at or below the target. That
one event has probability at most $\eta$, so the level does not have to
be divided among the scenarios of $\mathcal{S}$. The rule that every scenario
must pass is an intersection-union test, which rejects a union of null
hypotheses, here the hypotheses $e(v,\theta)>a$ for $\theta\in\mathcal{S}$,
only when it rejects each of them. This is also why the two one-sided tests
procedure needs no correction for running two tests
\cite{si:berger1982multiparameter}. Bounds with simultaneous coverage, holding in
all scenarios at once, would require dividing the level among the scenarios.
Such a division replaces $\log(1/\eta)$ by about $\log(|\mathcal{S}|/\eta)$ in
the simulation effort, and valid versions then fail more often with no gain in
validity.

\subsection{A version at the target cannot be validated}

\begin{proposition}[Validation fails at the target]
\label{prop:slack}
Let $a<1$, and let $B_\theta=U(X_\theta;m,\eta)$ be the direct bound computed
from the number $X_\theta$ of rejections of $v$ among $m$ simulated studies of
scenario $\theta$. If $e(v,\theta)\ge a$, then $v$ passes scenario $\theta$ with
probability at most $\eta$, for every $m$.
\end{proposition}

\begin{proof}
The count $X_\theta$ is binomial with success probability $p=e(v,\theta)$. If
$p<1$, passing requires $U(X_\theta;m,\eta)\le a\le p$, which has probability at most
$\eta$ by Lemma~\ref{lem:cp}. If $p=1$, then $X_\theta=m$ and
$U(X_\theta;m,\eta)=1>a$, so $v$ never passes.
\end{proof}

Simulation can confirm that a Type I error rate lies below the target by some
margin, and never that it equals the target. A version whose Type I
error rate equals the target is valid, yet it fails validation with
probability at least $1-\eta$ however many studies are simulated. The same
holds for the paired and combined bounds of the next subsection
(Theorem~\ref{thm:relative}).

This limit also governs the choice of the reference level $\alpha_0$. The two
one-sided tests procedure at level $a$ has a Type I error rate within a
negligible distance of $a$ at the equivalence boundary when a study is precise
(Theorem~\ref{thm:cont-ref}). A reference analysis run at the validation level
would thus fail validation with probability close to $1-\eta$ in scenarios
where its bound comes from simulation. The paired bound also needs a margin
below the target. Since $U(0;m,\eta)>0$, even a version that copies every decision
of the reference analysis passes the paired bound only if $b_\theta<a$. The
reference analysis therefore runs at a reference level $\alpha_0$
below the validation level $\alpha'=0.0452$.

The margin $\gamma=\alpha'-\alpha_0$ is available to every future version and
does not shrink, because every version is compared with the same reference
analysis and not with its predecessor, so increments do not accumulate across
versions.

\subsection{Paired and direct bounds}
\label{sec:routes}

The validator runs the version $v$ and the reference analysis $v_0$ on the
same $m$ simulated studies of scenario $\theta$, with independent internal
randomness for every run. It counts the studies on which only $v$ rejects,
$N^{+}_\theta$, and those on which only the reference analysis rejects,
$N^{-}_\theta$, together with the rejections of $v$, $X_\theta$. Let
$q^{+}_\theta$ and $q^{-}_\theta$ be the probabilities of an added and of a
removed rejection in scenario $\theta$.
Then
\begin{equation}
  e(v,\theta)=e(v_0,\theta)+q^{+}_\theta-q^{-}_\theta
  \label{eq:paired-identity}
\end{equation}
exactly. Given a reference bound $b_\theta\ge e(v_0,\theta)$ and levels
$\eta_1,\eta_2,\eta_3$, the paired, direct and combined bounds are
\begin{equation}
\begin{aligned}
  B^{\mathrm{pair}}_\theta&=b_\theta+U(N^{+}_\theta;m,\eta_1)-L(N^{-}_\theta;m,\eta_2),\\
  B^{\mathrm{dir}}_\theta&=U(X_\theta;m,\eta_3),\qquad
  B_\theta=\min\bigl(B^{\mathrm{pair}}_\theta,B^{\mathrm{dir}}_\theta\bigr).
\end{aligned}
\label{eq:bounds}
\end{equation}

\begin{theorem}[Paired and direct bounds]
\label{thm:relative}
Fix the version $v$, the scenario $\theta$ and a reference bound
$b_\theta\ge e(v_0,\theta)$, and let the counts come from $m$ independent
simulated studies of $\theta$.
\begin{enumerate}[label=(\roman*),leftmargin=2.2em,itemsep=0pt]
\item $\Pr\{e(v,\theta)>B^{\mathrm{pair}}_\theta\}\le\eta_1+\eta_2$,
  $\Pr\{e(v,\theta)>B^{\mathrm{dir}}_\theta\}\le\eta_3$, and
  $\Pr\{e(v,\theta)>B_\theta\}\le\eta_1+\eta_2+\eta_3$.
\item If $e(v,\theta)<1$, then also
  $\Pr\{B_\theta\le e(v,\theta)\}\le\eta_1+\eta_2+\eta_3$.
\end{enumerate}
With $\eta_1=\eta_2=\eta/4$ and $\eta_3=\eta/2$ the combined bound is valid at
level $\eta$. Suppose that every $b_\theta$ bounds the Type I error rate of
the reference analysis. By the intersection-union argument of
Theorem~\ref{thm:iut}, a version with $e(v,\mathcal{S})>a$ then has
$B_\theta\le a$ in every scenario, and so passes validation, with probability
at most $\eta$.
\end{theorem}

\begin{proof}
The $m$ simulated studies and the runs on them are independent and identically
distributed, so $N^{+}_\theta$, $N^{-}_\theta$ and $X_\theta$ are binomial with
success probabilities $q^{+}_\theta$, $q^{-}_\theta$ and $e(v,\theta)$. Their
joint law plays no role.

(i) By Lemma~\ref{lem:cp}, the event $U(N^{+}_\theta;m,\eta_1)<q^{+}_\theta$
has probability at most $\eta_1$, and the event
$L(N^{-}_\theta;m,\eta_2)>q^{-}_\theta$ has probability at most $\eta_2$.
Outside these two events, \eqref{eq:paired-identity} and
$b_\theta\ge e(v_0,\theta)$ give
$e(v,\theta)\le b_\theta+q^{+}_\theta-q^{-}_\theta\le B^{\mathrm{pair}}_\theta$.
The statement for the direct bound is Lemma~\ref{lem:cp}, and the combined
statement follows by the union bound.

(ii) If $B^{\mathrm{pair}}_\theta\le e(v,\theta)$, then
$U(N^{+}_\theta;m,\eta_1)-L(N^{-}_\theta;m,\eta_2)\le q^{+}_\theta-q^{-}_\theta$,
so either $U(N^{+}_\theta;m,\eta_1)\le q^{+}_\theta$ or
$L(N^{-}_\theta;m,\eta_2)>q^{-}_\theta$. Since
$q^{+}_\theta\le e(v,\theta)<1$, Lemma~\ref{lem:cp} bounds the two
probabilities by $\eta_1$ and $\eta_2$. The event
$B^{\mathrm{dir}}_\theta\le e(v,\theta)$ has probability at most $\eta_3$ by the
same lemma. The event $B_\theta\le e(v,\theta)$ lies in the union of the events
for the paired and direct bounds.
\end{proof}

The paired bound compares the version with the reference analysis on the same
simulated studies. It therefore needs to resolve only how often the version
adds a rejection that the reference analysis would not make. That rate is zero
for a version that never adds rejections and small for a modest edit, whatever
the Type I error rate of the reference analysis. The direct bound ignores the
reference analysis, and it is the better bound for a version that changes many
of the decisions of the reference analysis. Taking the smaller bound in each
scenario costs a quarter of the level for each count in the paired bound and
half of it for the direct bound. The validator need not know in advance which
bound will be smaller.

With Theorem~\ref{thm:iut}, part (i) bounds the probability of the joint event
that a version passes and lies above the target, so no conditioning on passing
is involved.

\subsection{What validation costs}

Write $\mathrm{kl}(x,y)$ for the relative entropy between Bernoulli
distributions with success probabilities $x$ and $y$. Let $\zeta\in(0,1)$ be an
acceptable probability that a valid version fails validation.

\begin{theorem}[Cost of the two bounds]
\label{thm:cost-rel}
Fix a scenario $\theta$, a target $a\in(0,1)$, a margin $\gamma\in(0,a)$ and a
level $\eta\in(0,1/4)$, and write $\ell=\max\{\log(1/\eta),\log(1/\zeta)\}$.
\begin{enumerate}[label=(\roman*),leftmargin=2.2em,itemsep=0pt]
\item A version that never adds rejections. If $q^{+}_\theta=0$ and
  $b_\theta\le a-\gamma$, then $B^{\mathrm{pair}}_\theta\le a$ with
  probability one as soon as $m\ge\log(1/\eta)/\log\{1/(1-\gamma)\}$, and
  $m\ge\log(1/\eta)/\gamma$ suffices.
\item A version that adds few rejections. If $q^{+}_\theta\le c\gamma$ with
  $0\le c<1$ and $b_\theta\le a-\gamma$, then $B^{\mathrm{pair}}_\theta\le a$
  with probability at least $1-\zeta$ whenever $m\ge8\ell/\{(1-c)^2\gamma\}$.
\item Validation from scratch. If $e(v,\theta)=a-\gamma$, then
  $B^{\mathrm{dir}}_\theta\le a$ with probability at least $1-\zeta$ whenever
  $m\ge8a\ell/\gamma^2$. Conversely, consider a rule that decides from the
  rejections of a version on $m$ simulated studies of $\theta$. Suppose that
  it passes every version with $e(v,\theta)>a$ with probability at most
  $\eta$. Then it passes a version with $e(v,\theta)=a-\gamma$ with probability
  below one half unless $m\ge a(1-a)\log\{1/(4\eta)\}/(2\gamma^2)$.
\end{enumerate}
Here $\eta$ is the level of the bound in question, that is, $\eta_1$ in (i) and
(ii), and $\eta_3$ in (iii).
\end{theorem}

\begin{proof}
The proof uses three standard facts about binomial tails and relative entropy.
First, for $X$ binomial with $m$ trials and
success probability $y$, Chernoff's bound gives
$\Pr\{X\le mx\}\le\exp\{-m\,\mathrm{kl}(x,y)\}$ for $x\le y$ and
$\Pr\{X\ge mx\}\le\exp\{-m\,\mathrm{kl}(x,y)\}$ for $x\ge y$. Second,
$\mathrm{kl}(x,y)\ge(x-y)^2/\{2\max(x,y)\}$, because the second derivative of
$\mathrm{kl}(\cdot,y)$ at $t$ is $1/\{t(1-t)\}\ge1/\max(x,y)$ for $t$ between
$x$ and $y$. Third, $\mathrm{kl}(x,y)\le(x-y)^2/\{y(1-y)\}$, the chi-square
bound. By the proof of Lemma~\ref{lem:cp}, $U(x;m,\eta)\le g$ holds for $x<m$
exactly when $\Pr_g\{X\le x\}\le\eta$, so the first two facts give
\begin{equation}
  U(x;m,\eta)\le g\quad\text{whenever}\quad x/m\le g\ \text{ and }\
  (g-x/m)^2\ge2g\log(1/\eta)/m .
  \label{eq:ucb-suff}
\end{equation}

(i) Here $N^{+}_\theta=0$ almost surely, $U(0;m,\eta)=1-\eta^{1/m}$ because
$\Pr_q\{X=0\}=(1-q)^m$, and $L\ge0$. Hence
$B^{\mathrm{pair}}_\theta\le a-\gamma+1-\eta^{1/m}$, which is at most $a$ when
$\eta^{1/m}\ge1-\gamma$, that is when $m\log\{1/(1-\gamma)\}\ge\log(1/\eta)$.
Finally, the inequality $\log\{1/(1-\gamma)\}\ge\gamma$ gives the simpler
sufficient condition in (i).

(ii) Let $y=(1+c)\gamma/2$. On the event $N^{+}_\theta\le my$ we have
$\gamma-N^{+}_\theta/m\ge(1-c)\gamma/2$, and \eqref{eq:ucb-suff} with $g=\gamma$
gives $U(N^{+}_\theta;m,\eta)\le\gamma$ once
$m\ge8\log(1/\eta)/\{(1-c)^2\gamma\}$. Then
$B^{\mathrm{pair}}_\theta\le b_\theta+\gamma\le a$. The count
$N^{+}_\theta$ is binomial with success probability at most $c\gamma<y$. The
first two facts therefore give
$\Pr\{N^{+}_\theta>my\}\le\exp\{-m(y-c\gamma)^2/(2y)\}
=\exp\{-m(1-c)^2\gamma/(4(1+c))\}$, which is at most $\zeta$ once
$m\ge8\log(1/\zeta)/\{(1-c)^2\gamma\}$.

(iii) On the event $X_\theta\le m(a-\gamma/2)$, \eqref{eq:ucb-suff} with $g=a$
gives $B^{\mathrm{dir}}_\theta\le a$ once $m\ge8a\log(1/\eta)/\gamma^2$.
Moreover, $\Pr\{X_\theta>m(a-\gamma/2)\}\le\exp\{-m(\gamma/2)^2/(2a)\}$, which
is at most $\zeta$ once $m\ge8a\log(1/\zeta)/\gamma^2$. For the converse, let
$A$ be the event that the rule passes the version. Write $\Pr_p(A)$ for its
probability when the rejections are independent with probability $p$. The
probability $\Pr_p(A)$ is a polynomial in $p$, hence continuous, so
$\Pr_a(A)\le\eta$. The data-processing inequality for
relative entropy gives
\[
  m\,\mathrm{kl}(a-\gamma,a)\ge\mathrm{kl}\bigl(\Pr\nolimits_{a-\gamma}(A),\Pr\nolimits_a(A)\bigr).
\]
Since $\mathrm{kl}(x,y)$ increases in $x$ for $x\ge y$ and decreases in $y$ for
$y\le x$, the right side is at least
$\mathrm{kl}(1/2,\eta)=\tfrac12\log[1/\{4\eta(1-\eta)\}]\ge\tfrac12\log\{1/(4\eta)\}$
whenever $\Pr_{a-\gamma}(A)\ge1/2$. The chi-square bound gives
$\mathrm{kl}(a-\gamma,a)\le\gamma^2/\{a(1-a)\}$, and the two inequalities
together give the stated lower bound on $m$.
\end{proof}

Validation against the reference analysis needs a number of simulated studies
of order $1/\gamma$ for a version that never adds rejections or adds them at a
rate below the margin. Validation from scratch needs order $a/\gamma^2$ for
a version whose Type I error rate lies $\gamma$ below the target, and no rule
that uses only the decisions of the version can do better. The ratio is of
order $a/\gamma$, about 9 at the margin $\gamma=\alpha'-\alpha_0=0.0052$ of
Example 2, and the constants of the exact bounds widen it.
Table~\ref{tab:costs} gives the exact numbers of simulated studies at the
first update. The two numbers differ by a factor of \FoneCostFactorA{} at the
margin of Example 2 and of \FoneCostFactorC{} at a margin of $0.001$. The
combined bound costs no more than the better bound in every scenario.

The saving shrinks as a version changes more decisions of the reference
analysis, and Table~\ref{tab:costs} gives the profile. A version whose
rejections contain those of the reference analysis and add rejections at half
the margin needs \FoneNestedPairedB{} simulated studies against
\FoneNestedScratchB{} from scratch, a factor of \FoneNestedFactorB. A version
that keeps the Type I error rate of the reference analysis but adds and removes
rejections at equal rates gains less, because the paired bound carries the
uncertainty of both counts. Its factor, a Monte Carlo approximation accurate to
about 2 percent, falls from \FoneDiscFactorA{} to \FoneDiscFactorC{} as the
rate of added and of removed rejections rises from $0.002$ to $0.012$ each, and the combined bound
then rests on the direct bound.

\begin{table}[!htbp]
\centering
\caption{Simulated studies per scenario needed at the first update.}
\label{tab:costs}
\begin{adjustbox}{max width=\textwidth}\begin{tabular}{lccc}
\toprule
\multicolumn{4}{l}{A. Versions that add no rejection, by margin $\gamma$} \\
Margin $\gamma$ & $0.0052$ & $0.002$ & $0.001$ \\
\midrule
Against the reference analysis & \FoneCostPairedA & \FoneCostPairedB & \FoneCostPairedC \\
From scratch & \FoneCostScratchA & \FoneCostScratchB & \FoneCostScratchC \\
Factor & \FoneCostFactorA & \FoneCostFactorB & \FoneCostFactorC \\
\midrule
\multicolumn{4}{l}{B. Versions that only add rejections, at rate $c\gamma$} \\
Fraction $c$ & $0.25$ & $0.5$ & $0.75$ \\
\midrule
Against the reference analysis & \FoneNestedPairedA & \FoneNestedPairedB & \FoneNestedPairedC \\
From scratch & \FoneNestedScratchA & \FoneNestedScratchB & \FoneNestedScratchC \\
Factor & \FoneNestedFactorA & \FoneNestedFactorB & \FoneNestedFactorC \\
\midrule
\multicolumn{4}{l}{C. Versions that add and remove rejections, each at rate $q$} \\
Rate $q$ & $0.002$ & $0.006$ & $0.012$ \\
\midrule
Combined bound & \FoneDiscCombinedA & \FoneDiscCombinedB & \FoneDiscCombinedC \\
From scratch & \FoneCostScratchA & \FoneCostScratchA & \FoneCostScratchA \\
Factor & \FoneDiscFactorA & \FoneDiscFactorB & \FoneDiscFactorC \\
\bottomrule
\end{tabular}\end{adjustbox}

\addtabletext{All entries are computed at $\alpha=0.05$, $\delta=0.005$
and $\delta_0=0.001$, so that $\alpha'=0.0452$ and the first share is
$\delta_1=0.0024$. Validation from scratch means the direct bound at the
level $\delta_1/2$ it carries inside the combined bound, and its row gives the
smallest $m$ at which $U(X_\theta;m,\delta_1/2)\le\alpha'$ with probability
0.9. In A the reference analysis and the version both have Type I error rate
$\alpha'-\gamma$, with $b_\theta=\alpha'-\gamma$, and the first row gives
the smallest $m$ with $U(0;m,\delta_1/4)\le\gamma$, at which the version passes
with certainty. In B and C the reference analysis has Type I error rate
$\alpha_0$, with $b_\theta=\alpha_0$ and the margin $\gamma=\alpha'-\alpha_0$
of Example 2. In B the version rejects whenever the reference analysis
does, so its Type I error rate is $\alpha_0+c\gamma$, and the first row gives the smallest $m$ at
which the paired bound passes with probability 0.9. In C the version keeps the
Type I error rate $\alpha_0$, and the first row gives the smallest $m$ at which
the combined bound passes with probability 0.9, estimated from 4{,}000
simulated rounds per $m$ on a grid of one percent. The values of C are
therefore Monte Carlo approximations, accurate to about 2 percent. All numbers
grow with the update index only through $\log(1/\delta_r)$.\par}
\end{table}

\subsection{A reference analysis validated by proof}
\label{sec:analytic}

In the two-period, two-sequence crossover, subject $k$ of sequence
$i\in\{1,2\}$ receives the test formulation T and the reference R in different
periods, with T first in sequence 1. On the log scale its response in period
$j$ is
\begin{equation}
  y_{ikj}=\mu_0+\beta_i+p_j+\tau_{f(i,j)}+u_{ik}+\varepsilon_{ikj},
  \label{eq:crossover}
\end{equation}
with $f(1,1)=f(2,2)=\mathrm{T}$ and $f(1,2)=f(2,1)=\mathrm{R}$. Here the
$\beta_i$ are sequence effects, the $u_{ik}$ are subject effects, and the
errors $\varepsilon_{ikj}$ are independent and Gaussian with mean zero and
variance $\sigma_W^2$. The log-ratio is $\mu=\tau_{\mathrm{T}}-\tau_{\mathrm{R}}$,
and the sequence effect is $\kappa=\beta_1-\beta_2$.

The reference analysis $v_0$ is the two one-sided tests procedure at level
$\alpha_0$ on the completers, the subjects who complete both periods. Let $n_1$
and $n_2$ be the numbers of completers in the two sequences. When dropout
leaves a sequence without completers, or fewer than three completers in all,
the reference analysis retains the null hypothesis, so we may take
$n_1,n_2\ge1$ and $n_1+n_2\ge3$ below. The within-subject differences
$d_{ik}=y_{ik1}-y_{ik2}$ remove $\mu_0$, $\beta_i$ and $u_{ik}$. The estimate
$\hat\mu=(\bar d_1-\bar d_2)/2$ is Gaussian with mean $\mu$ and variance
$\sigma_W^2(n_1^{-1}+n_2^{-1})/2$. With the pooled variance $s_d^2$ of the
differences on $\nu=n_1+n_2-2$ degrees of freedom, the standard error is
$\mathrm{se}=(s_d/2)(n_1^{-1}+n_2^{-1})^{1/2}$. The reference analysis rejects
the null hypothesis of nonequivalence when the interval
$\hat\mu\pm t_{1-\alpha_0,\nu}\,\mathrm{se}$ lies inside $(-\Delta,\Delta)$,
and retains it otherwise.

\begin{theorem}[A reference analysis validated by proof]
\label{thm:cont-ref}
Under model \eqref{eq:crossover}, for every $\mu$ with $|\mu|\ge\Delta$, every
$\sigma_W>0$, all subject, period and sequence effects, and every dropout
mechanism that does not depend on the responses,
\[
  e(v_0,\theta)\le\alpha_0 .
\]
At $\mu=\Delta$ and without dropout, the Type I error rate tends to $\alpha_0$
as $\sigma_W\to0$ with the design held fixed, so the bound cannot be improved.
\end{theorem}

\begin{proof}
Condition on the sets of completers, which are independent of the responses.
Given them, $T_\mu=(\hat\mu-\mu)/\mathrm{se}$ has the $t$ distribution with
$\nu$ degrees of freedom. For $\mu\ge\Delta$ the rejection event lies in
\[
  \{\hat\mu+t_{1-\alpha_0,\nu}\,\mathrm{se}<\Delta\}
  =\{T_\mu<-t_{1-\alpha_0,\nu}+(\Delta-\mu)/\mathrm{se}\}
  \subseteq\{T_\mu<-t_{1-\alpha_0,\nu}\},
\]
because $\Delta-\mu\le0$, and the last event has probability $\alpha_0$. The
case $\mu\le-\Delta$ is symmetric, and averaging over the sets of completers
gives the bound. At $\mu=\Delta$ the rejection event is
$\{t_{1-\alpha_0,\nu}-2\Delta/\mathrm{se}<T_\mu<-t_{1-\alpha_0,\nu}\}$. As
$\sigma_W\to0$ with the design fixed and no dropout, the ratio of
$\mathrm{se}$ to the standard deviation of $\hat\mu$ keeps its distribution and
$2\Delta/\mathrm{se}\to\infty$ in probability, so the
Type I error rate tends to $\Pr\{T_\mu<-t_{1-\alpha_0,\nu}\}=\alpha_0$. With
dropout the same argument gives the limit
$\alpha_0\Pr\{n_1,n_2\ge1,\ n_1+n_2\ge3\}$.
\end{proof}

The two one-sided tests procedure holds its level exactly, a classical fact
\cite{si:schuirmann1987comparison}. The theorem records this fact for every design
and dropout pattern that the scenarios use. The Gaussian null scenarios are the
null scenarios of model \eqref{eq:crossover} with dropout that does not depend
on the responses. On these scenarios the reference bound is
$b_\theta=\alpha_0$ with no simulation, and no part of the lifetime risk is
spent on the reference analysis. By contrast, on scenarios built from real
studies, whose errors are not Gaussian, the reference analysis receives a
one-time validation (Corollary~\ref{cor:lifetime-ref}).

The reference level costs some power. At a true ratio of $0.95$, the two
one-sided tests at a reference level of 0.040 declare between \FoneReserveLossLow{} and
\FoneReserveLossHigh{} percentage points fewer equivalent products
bioequivalent than at $0.05$, in four crossovers of 24 to 48 subjects with
\FoneReservePowerLow{} to \FoneReservePowerHigh{} percent power at $0.05$.
Updated versions that pass validation can recover part of this power, for
example by running the two one-sided tests at a level just below the
validation level.

The proved bound carries over to every version whose rejections are contained
in those of the reference analysis (Corollary~\ref{cor:domination}). Examples
include a version that defers studies with high variability and a version that
adds a check that can only veto a rejection. The bound does not carry over to
versions that use the data to change which observations enter the analysis. A
version that excludes outlying subjects, or recovers information from subjects
observed in one period only, changes the analysis set. The pivot $T_\mu$ then
loses its $t$ distribution, so such versions are validated on scenarios.

\FloatBarrier

\section{Lifetime validity}
\label{sec:lifetime}

This section proves the guarantees of the main text in general form. Two
devices chain the validations of single versions over the life of the analyst.
Each round is judged by simulated studies drawn after the version of that round
is fixed, and each round spends its own share of the lifetime risk $\delta$.
Theorem~\ref{thm:lifetime} bounds the lifetime risk,
Corollary~\ref{cor:every-decision} gives every decision the level $\alpha$, and
Proposition~\ref{prop:power} shows that valid updates pass. The last three
subsections obtain the reference bounds, extend the scenario set and add random
scenarios.

\subsection{Lifetime risk control}

At round $r$ the validator receives update $r$. Before any simulated study
of the round exists, the validator fixes the updated version $v_r$ and the
scenario set $\mathcal{S}_r\supseteq\mathcal{S}_0$, together with the
reference bounds $b_\theta$ for $\theta\in\mathcal{S}_r$. It also fixes the
number $m_r$ of simulated studies per scenario and the share $\delta_r$. We
call this step the freeze of round $r$, and we write $\mathcal{G}_r$ for the
information available at the freeze. It contains everything that has happened
before, including all earlier updates, deployment decisions, bounds and simulated studies,
as well as the validation of the reference analysis. It also contains the real
studies decided so far and the randomness the proposer used to produce $v_r$.

The validator then draws the simulated studies of the round and computes the
combined bounds \eqref{eq:bounds} with $\eta_1=\eta_2=\delta_r/4$ and
$\eta_3=\delta_r/2$. Write $B_{r,\theta}$ for the combined bound of scenario
$\theta$ in round $r$. The updated version $v_r$ passes validation if
$B_{r,\theta}\le\alpha'$ for every $\theta\in\mathcal{S}_r$, and only a version
that passes can be deployed. Let $E_0$ be the event that
$b_\theta\ge e(v_0,\theta)$ for every scenario the validator ever uses.

\begin{theorem}[Lifetime risk control]
\label{thm:lifetime}
Suppose that
\begin{enumerate}[label=(\alph*),leftmargin=2.2em,itemsep=0pt]
\item in every round $r$, the updated version $v_r$, the scenario set $\mathcal{S}_r\supseteq\mathcal{S}_0$
  with its reference bounds, the number $m_r$ and the share $\delta_r$ are fixed
  at the freeze, while the proposer is otherwise unrestricted and may be
  adaptive, randomized or adversarial,
\item in every round $r$, given $\mathcal{G}_r$, the simulated studies are
  independent draws from the scenarios of $\mathcal{S}_r$, and the internal
  randomness of the runs on them is independent of $\mathcal{G}_r$, and
\item $\Pr(E_0^{c})\le\delta_0$, and $\delta_0+\sum_{r\ge1}\delta_r\le\delta$
  holds with probability one.
\end{enumerate}
Then
\begin{equation}
  \Pr\bigl\{\exists\,r\ge1:\ v_r\text{ passes validation and }
  e(v_r,\mathcal{S}_r)>\alpha'\bigr\}\le\delta .
  \label{eq:lifetime-bound}
\end{equation}
Moreover, the probability that some reference bound fails or some $v_r$
passes validation with $e(v_r,\mathcal{S}_r)>\alpha'$ is at most $\delta$.
\end{theorem}

\begin{proof}
Let $A_r$ be the event that $v_r$ passes validation and
$e(v_r,\mathcal{S}_r)>\alpha'$.
Let $\theta^*_r$ be the first scenario, in a fixed enumeration of
$\mathcal{S}_r$, at which $e(v_r,\cdot)$ is largest. Versions are finite
specifications and take values in a countable set. Hence, by (a), the updated
version $v_r$, its Type I error rates $e(v_r,\theta)$, the scenario
$\theta^*_r$ and the event $e(v_r,\theta^*_r)>\alpha'$ are determined at the
freeze. On $A_r\cap E_0$ the version $v_r$ has passed validation, so the bound
of scenario $\theta^*_r$ satisfies
$B_{r,\theta^*_r}\le\alpha'<e(v_r,\theta^*_r)$, while
$b_{\theta^*_r}\ge e(v_0,\theta^*_r)$. The proof of
Theorem~\ref{thm:relative}(i) shows that one of the three binomial bounds of
scenario $\theta^*_r$ must then miss its parameter, and we call this event
$M_r$. Given $\mathcal{G}_r$, the updated version, the scenario $\theta^*_r$,
the number $m_r$ and the share $\delta_r$ are fixed. By (b), the studies of the
round in $\theta^*_r$ are $m_r$ independent draws from that scenario, run
through both procedures with internal randomness independent of $\mathcal{G}_r$. The three counts are then
binomial with the parameters of Theorem~\ref{thm:relative}, and
Lemma~\ref{lem:cp} gives
$\Pr(M_r\mid\mathcal{G}_r)\le\delta_r/4+\delta_r/4+\delta_r/2=\delta_r$.
Hence
\[
  \Pr\Bigl(E_0^{c}\cup\bigcup_{r\ge1}A_r\Bigr)
  \le\Pr(E_0^{c})+\sum_{r\ge1}\E\bigl[\Pr(M_r\mid\mathcal{G}_r)\bigr]
  \le\delta_0+\E\Bigl[\sum_{r\ge1}\delta_r\Bigr]\le\delta ,
\]
which is the second statement and contains \eqref{eq:lifetime-bound}. Rounds
that never take place contribute nothing.
\end{proof}

The theorem controls a familywise error rate over all updates, the probability
that validation ever passes an update above the validation level in a scenario
of its round. It covers every deployed update and also versions that pass
without being deployed, such as the any-path rule of a set of allowed paths, from
which section~\ref{sec:modifications} carries validity to deployed versions. It
assumes nothing about the proposer, which may be a language model that reads
every deployment decision, an optimizer that maximizes power, or an adversary. Within a
round the argument is the intersection-union argument of
Theorem~\ref{thm:iut}, and across rounds it is a union bound, as in alpha
spending \cite{si:foster2008alpha,si:tian2021online}.

Updates are counted by attempt, so update $r$ spends the share $\delta_r$
whether or not it passes. If shares were indexed by deployed versions, a
proposer could resubmit after each failure at an unchanged share, and the
familywise error rate would grow as in Lemma~\ref{lem:flat}. The shares may be
chosen adaptively, as long as each is fixed at its freeze and their total never
exceeds $\delta-\delta_0$. The protocol uses
\begin{equation}
  \delta_r=(\delta-\delta_0)\,\frac{6}{\pi^2r^2},\qquad
  \sum_{r\ge1}\delta_r=\delta-\delta_0 ,
  \label{eq:spending}
\end{equation}
and any other summable schedule works as well. The shrinking share costs little,
because the number of simulated studies that an update needs grows only with
$\log(1/\delta_r)$, which is $2\log r$ plus a constant. A version that never adds
rejections thus needs \FonePairedLast{} simulated studies per scenario at the
thousandth update, against \FonePairedFirst{} at the first
(Theorem~\ref{thm:cost-rel}(i) and Fig.~\ref{fig:si-concept}\textit{C}). The scenario set
of a round may be any set fixed at its freeze that contains $\mathcal{S}_0$, so
every version that passes is tested on the initial scenario set, and the set
grows with the analyst (section~\ref{sec:extension}).

\subsection{Every decision the analyst makes}

Studies are decided while the analyst evolves, and we write $\mathcal{H}_s$ for
the information available when the data of study $s$ arrive. It includes every
round completed by then with its simulated studies, all earlier studies with
their data and decisions, and the true scenario $\theta_s$. The version in force
for study $s$, written $v_{(s)}$, is determined by $\mathcal{H}_s$.

For a version $v$, $\mathcal{S}(v)$ denotes the scenario set on which it
passed validation or whose validity it inherits. The reference analysis has
$\mathcal{S}(v_0)=\mathcal{S}_0$, and so does the version that defers every
study. A version deployed without validation inherits the scenario set named
in section~\ref{sec:kinds}.

\begin{corollary}[Type I error control for every decision]
\label{cor:every-decision}
Let the validator run the protocol of Theorem~\ref{thm:lifetime} with
validation level $\alpha'=(\alpha-\delta)/(1-\delta)$. Let the analyst start
with the reference analysis, or with a version that rejects only where the
reference analysis rejects, if $\max_{\theta\in\mathcal{S}_0}b_\theta\le\alpha'$,
and otherwise with the version that defers every study. Consider a study $s$
such that
\begin{enumerate}[label=(\roman*),leftmargin=2.2em,itemsep=0pt]
\item whatever the history of updates, its true scenario is one on which the
  version in force was validated or whose validity that version inherits,
  that is, $\theta_s\in\mathcal{S}(v_{(s)})$ almost surely, where $\theta_s$
  may depend on everything that happened before, and
\item given $\mathcal{H}_s$, its data are a draw from scenario $\theta_s$, as
  Assumption~\ref{ass:faithful} provides for a real study in that state, and
  the internal randomness of the analyst on it is independent of $\mathcal{H}_s$, so that the version
  deciding the study was fixed before the data of the study were seen.
\end{enumerate}
Then
\begin{equation}
  \Pr\{Y_s=\text{reject}\}\le\alpha'+(1-\alpha')\,\delta=\alpha .
  \label{eq:every-decision}
\end{equation}
Under (ii) alone,
$\Pr\{Y_s=\text{reject},\ \theta_s\in\mathcal{S}(v_{(s)})\}\le
\alpha'\Pr\{\theta_s\in\mathcal{S}(v_{(s)})\}+(1-\alpha')\,\delta$.
\end{corollary}

\begin{proof}
Let $\mathcal{E}$ be the event that every reference bound holds and no $v_r$
passes validation with $e(v_r,\mathcal{S}_r)>\alpha'$, the complement of the
event $E_0^{c}\cup\bigcup_rA_r$ of the proof of Theorem~\ref{thm:lifetime}. The
second statement of Theorem~\ref{thm:lifetime} gives
$\Pr(\mathcal{E}^{c})\le\delta$. On
$\mathcal{E}$ every version that passes validation in a round $r$ is valid at
level $\alpha'$ on $\mathcal{S}_r$ by the definition of $A_r$, whether or not
it is deployed. These versions include deployed updates, any-path rules of
allowed paths, members of validated families and updates held back by further
conditions for deployment. On $\mathcal{E}$ the starting version is valid at
level $\alpha'$ on $\mathcal{S}_0$ as well. The reference analysis has
$e(v_0,\theta)\le b_\theta\le\alpha'$ there, a version that rejects only where
it rejects has a Type I error rate at most that of the reference analysis
(Corollary~\ref{cor:domination}), and the version that defers every study
never rejects. Versions deployed without validation inherit this validity on
their scenario sets (section~\ref{sec:kinds}), so on $\mathcal{E}$ every
version that is ever in force is valid at level $\alpha'$ on its scenario set.
Write $C_s$ for the event $\theta_s\in\mathcal{S}(v_{(s)})$, which is
determined by $\mathcal{H}_s$. By (ii),
$\Pr\{Y_s=\text{reject}\mid\mathcal{H}_s\}=e(v_{(s)},\theta_s)$, which is at most
$\alpha'$ on $\mathcal{E}\cap C_s$ and at most one elsewhere, so
\[
  \Pr\{Y_s=\text{reject},\ C_s\}=\E\bigl[e(v_{(s)},\theta_s)\,1_{C_s}\bigr]
  \le\alpha'\Pr(C_s)+(1-\alpha')\Pr(C_s\cap\mathcal{E}^{c})
  \le\alpha'\Pr(C_s)+(1-\alpha')\,\delta .
\]
Under (i), $\Pr(C_s)=1$ and this is \eqref{eq:every-decision}, since
$\alpha'(1-\delta)+\delta=\alpha$ for $\alpha'=(\alpha-\delta)/(1-\delta)$.
\end{proof}

Suppose that condition (ii) holds for every study, that a version of row 4 of
Table~\ref{tab:modifications} is deployed only while every reference bound on
$\mathcal{S}_0$ is at most $\alpha'$, and that every routing family of row 5
is validated on $\mathcal{S}_0$. Then $\mathcal{S}(v)\supseteq\mathcal{S}_0$
for every version in force, and the analyst is lifetime valid at level $\alpha$
on $\mathcal{S}_0$ in the sense of Definition~\ref{def:lifetime}.
Section~\ref{sec:modifications} shows that versions deployed without new
simulation keep the corollary.

The bound of Corollary~\ref{cor:every-decision} is a probability over the whole
history of the analyst. Given the history up to its data, a decision on a study whose
true scenario lies in the scenario set of the deciding version is a Type I
error with probability at most $\alpha'$, except on the event
$\mathcal{E}^{c}$, whose probability is at most $\delta$ by
Theorem~\ref{thm:lifetime}.

By condition (ii), an update prompted by a study applies to later studies
only, and this condition cannot be dropped. Suppose that, after seeing the
data $D_s$ of study $s$, a proposer submits the version $\tilde v$ that rejects on $D_s$ and
otherwise copies the reference analysis. When $\Pr_\theta$ has no atoms,
$e(\tilde v,\theta)=e(v_0,\theta)$, so $\tilde v$ passes validation. Yet if
$\tilde v$ decided study $s$, it would reject with probability one. The
protocol therefore deploys a version that passes validation only for studies
whose data arrive later.

The guarantee is continuous in total variation in the law of the data. For
every version and every law $P$ of the data, the probability of a rejection
under $P$ differs from $e(v,\theta)$ by at most the total variation distance
between $P$ and $\Pr_\theta$, because the rejection probability given the data
lies in $[0,1]$. A scenario whose law is close to that of real studies in total
variation therefore carries the guarantee to them up to that distance. Whole
studies are rarely that close, since a small change in a variance or a
different sample size already moves the law of a study far in total
variation. Between and beyond the fixed scenarios, the guarantee rests on
random scenarios (Proposition~\ref{prop:random}). Without them, a version can
exceed the validation level off the fixed scenarios
(Proposition~\ref{prop:grid} and the example of section~\ref{sec:random}).

\begin{corollary}[Long-run frequency of false rejections]
\label{cor:frequency}
Suppose the protocol of Corollary~\ref{cor:every-decision} runs, that its
condition (ii) holds for every study, and that each decision is made before the
data of the next study arrive. Let $s_1<s_2<\cdots$ be the studies whose true scenario lies in the
scenario set of the version in force. Whether a study is of this kind is
determined before its data arrive. Let $F_k$ be the number of rejections among
the first $k$ of these studies. Then, for every $k$ and every $h>0$,
\begin{enumerate}[label=(\alph*),leftmargin=2.2em,itemsep=0pt]
\item $\E F_k\le k\alpha$,
\item $\Pr\{F_k>k(\alpha'+h)\}\le\delta+\exp(-2kh^2)$, and
\item with probability at least $1-\delta$, $\limsup_{k\to\infty}F_k/k\le\alpha'$
  whenever infinitely many such studies occur.
\end{enumerate}
\end{corollary}

\begin{proof}
The event that study $s$ is the $i$th study of this kind is determined by
$\mathcal{H}_s$, so the proof of Corollary~\ref{cor:every-decision} applies to
each $s_i$ and gives (a). For (b) and (c), let $Z_i$ indicate a rejection on
study $s_i$, and let
$p_i=\Pr(Z_i=1\mid\mathcal{H}_{s_i})=e(v_{(s_i)},\theta_{s_i})$. The decision $Z_i$
is known when the data of study $s_{i+1}$ arrive, so the sums
$\Psi_k=\sum_{i\le k}(Z_i-p_i)$ form a martingale with increments in intervals
of length one. The Azuma-Hoeffding inequality then gives
$\Pr\{\Psi_k>kh\}\le\exp(-2kh^2)$. On the event $\mathcal{E}$ of the proof of
Corollary~\ref{cor:every-decision}, every version ever in force is valid at
level $\alpha'$ on its scenario set, so $p_i\le\alpha'$ for every $i$ and
$F_k\le k\alpha'+\Psi_k$, and $\Pr(\mathcal{E}^{c})\le\delta$ by the second
statement of Theorem~\ref{thm:lifetime}. Together these give (b). The probabilities
$\Pr\{|\Psi_k|>kh\}\le2\exp(-2kh^2)$ are summable in $k$, so $\Psi_k/k\to0$
almost surely by the Borel-Cantelli lemma, and on $\mathcal{E}$ this gives (c).
\end{proof}

\subsection{Valid updates pass}
\label{sec:power}

Type I error control alone is trivial, since a validator that never passes an update achieves it.
The next result shows that valid updates pass continual validation with high probability, at a number of
simulated studies that grows only with the logarithm of the update index.

\begin{proposition}[Valid updates pass]
\label{prop:power}
Fix an update $r$, a margin $\gamma\in(0,\alpha')$ and $\zeta\in(0,1)$, and let the scenario set of round $r$ contain
$J_r$ scenarios. Suppose $e(v_r,\theta)\le\alpha'-\gamma$ for every scenario $\theta$ of the round.
\begin{enumerate}[label=(\roman*),leftmargin=2.2em,itemsep=0pt]
\item If $m_r\ge 8\alpha'\ell_r/\gamma^2$ with $\ell_r=\max\{\log(2/\delta_r),\log(J_r/\zeta)\}$, then $v_r$ passes
  validation with probability at least $1-\zeta$.
\item If $v_r$ never rejects where the reference analysis does not, so that $q^{+}_\theta=0$, and
  $b_\theta\le\alpha'-\gamma$ in every scenario, then $v_r$ passes validation with probability one as soon as
  $m_r\ge\log(4/\delta_r)/\gamma$.
\end{enumerate}
\end{proposition}

\begin{proof}
(i) Apply Theorem~\ref{thm:cost-rel}(iii) in each scenario with target $a=\alpha'$, level $\eta_3=\delta_r/2$ and
failure probability $\zeta/J_r$. The theorem is stated for $e(v_r,\theta)=a-\gamma$. The direct bound increases with the
count of rejections, and this count is stochastically increasing in $e(v_r,\theta)$, so the probability that the
direct bound is at most $\alpha'$ can only grow when $e(v_r,\theta)<a-\gamma$. The combined bound is at most the direct
bound, and a union bound over the $J_r$ scenarios gives (i). (ii) Apply Theorem~\ref{thm:cost-rel}(i) with level
$\eta_1=\delta_r/4$ in every scenario.
\end{proof}

\begin{remark}[Further conditions for deployment]
A deployment method may add conditions to validation, for example that a lower confidence bound on the gain in power is
positive. Such a method deploys a subset of the versions that pass validation. The events $A_r$ in the proof of
Theorem~\ref{thm:lifetime} concern passing and not deployment, so they are unchanged, and Theorem~\ref{thm:lifetime}
and Corollary~\ref{cor:every-decision} continue to hold. An update held back by such a condition has still passed
validation, and a version deployed later without validation can inherit its validity (section~\ref{sec:kinds}).
\end{remark}
\subsection{The reference analysis}

\begin{corollary}[Lifetime risk control with a validated reference analysis]
\label{cor:lifetime-ref}
Let the reference bounds of the initial scenario set be obtained in either of
two ways, or by a mixture of them across scenarios.
\begin{enumerate}[label=(\roman*),leftmargin=2.2em,itemsep=0pt]
\item By proof. On Gaussian null scenarios $b_\theta=\alpha_0$, by
  Theorem~\ref{thm:cont-ref}.
\item By one validation. Before the first update the validator simulates
  $m_0$ studies in each remaining scenario of $\mathcal{S}_0$, counts the rejections $X^0_\theta$ of the reference analysis, and sets
  $b_\theta=U(X^0_\theta;m_0,\delta_{0,\theta})$, with levels $\delta_{0,\theta}$
  fixed before the first update whose sum $\delta_0'$ over these scenarios is at
  most $\delta_0$.
\end{enumerate}
Let every scenario added later receive a bound by proof, the trivial bound
$b_\theta=1$, or a validation of the same kind. The $k$th added scenario
that receives such a validation has a level $\delta_{0,k}$ fixed before the
first update, with $\delta_0'+\sum_k\delta_{0,k}\le\delta_0$. Its $m_0$
simulated studies are drawn after the scenario is fixed, independently of
everything before.
Then $\Pr(E_0^{c})\le\delta_0$, and the protocol with the shares
\eqref{eq:spending} satisfies \eqref{eq:lifetime-bound}.
\end{corollary}

\begin{proof}
Analytic bounds and the trivial bound hold with certainty. The reference
analysis is fixed before its bounds are computed, so each count
$X^0_\theta$ for $\theta\in\mathcal{S}_0$ is binomial with success probability
$e(v_0,\theta)$. By Lemma~\ref{lem:cp}, the event
$b_\theta<e(v_0,\theta)$ then has probability at most $\delta_{0,\theta}$ in each
of these scenarios. An added scenario $\theta_k$ may be chosen from the history.
Given the information available when $\theta_k$ is fixed, its count
$X^0_{\theta_k}$ is binomial with success probability $e(v_0,\theta_k)$. By
Lemma~\ref{lem:cp}, the event $b_{\theta_k}<e(v_0,\theta_k)$ then has
conditional probability at most $\delta_{0,k}$. Taking expectations and a union
bound over $k$ and over $\mathcal{S}_0$ gives
$\Pr(E_0^{c})\le\delta_0'+\sum_k\delta_{0,k}\le\delta_0$, and
Theorem~\ref{thm:lifetime} applies.
\end{proof}

The reference bounds need simultaneous coverage, holding in all scenarios at
once. They enter every later validation, and a proposer that adapts to earlier
deployment decisions can steer its updated version toward a scenario whose reference bound
happens to be too low.
This requirement divides the level of the bounds among the $J$ scenarios.
The division is made once, when the reference analysis is validated, and the
updates themselves need none (Theorem~\ref{thm:iut}). On Gaussian null
scenarios neither division nor simulation is needed. In the bioequivalence
analyst, the validation of the reference analysis on scenarios built from real
studies spends $\delta_0=0.001$, and the updates share the remaining $0.004$.

In a scenario with the trivial bound, the paired bound exceeds the validation
level unless the reference analysis rejects on almost every simulated study. The
direct bound therefore decides validation there, and the scenario set can grow
at any time without a new validation of the reference analysis.

\subsection{Scenarios grow with the information a version reads}
\label{sec:extension}

A validation concerns only the scenarios that it uses. A version may read
information that the validation scenarios hold fixed, and its Type I error
rate can then depend on a factor that no scenario varies. The following
existence statement makes this precise in the crossover model.

\begin{proposition}[New information needs new scenarios]
\label{prop:grid}
In model \eqref{eq:crossover} with Gaussian subject effects, let
$\gamma=\alpha'-\alpha_0$, and let the scenario set of round $r$ consist of
$J_r$ null scenarios without sequence effect, $\kappa=0$, with the analytic
reference bound $b_\theta=\alpha_0$. There is a version $v'$ with the
following three properties. First, $e(v',\theta)\le\alpha_0+\gamma/2$ in every
null scenario with $\kappa=0$. Second, it passes validation in round $r$ with
probability at least $1-\zeta$ once
$m_r\ge32\max\{\log(4/\delta_r),\log(J_r/\zeta)\}/\gamma$.
Third, $e(v',\theta)\to1$ as $\kappa\to\infty$ with the other parameters fixed.
The Type I error rate of the reference analysis does not depend on $\kappa$.
\end{proposition}

\begin{proof}
Let $w_{ik}=y_{ik1}+y_{ik2}$ be the subject totals. Their means in the two
sequences differ in expectation by $2\kappa$, whatever $\mu$. Let $s_w$ be the
pooled standard deviation of the totals on $\nu$ degrees of freedom. When
$\kappa=0$, the statistic
$T_\kappa=(\bar w_1-\bar w_2)/\{s_w(n_1^{-1}+n_2^{-1})^{1/2}\}$ has the $t$
distribution with $\nu$ degrees of freedom, for every $\mu$, $\sigma_W$ and
subject-effect variance. Let $v'$ reject when the
reference analysis rejects or when $T_\kappa>t_{1-\gamma/2,\nu}$. When $\kappa=0$ the
second event has probability $\gamma/2$, so $e(v',\theta)\le\alpha_0+\gamma/2$
by Theorem~\ref{thm:cont-ref}, and $v'$ adds rejections at rate
$q^{+}_\theta\le\gamma/2$. Theorem~\ref{thm:cost-rel}(ii) with $c=1/2$,
$b_\theta=\alpha'-\gamma$ and level $\eta_1=\delta_r/4$ shows that each of
the $J_r$ scenarios passes with probability at least $1-\zeta/J_r$ under the
stated $m_r$. Hence all of them pass with probability at least $1-\zeta$. As
$\kappa\to\infty$, $T_\kappa\to\infty$ in probability, so
$e(v',\theta)\ge\Pr\{T_\kappa>t_{1-\gamma/2,\nu}\}\to1$. The differences
$d_{ik}$ do not involve the sequence effects, so the reference analysis does
not depend on $\kappa$.
\end{proof}

The version $v'$ is the simplest example, and a practical version behaves the
same way. A version that recovers information from subjects observed in only
one period uses a between-subject contrast whose mean is $\mu$ plus the
sequence effect. Its Type I error rate thus moves with $\kappa$, while the
reference analysis, which uses within-subject differences, is exactly
invariant to it. A scenario set that suffices for the reference analysis
therefore does not suffice for such a version.

The validator extends the scenario set whenever an updated version reads a
kind of information that no current scenario varies, here by adding scenarios
with a sequence effect. The validator decides the extension and fixes it at
the freeze, and the added scenarios stay in the set from then on. The
reference bound on them comes from proved invariance as here, from a
validation at a level set aside within $\delta_0$, or from the trivial bound
(Corollary~\ref{cor:lifetime-ref}). The guarantee for a decision covers every
scenario on which the version in force is validated whatever the history
(Corollary~\ref{cor:every-decision}), such as every scenario of the initial set,
and the scenario sets of later versions grow whenever the analyst gains a new
capability. The extension can be audited when every update declares the
inputs its decision reads and the software enforces the declaration, because
the validator then sees which factors need scenarios.

\subsection{Random scenarios extend the guarantee between fixed scenarios}
\label{sec:random}

A fixed scenario set is known to the proposer, who learns every deployment decision and the
scenario that decided it. A version can then copy the reference analysis on
every validation scenario and behave differently elsewhere. Consider a version
that copies the reference analysis whenever the planned sample size is one of
those the scenarios use and rejects at every other sample size. It adds no
rejection in any validation scenario, so its paired bound passes whenever that
of the reference analysis would, yet its Type I error rate is one at every
other sample size. This is the analogue of Proposition~\ref{prop:reuse} for
scenarios. Newly simulated studies defeat
versions {that overfit} a stored suite of studies, and newly drawn scenarios do the same
for versions {that overfit} a fixed set of scenarios, up to a set of small
probability.

\begin{proposition}[Random scenarios]
\label{prop:random}
Let $\Pi$ be a probability distribution on $\Theta_0$, fixed before the first
update, under which every Type I error rate $e(v,\cdot)$ is measurable. In
round $r$, after the freeze, let the validator draw a set $\mathcal{Q}_r$ of
$K_r$ scenarios independently from $\Pi$ and independently of
$\mathcal{G}_r$, where the numbers $K_r$ are fixed before the first
update. It adds them to the scenarios of the round with the analytic
reference bound where one exists and the trivial bound $b_\theta=1$
otherwise, and only then draws the simulated studies of the round. Let
$V_r=\{\theta\in\Theta_0:e(v_r,\theta)>\alpha'\}$ be the set on which the
updated version $v_r$ exceeds the validation level. Under the conditions of
Theorem~\ref{thm:lifetime}, for every $\varepsilon_\Pi\in(0,1)$,
\begin{equation}
  \Pr\bigl\{\exists\,r\ge1:\ v_r\text{ passes validation, and }
  e(v_r,\mathcal{S}_r)>\alpha'\text{ or }\Pi(V_r)>\varepsilon_\Pi\bigr\}
  \le\delta+\sum_{r\ge1}(1-\varepsilon_\Pi)^{K_r}.
  \label{eq:random-bound}
\end{equation}
\end{proposition}

\begin{proof}
Let $\mathcal{G}'_r$ be the information at the freeze together with the drawn
set $\mathcal{Q}_r$. Given $\mathcal{G}'_r$, the updated version, the scenario
set $\mathcal{S}_r\cup\mathcal{Q}_r$ with its reference bounds, the number
$m_r$ and the share $\delta_r$ are fixed, and the simulated studies are drawn
afterward. Conditions (a) and (b) of Theorem~\ref{thm:lifetime} therefore hold
with $\mathcal{G}'_r$ in place of $\mathcal{G}_r$. Analytic and trivial bounds
hold with certainty, so the event $E_0$ is unchanged. The theorem gives
$\Pr\{\exists\,r:\ v_r\text{ passes validation and }
e(v_r,\mathcal{S}_r\cup\mathcal{Q}_r)>\alpha'\}\le\delta$. Let $\Gamma_r$ be
the event that $\Pi(V_r)>\varepsilon_\Pi$ and no scenario of $\mathcal{Q}_r$
lies in $V_r$. The set $V_r$ is determined by $v_r$ and hence by
$\mathcal{G}_r$, while $\mathcal{Q}_r$ is drawn independently of
$\mathcal{G}_r$, so
$\Pr(\Gamma_r\mid\mathcal{G}_r)=\{1-\Pi(V_r)\}^{K_r}\,1\{\Pi(V_r)>\varepsilon_\Pi\}
\le(1-\varepsilon_\Pi)^{K_r}$. If $v_r$ passes validation and
$\Pi(V_r)>\varepsilon_\Pi$, then either $\Gamma_r$ occurs or a scenario of
$\mathcal{Q}_r$ lies in $V_r$, and in the second case
$e(v_r,\mathcal{S}_r\cup\mathcal{Q}_r)>\alpha'$. A union bound over these
events and over $r$ gives \eqref{eq:random-bound}.
\end{proof}

The random scenarios cost no share of the lifetime risk, because within a round
the intersection-union argument needs no division across scenarios. They cost
only the probability that no draw falls in a set of $\Pi$-probability above
$\varepsilon_\Pi$ on which a version exceeds the validation level. Since
$(1-\varepsilon_\Pi)^{K_r}\le\exp(-\varepsilon_\Pi K_r)$, the choice
$K_r=\lceil\varepsilon_\Pi^{-1}\log\{\pi^2r^2/(6\omega)\}\rceil$ holds the
added term at $\omega$. With $\varepsilon_\Pi=0.01$ and $\omega=0.005$, the
first update draws \FoneRandKFirst{} scenarios and the thousandth draws
\FoneRandKLast. With probability at least $1-\delta-\omega$, every update that
passes validation, and so every deployed update, is then within the validation
level on every fixed scenario of its round and on all of the null space except
a set of $\Pi$-probability at most $\varepsilon_\Pi$. Once $\Pi$ gives sample
sizes outside the fixed scenarios a probability above $\varepsilon_\Pi$, the
version of the example above fails validation except with probability at most
$(1-\varepsilon_\Pi)^{K_r}+\delta_r$.

All random scenarios of a round must pass together. For the
\FoneRandKFirst{} scenarios of the first update to pass jointly with
probability 0.9 by the direct bound, a version at the reference level needs
\FoneRandM{} simulated studies in each, against \FoneScratchFirst{} for a
single scenario (Table~\ref{tab:costs}). On Gaussian scenarios the analytic
bound of the reference analysis holds everywhere (Theorem~\ref{thm:cont-ref}),
and random scenarios keep the paired bound and its lower cost. The draws also
make the identity of a deciding scenario uninformative for the proposer, since
the next round draws new ones.

\FloatBarrier

\section{Updates that need no new validation}
\label{sec:modifications}

Every update produces a new version, and Theorem~\ref{thm:lifetime} never
asks what the update is or who made it. Many updates need no new simulation,
because an argument bounds their Type I error rate by that of a version
already known to be valid. The subsections below give these arguments, and
Table~\ref{tab:modifications} collects them.

\subsection{Allowed paths}

An AI analyst is usually a language model inside software that limits what the
model can do. An allowed path is one complete sequence of steps, from the
start of a study to a decision, that the software lets the model take. For example, one path
might check the data against the protocol, analyze the completers with the two
one-sided tests procedure at level $\alpha_0$, and decide. Another path might
defer the study to a human expert. Let $\mathcal{R}$ be the set of allowed
paths, where each path $\rho\in\mathcal{R}$ ends in a decision $y_\rho(D)$ that
is a fixed function of the data. The language model chooses which path to
follow, and its choice may depend on the data, on its own randomness, on its
instructions and on the model itself.

\begin{proposition}[Allowed paths]
\label{prop:envelope}
Let $v$ be a version whose software mechanically confines every run to a path
in $\mathcal{R}$. Each path ends in a decision $y_\rho(D)$ that is a fixed
function of the data, so that the model selects paths and computes nothing
that enters a decision. Let the any-path rule $v_{\mathcal{R}}$ be the
procedure that rejects whenever some allowed path rejects,
\[
  \{Y_{v_{\mathcal{R}}}=\text{reject}\}=\bigcup_{\rho\in\mathcal{R}}\{y_\rho(D)=\text{reject}\}.
\]
Then $e(v,\theta)\le e(v_{\mathcal{R}},\theta)$ in every scenario, whatever rule
selects the path. In particular, suppose every path that can reject applies
one analysis to data prepared by one fixed rule. Then the Type I error rate of
$v$ is at most that of this analysis.
\end{proposition}

\begin{proof}
Whatever path is followed, the decision is a rejection only if that path
rejects, so
\[
  \{Y_v=\text{reject}\}\subseteq\bigcup_{\rho\in\mathcal{R}}
  \{\rho\text{ is followed and }y_\rho(D)=\text{reject}\}
  \subseteq\{Y_{v_{\mathcal{R}}}=\text{reject}\}.
\]
The last event depends neither on the rule that selects the path nor on the
internal randomness of the model.
\end{proof}

A version is validated through its allowed paths when the any-path rule
$v_{\mathcal{R}}$ passes validation. A later update that keeps the paths and
their decision functions leaves this bound in force and needs no new
validation. A swapped model, an edited prompt, a changed sampling temperature
and a provider update are updates of this kind. The model may be stochastic,
may drift, may become more or less cautious, and may act adversarially within
the allowed paths. None of this can enlarge the rejection region beyond that
of $v_{\mathcal{R}}$.

The software must deliver both hypotheses, enforcing the set of allowed paths
mechanically and never letting the model compute a number that enters a
decision. The model then chooses among paths, and the paths compute. A
validation of the version itself, by contrast, depends on how the model
chooses, and a model swap changes the object it validates. The bound is
tightest when one analysis carries every rejection, as in an analyst whose
model can only choose between running the sanctioned analysis and deferring.

\subsection{Deferral}

\begin{proposition}[Deferral never raises the Type I error rate]
\label{prop:escalation}
Let $v'$ be a version, and let $v$ run exactly as $v'$ except that it defers on
an event $E$. The event $E$ may depend on anything, including the data and the
internal randomness of the analyst. Then $e(v,\theta)\le e(v',\theta)$ in every
scenario.
\end{proposition}

\begin{proof}
On a common run, $\{Y_v=\text{reject}\}=\{Y_{v'}=\text{reject}\}\cap E^{c}
\subseteq\{Y_{v'}=\text{reject}\}$.
\end{proof}

Handing a study to a human never raises the Type I error rate of the decisions the
analyst makes on its own. A stricter deferral rule, a new reason to ask for
review and a human who can stop a rejection therefore need no validation. A
human who can also turn a deferral into a rejection forms a new decision
procedure together with the analyst. That procedure can be validated as a
version in its own right.

\subsection{Containment in the rejection region of the reference analysis}

\begin{corollary}[Containment]
\label{cor:domination}
Suppose $v$ rejects only on runs on which the reference analysis rejects, that
is $\{Y_v=\text{reject}\}\subseteq\{Y_{v_0}=\text{reject}\}$ for every data set and every
value of the internal randomness. Then $e(v,\theta)\le e(v_0,\theta)$ in every
scenario. For the reference analysis of Theorem~\ref{thm:cont-ref},
$e(v,\theta)\le\alpha_0$ in every Gaussian null scenario, with no simulation and
no share of the lifetime risk.
\end{corollary}

\begin{proof}
Probability is monotone under inclusion, and Theorem~\ref{thm:cont-ref} gives
the second statement.
\end{proof}

Consider a version that keeps the reference analysis as the last step before
every rejection and adds only reasons to defer or to retain the null
hypothesis. Such a version is valid on the whole Gaussian null family by
construction. Containment is proved from the definition of the version, and it
then holds exactly in every scenario, whether validation uses it or not. When
containment holds but has not been proved, such a version passes the paired
bound with few simulated studies on the validation scenarios, because it adds
no rejections (Theorem~\ref{thm:cost-rel}(i)).

\subsection{Choosing among validated analyses}

\begin{proposition}[Routing among validated analyses]
\label{prop:routes}
Let $\{\phi_\lambda:\lambda\in\Lambda\}$ be a countable family of analyses with
$e(\phi_\lambda,\theta)\le a$ for every $\lambda$ and every $\theta$ in a set
$\mathcal{S}\subseteq\Theta_0$. Let $v$ compute an index $\hat\lambda=\lambda(I)$
from routing information $I$ and reject when $\phi_{\hat\lambda}$ rejects.
Suppose that in every scenario $\theta\in\mathcal{S}$ the routing information
$I$ is independent of the statistics from which the analyses decide, together
with all internal randomness that they use. Then $e(v,\theta)\le a$ for every $\theta\in\mathcal{S}$.
\end{proposition}

\begin{proof}
Each $\phi_\lambda$ is a function of these statistics and of the internal
randomness of the analyses, which are independent of $I$ and hence of
$\hat\lambda$, so
\[
  e(v,\theta)=\sum_{\lambda\in\Lambda}
  \Pr\nolimits_\theta\{\hat\lambda=\lambda\}\,\Pr\nolimits_\theta\{\phi_\lambda=\text{reject}\}\le a .
  \qedhere
\]
\end{proof}

Independence of the routing information from the data alone is not enough. Let
$W_1,\dots,W_K$ be independent uniform variables, independent of the data, and
let $\phi_k$ reject when $W_k<\alpha$. The information $I=(W_1,\dots,W_K)$ is
independent of the data, yet the rule that picks the first $k$ with
$W_k<\alpha$ rejects with probability $1-(1-\alpha)^K$.

Routing on information fixed before the data, such as the design declared in
the protocol, satisfies the condition. So does routing on an ancillary
statistic, one whose distribution is free of the unknown parameters, provided
the analyses are functions of a complete sufficient statistic. By Basu's
theorem, the routing statistic and the complete sufficient statistic are then
independent. In the Gaussian crossover model, for example, consider
level-$\alpha_0$ analyses that
use the data only through $(\bar d_1,\bar d_2,s_d)$. A diagnostic computed from
the standardized residuals of the within-subject differences may choose among
them, and the Type I error rate stays at most $\alpha_0$.

A routing family, the family of analyses of Proposition~\ref{prop:routes}, is
validated on $\mathcal{S}$ when each member is valid at level $\alpha'$ on
$\mathcal{S}$, by proof or by passing validation on the event $\mathcal{E}$ of
Corollary~\ref{cor:every-decision}.

\subsection{What each kind of change owes}
\label{sec:kinds}

Table~\ref{tab:modifications} sorts changes into six kinds by the evidence
they owe. Versions deployed through rows 1 to 5 inherit validity on the event
$\mathcal{E}$ of Corollary~\ref{cor:every-decision}. This holds because each
row bounds the Type I error
rate of the new version by that of a version already known to be valid on the
same scenarios. That version is the one being changed, the any-path rule of
the allowed paths, the reference analysis when its bounds are at most
$\alpha'$, or a validated family of analyses. The new version $v$ inherits a
scenario set $\mathcal{S}(v)$, which enters condition (i) of
Corollary~\ref{cor:every-decision}. This set is $\mathcal{S}(v')$ for a change
of $v'$ in row 1 or row 3, and $\mathcal{S}(v_{\mathcal{R}})$ for a change in
row 2. For a version in row 4 it consists of the scenarios in force at its
deployment with $b_\theta\le\alpha'$. For a version in row 5 it is the set on
which every member of the family is validated.
Corollaries~\ref{cor:every-decision} and~\ref{cor:frequency} therefore cover
every version the analyst ever runs. A change in row 6 is validated as an
update in the sense of section~\ref{sec:lifetime}.

\begin{table}[!htbp]
\centering\small
\caption{Kinds of change and the evidence each owes before deployment. Changes
in rows 1 to 5 need no new simulation and spend no share of the lifetime
risk.}
\label{tab:modifications}
\renewcommand{\arraystretch}{1.15}
\begin{adjustbox}{max width=\textwidth}\begin{tabular}{@{}p{0.04\textwidth}p{0.25\textwidth}p{0.25\textwidth}p{0.20\textwidth}p{0.16\textwidth}@{}}
\toprule
& kind of change & examples & evidence owed & result \\
\midrule
1 & cannot change any decision & report formatting, logging, an equivalent
implementation & a proof that every path to a decision runs identical code on
identical inputs. An identity that is only tested goes through row 6, where its
validation needs few simulated studies & same Type I error rate \\
2 & keeps the allowed paths and their decision functions & model swap, prompt
edit, sampling temperature, provider update & none beyond the validation of
the any-path rule & Prop.~\ref{prop:envelope} \\
3 & can only replace a rejection by a deferral & stricter deferral rule, a new reason
for review, a human veto & none & Prop.~\ref{prop:escalation} \\
4 & rejects only where the reference analysis rejects, proved from its
definition & a deferral path placed before the reference analysis, a check
that can only veto & none, valid on the whole Gaussian null family at $\alpha_0$ &
Cor.~\ref{cor:domination}, Thm.~\ref{thm:cont-ref} \\
5 & chooses among validated analyses by information independent of the
statistics they use and of their randomness & the analysis chosen by the design declared in the
protocol & validation of the family & Prop.~\ref{prop:routes} \\
6 & anything else & a new exclusion rule, a switch between analyses after
seeing data, a second stage of sampling, a new input & validation with the
combined bounds on newly simulated studies at the share of the update, with
scenarios added for new inputs & Thm.~\ref{thm:lifetime} \\
\bottomrule
\end{tabular}\end{adjustbox}
\end{table}

\FloatBarrier

\section{The procedure}
\label{sec:protocol}

Algorithm~\ref{alg:protocol} collects the pieces of the preceding sections in
one procedure. It repeats Algorithm~1 of the main text line by line and adds two
details, the random scenarios of section~\ref{sec:random} in line 8 and the
report of the deciding scenario in line 16. Each of its lines delivers a
condition of the results above.

\begin{algorithm}[!htbp]
\caption{Continual validation}
\label{alg:protocol}
\begin{algorithmic}[1]
\Require Type I error level $\alpha$, lifetime risk $\delta$, reference risk
  $\delta_0$ with the levels $\delta_0'$ and $\delta_{0,k}$ of the reference
  bounds, initial null scenarios $\mathcal{S}_0$, reference analysis $v_0$ with
  bounds $b_\theta$, the law $\Pi$ and the numbers $K_r\ge0$ of random scenarios
\State Set $\alpha'\gets(\alpha-\delta)/(1-\delta)$ and $\mathcal{S}\gets\mathcal{S}_0$
\State Deploy $v_0$, or a version that rejects only where $v_0$ does, if
  $b_\theta\le\alpha'$ for every $\theta\in\mathcal{S}_0$. Otherwise, start with a
  version that defers every study.
\For{each proposed update $r=1,2,\ldots$}
  \State Freeze the updated version $v_r$
  \State Add scenarios to $\mathcal{S}$ for any new inputs used by $v_r$, each with
    a reference bound $b_\theta$ that is analytic, validated at its level
    $\delta_{0,k}$, or $1$
  \State Set $\mathcal{S}_r\gets\mathcal{S}$
  \State Set $\delta_r\gets6(\delta-\delta_0)/(\pi^2r^2)$ and fix the number $m_r$
    of datasets per scenario
  \State Draw $K_r$ scenarios from $\Pi$ with analytic or trivial bounds and add
    them to $\mathcal{S}_r$ for this round only
  \State Simulate $m_r$ new datasets from each $\theta\in\mathcal{S}_r$
  \State Run $v_r$ and $v_0$ on the same datasets
  \State For each $\theta\in\mathcal{S}_r$, count rejections $X_\theta$ of $v_r$,
    added rejections $N_\theta^{+}$, and removed rejections $N_\theta^{-}$
  \State $B^{\mathrm{pair}}_\theta\gets b_\theta+U(N^{+}_\theta;m_r,\delta_r/4)
    -L(N^{-}_\theta;m_r,\delta_r/4)$
  \State $B^{\mathrm{dir}}_\theta\gets U(X_\theta;m_r,\delta_r/2)$
  \State $B_\theta\gets\min(B^{\mathrm{pair}}_\theta,B^{\mathrm{dir}}_\theta)$
  \State If $B_\theta\le\alpha'$ for every $\theta\in\mathcal{S}_r$, deploy $v_r$ for
    every study whose data arrive from now on. Otherwise, retain the current
    version.
  \State Report the deployment decision and the scenario with the largest $B_\theta$ to the
    proposer
\EndFor
\end{algorithmic}
\end{algorithm}

All inputs, including the levels of the reference bounds
(Corollary~\ref{cor:lifetime-ref}) and the law and numbers of the random
scenarios (Proposition~\ref{prop:random}), are fixed before the working life of
the analyst begins. Line 2 puts the reference analysis, or a version that
rejects only where it rejects (Corollary~\ref{cor:domination}), into service only
if the reference bounds are at most $\alpha'$ on the initial scenario set, as
Corollary~\ref{cor:every-decision} requires. Otherwise the analyst defers every
study until its first update passes validation. Any other first version is
proposed as update 1 and validated like every later update, so it needs no
share outside $\delta_0+\sum_r\delta_r$. The reference bounds come from
Theorem~\ref{thm:cont-ref} on Gaussian null scenarios and from the one-time
validation of the reference analysis on scenarios built from real studies
(Corollary~\ref{cor:lifetime-ref}).

Lines 4 to 7 are the freeze of section~\ref{sec:lifetime}. Line 4 fixes the
updated version as a complete specification, including model, instructions,
tools, code and thresholds. Line 5 is the extension of
section~\ref{sec:extension}, which the validator decides, and line 6 fixes the
scenarios of the round. Line 7 fixes the share and the number of simulated
datasets. Line 8 draws the random scenarios after the freeze, independently of
everything before, and they serve this round only
(Proposition~\ref{prop:random}). Only then does line 9 draw the simulated
datasets, with randomness that is new and independent of everything the proposer
has seen. These lines deliver conditions (a) and (b) of
Theorem~\ref{thm:lifetime}. The share is indexed by the attempt $r$, so an
update that fails validation spends its share as well, and the shares sum to
$\delta-\delta_0$, which is condition (c).

Lines 12 to 15 compute the bounds of Theorem~\ref{thm:relative} and decide by
the intersection-union rule of Theorem~\ref{thm:iut}, with no correction across
scenarios. The number $m_r$ may grow with $r$ to keep the probability that a
valid update passes high (Proposition~\ref{prop:power}). A round never stops
early, because the exact binomial bounds assume a fixed number of datasets, and
optional stopping would invalidate them.

Line 15 deploys a version that passes validation for the studies whose data
arrive after the deployment decision, which is condition (ii) of
Corollary~\ref{cor:every-decision}. Line 16 returns the deployment decision and the scenario
with the largest bound. Because the next round draws new datasets, the proposer
may learn from the deployment decision, from the bounds and even from the simulated datasets
of the round.

Changes in rows 1 to 5 of Table~\ref{tab:modifications} may be deployed without
lines 7 to 15 and spend no share, whereas every other change goes through
validation.

\FloatBarrier

\section{Example 1, a bioequivalence analysis agent}
\label{sec:exp-be}

Both examples ask one question. While a language model keeps updating an AI
analyst to make it more useful, which deployment methods keep the Type I error
rate of the analyst at the chosen level, and which of them still deploy useful
updates? This section presents the first example, an agent that analyzes
bioequivalence studies. It states the hypothesis test that the agent decides,
how the agent reaches its decision, what an update can change and how studies
are simulated. It then gives the design that both examples share, the metrics
and the results. The second example, an analyst of product experiments,
follows in section~\ref{sec:exp-launch}.

\subsection{The hypothesis test}
\label{sec:exp-be-test}

A bioequivalence study is one part of the evidence on which a generic drug is
approved. It tests whether the generic, the test product, delivers the same
exposure as the reference product. Healthy volunteers receive both
formulations in a crossover study, in which each subject receives both in
different periods. In a replicate study the reference, or both formulations,
are given twice. Exposure is summarized by the area under the concentration
curve (AUC) and by the peak concentration (Cmax). With $\mu$ the log-ratio of
the geometric mean exposures of test and reference, as in
section~\ref{sec:framework}, the null hypothesis $H_0$ is that $|\mu|\ge\Delta$
with $\Delta=\log1.25$, so that the true ratio lies at or beyond 80.00 or
125.00 percent and the products are not equivalent. The alternative $H_1$ is
that $|\mu|<\Delta$. The products are declared equivalent, and $H_0$ is
rejected, when the 90 percent confidence interval of the ratio lies within
80.00 to 125.00 percent \cite{si:ema2010bioequivalence,si:ich2024m13a}. This is
the two one-sided tests at level 0.05. A Type I error is a false conclusion of
equivalence, and its probability is known in bioequivalence as the consumer
risk.

A drug is highly variable when the within-subject coefficient of variation of
its reference formulation exceeds 30 percent. For such a drug the EMA widens
the limits for Cmax to $\exp(\pm0.760\,s_{wR})$, where $s_{wR}$ is the
observed within-reference standard deviation on the log scale. The widening
stops at 69.84 and 143.19 percent, which are reached at a coefficient of
variation of 50 percent, and the point estimate of the ratio must still lie
within 80.00 to 125.00 percent \cite{si:ema2010bioequivalence}. For a highly
variable drug the null hypothesis uses the limits that the true
within-reference standard deviation implies, as
section~\ref{sec:exp-be-scenarios} states. The EMA's average bioequivalence
with expanding limits (ABEL) and the FDA's reference-scaled average
bioequivalence (RSABE) switch to scaled limits when the observed coefficient of
variation exceeds 30 percent. This switch is known to inflate the Type I error
rate when the true variability lies near 30 percent
\cite{si:labes2016inflation,si:wonnemann2015inflation,si:munoz2016consumer}.

\subsection{How the agent decides}

We built the agent as a statistical-agent harness for planning and analyzing
bioequivalence studies. Claude Code skills read written playbooks, whose
decision tables name the next command, and deterministic Python and R tools
compute and record every number. The software that computes the results calls
no language model. Its standard analysis is unscaled average bioequivalence for
a $2\times2$ crossover. A set of extended methods, still under development,
also covers parallel and replicate designs, with the FDA's RSABE and its switch
to unscaled average bioequivalence, the EMA's ABEL by Method A, the FDA's
criteria for narrow-therapeutic-index drugs and the EMA's narrowed limits.
These methods are checked against the EMA's two published replicate datasets
and against independent results computed in base R. The bioequivalence agent is available at
\url{https://be.statrust.ai/}.
The simplified version used in our experiments
will be released with the study code at
\url{https://github.com/xiaojing777/lifetime-valid-ai-analysts}.

For each study the agent takes the same steps. It prepares the data with a
prespecified analysis-set rule, complete case or available case, and records
every omitted record with its reason. It flags outlying subjects with a
prespecified rule, by default the rule of replicateBE
\cite{si:schutz2022replicatebe}, which flags a subject when one of its
externally studentized residuals in the outlier model of the method lies beyond
Tukey's hinges by more than twice the distance between the hinges. A flagged
subject is never removed, and the result without the flagged subjects is
reported beside the primary result. The agent then applies the method of its
configuration. It decides on the bounds of the confidence interval in percent,
rounded half up to two decimals, and reports the unrounded result of every
criterion beside the rounded one. In this experiment every version decides on
the unrounded results, so that the reference analysis is exactly the two
one-sided tests.

The analyst of this example is the analysis part of the agent, run with a
configuration that the experiment adds. The configuration fixes every setting
that can change a decision, and the deterministic tools apply it to the data of
each study. A rejection of $H_0$ is a conclusion of equivalence. The language
model that drives the analysis chooses the next command from the decision
tables, and
every path that the playbooks allow ends either in the decision that the tools
compute under the configuration or in a deferral to a human. The only allowed
path that can reject is therefore the one that applies the configured decision
rule, and the any-path rule of the version is this decision rule
(Proposition~\ref{prop:envelope}). The validator runs it directly, so
validation needs no calls to a language model. The analyst defers to a human
any study for which the agent does not compute the configured combination of
settings. Section~\ref{sec:exp-be-agent} checks on simulated studies that the
agent reaches these decisions.

\subsection{Reference analysis, settings and new decision rules}
\label{sec:exp-be-settings}

The reference analysis $v_0$ is the agent's standard unscaled average
bioequivalence, the two one-sided tests at $\alpha_0=0.030$ on the
complete-case set, with outliers flagged but not removed and no deferral rule. On
the full replicate it uses the intra-subject contrast of each subject, the mean
test minus the mean reference log response, fitted on sequence. Equivalently,
it checks there whether the 94 percent confidence interval lies within 80.00 to
125.00 percent. The agent refuses it on the partial replicate, where $v_0$
therefore defers and never rejects. The reference level is lower than in Example 2
because the subject contrasts of real studies can be skewed. At 0.040 the two
one-sided tests reach a Type I error rate of about 0.052 on resampled studies of
dataset 5 with 24 subjects, above $\alpha'$, while at 0.030 they stay at or below
0.042 in every scenario of the base set.

On the Gaussian scenarios of section~\ref{sec:exp-be-scenarios}, the reference
bound is $b_\theta=\alpha_0$ in the full replicate and $b_\theta=0$ in the
partial replicate. These scenarios have equal within-subject variances for test
and reference, a subject effect common to both formulations and dropout that
does not depend on the responses. In the full replicate the contrast analysis is
then an exact $t$ test, and the proof of Theorem~\ref{thm:cont-ref} carries
over. In the partial replicate every rejection of an update is an added
rejection, so the combined bound of Algorithm~\ref{alg:protocol} equals the
direct bound.

On scenarios built from real studies, $b_\theta=0$ in the partial replicate for
the same reason. In the full replicate their reference bounds come from a
one-time validation at $\delta_0=0.001$ (Corollary~\ref{cor:lifetime-ref}).
Before the first update of each run, the validator simulates $10^6$ datasets in
every such scenario of the base set, 252 in all, and divides $\delta_0$ equally
among them. The extension adds only Gaussian scenarios, whose bounds are known. The
experiment then confirms that every reference bound of that run on the base set
is at most $\alpha'$, so that the starting version is valid there (Corollary~\ref{cor:every-decision}). A run whose check fails
starts from the version that defers every study
(Algorithm~\ref{alg:protocol}, line 2). Once, on the evaluation stream, the
experiment also confirms that the reference analysis counts as within $\alpha'$
on the full validation set, by the rule of section~\ref{sec:exp-metrics}.

An update may change one of six settings.
\begin{itemize}
\item Method. Unscaled average bioequivalence, or the EMA's ABEL by Method A,
  which widens the limits when the observed within-reference coefficient of
  variation exceeds 30 percent.
\item Analysis set. Complete case or available case.
\item Role of the outlier analysis. Reported beside the primary result, as in
  the agent, or reported instead of the primary result when it concludes
  equivalence and the primary result does not.
\item Level of the two one-sided tests. 0.030, 0.035, 0.040, 0.045 or 0.050,
  applied to every interval that the method uses.
\item Deferral by variability. None, or deferral when the observed
  within-reference coefficient of variation exceeds 0.30, 0.40 or 0.50.
\item Deferral of flagged files. On or off.
\end{itemize}
The menu of settings holds 320 configurations, and the reference analysis and
the starting version are two of them. The starting version is the reference
analysis with two deferral rules. It sends to a human every study whose
observed within-reference coefficient of variation exceeds 0.30 and every file
with a data-quality flag. Every study is treated as a Cmax study, so the
widening applies whenever the method calls for it.

An update may instead add a new decision rule, which the developer writes as a
Python function against a fixed interface (section~\ref{sec:exp-common}). A new
decision rule acts on top of the settings, for example as a new outlier
criterion or as a choice of model that depends on the data. It receives a batch
of simulated studies as arrays of subjects, sequences, periods, formulations,
log responses and data-quality flags, together with the flags of the agent's
prespecified outlier rule. It may call the estimators of the agent, through the
reimplementation described below, on any subset of the subjects of each study,
with any method, analysis set and level of the menu, and it may compute its own
intervals from quantiles of the $t$ distribution. It returns one of three
decisions for every study, equivalence shown, equivalence not shown, or
deferral to a human.

\subsection{Scenarios, workload and simulation sizes}
\label{sec:exp-be-scenarios}

Validation scenarios are built from ten reference datasets of the R package
replicateBE \cite{si:schutz2022replicatebe,si:schutz2020reference}. These are its
reference datasets 1, 2, 4, 5, 10, 11, 16, 23, 24 and 26, the ten that come
from real studies. They come in several replicate designs. Every scenario is
simulated in one of two replicate designs that the agent analyzes, the
four-period full replicate TRTR/RTRT and the three-period partial replicate
TRR/RTR/RRT. The agent refuses some configurations in these designs,
standard average bioequivalence on the partial replicate and the available-case
set for standard average bioequivalence on the full replicate, and the analyst
then defers the study to a human. A simulated subject is a real subject drawn
with replacement, whose centered responses are placed in the
periods of its simulated sequence. Each published dataset is simulated in one design, the
full replicate when its subjects have two test and two reference responses and
the partial replicate otherwise. The $k$th test and $k$th reference response of
a resampled subject fill its $k$th test and $k$th reference period, and the
partial replicate uses the first test response. A subject is eligible for
resampling when it has the observations that the design of its dataset needs,
two test and two reference responses for the full replicate, and one test and
two reference responses for the partial replicate. The design of each dataset and its
number of eligible subjects were fixed before the experiment. Before resampling, the log responses are centered within each
sequence and period. This removes the period, sequence and formulation effects
and keeps the spread between and within subjects.

The resampling is a smoothed bootstrap. It adds independent Gaussian noise to
every log response, with standard deviation one tenth of the residual standard
deviation of the EMA's Method A model fitted to the original dataset. The
noise removes exact ties between resampled subjects. The test responses are
then shifted by the true log-ratio of the scenario.

A scenario has $n\in\{24,36,48\}$ subjects, split evenly over the sequences,
and dropout of 0, 10 or 20 percent. A subject drops out with this probability,
independently of its responses, after a period drawn uniformly from all periods
but the last, and its later periods are missing. In every simulated study, null
or not, the file carries a data-quality flag with probability 0.15,
independently of the responses.

The true within-reference variance $\sigma_{wR}^2$ of a scenario is that of its
resampling distribution. It equals half the mean squared difference between
the two centered reference responses of a subject, averaged over the subjects
who receive the reference twice, plus the variance of the added noise. The
true coefficient of variation is $\{\exp(\sigma_{wR}^2)-1\}^{1/2}$. The true
geometric mean ratio lies at the regulatory limit for this variability, on
either side, or 0.05 beyond it on the log scale. The limit is 80.00 or 125.00
percent when the true coefficient of variation is at most 30 percent.
Otherwise it is the EMA's widened limit $\exp(\pm0.760\,\sigma_{wR})$, which
stays at 69.84 or 143.19 percent from a coefficient of variation of 50 percent
upward. Every scenario therefore belongs to the null hypothesis that the true
ratio lies at or beyond the regulatory limit for the true variability. Validation
holds every version to this null hypothesis. The
base set, which is the initial scenario set $\mathcal{S}_0$ of
Algorithm~\ref{alg:protocol}, crosses every dataset, in its design, with three
sample sizes, three dropout rates and these four ratios. It holds 36 scenarios
per dataset, 360 in all.

Validation adds scenarios when an update reads information that the base set
does not vary finely enough (section~\ref{sec:extension}). An update reads the
observed variability when it uses reference-scaled limits or deferral by the
coefficient of variation. Validation then adds Gaussian scenarios with true
within-reference coefficients of variation of 0.25, 0.28, 0.30, 0.32, 0.35,
0.40 and 0.50. They use the four-period full replicate TRTR/RTRT and the
three-period partial replicate TRR/RTR/RRT, with the sample sizes, dropout
rates and four placements of the true ratio of the base set. Their errors have
equal variances for test and reference. Their subject effects are common to
both formulations, with standard deviation 0.5 on the log scale, and they have
no period or sequence effects. These are 504 scenarios. Because the starting
version defers by the observed coefficient of variation, every update that
keeps this deferral is validated on these scenarios.

The available-case set reads incomplete records, whose dropout the base set
already varies, so it adds no scenarios. A new decision rule reads the data and can compute
the observed variability, so every new decision rule is validated on the
Gaussian scenarios above. Such a rule also sees the allocation of subjects to sequences and
the periods, which the scenarios above keep balanced and free of period
effects. Every new decision rule is therefore also validated on 72 Gaussian
scenarios that cross the two designs, the three sample sizes with two thirds of
the subjects in the first sequence and the rest split evenly over the others,
the three dropout rates and the four placements of the true ratio, at a true
within-reference coefficient of variation of 0.30, with log responses that rise
by 0.05 from each period to the next. Once added, a scenario stays in the set.
The full validation set, the base set together with both kinds of added
scenarios, holds 936 scenarios.

Off-grid scenarios, which no deployment method ever simulates, are built in the same way and
form five families. The first uses five further reference datasets of
replicateBE, numbers 7, 8, 22, 25 and 28, which validation never uses. It keeps the sample sizes and dropout
rates of the base set and places the true ratio at either limit. The authors of
the package simulated these five, since its ten datasets from real studies all
serve validation. The second uses the ten datasets with 30, 42 or 60 subjects,
split as evenly as the sequences allow. The third uses them with dropout of 5,
15 or 25 percent. The fourth places the true ratio 0.025 beyond either limit on
the log scale, between the two placements of the base set. The fifth applies to
these five datasets dropout that depends on the first-period response. Only
subjects whose first-period log response lies below the 20th or above the 80th
percentile of their sequence drop out, each with probability equal to the
dropout rate divided by the share of its sequence in these tails, so that the
overall rates are 10 and 20 percent. Factors that a family does not change keep the values
of the base set, and the true ratio lies at either limit unless the family sets
it. These families hold 690 scenarios, and each of them differs from every
validation scenario in at least one factor.

Power is measured on a workload of truly equivalent studies. It crosses the ten
datasets, the three sample sizes and the three dropout rates of the base set
with true ratios of 0.95 and 1.00, which gives 180 cells of equal weight. Power
is the probability that the deployed version concludes equivalence without
deferral, averaged over the cells. The usefulness metric of the performance report
is the share of workload studies that the agent decides without a human,
together with its power. It is estimated each round from 10,000 new workload
studies, and the experiment reports it alongside the five metrics. The
performance report and the improvement test draw each workload study from a cell
chosen uniformly at random.

Continual validation draws \SimBE{} datasets per scenario per update. The static
benchmark stores \SimBE{} datasets per scenario of the base set, and the
improvement test draws \SimBE{} new workload studies per round. Learn then
Test tests the 320 configurations once on \SimBE{} datasets per scenario, after
a pilot of one tenth that size.

\subsection{Design shared by both examples}
\label{sec:exp-common}

A language model acts as the developer of the analyst. It is Claude Sonnet 5.5
(claude-sonnet-5-5), with its default sampling and the medium setting of its
reasoning effort. Before every
round the developer receives four inputs. These are the description of the
analyst's settings, the current configuration, a performance report and
feedback on its earlier updates. The performance report gives the usefulness
metric of the example and the power of the deployed version on simulated
alternatives. The feedback states for every earlier update whether it was
deployed, together with the bound or estimate on which the deployment method
decided.

The developer proposes one update per round. An update either changes one
setting or adds a new decision rule, a Python function written against a
fixed, typed interface. Such a rule is deterministic, makes no network calls
and runs in a sandbox. The instructions to the developer are identical under every
deployment method and never name the method. They describe the analyst, its
settings and the interface, give the usefulness metric of the example as the
goal, and ask for exactly one update in a fixed machine-readable format. Each
round is one call to the model with the full history of the run. For every
deployment method and in each example, the experiment performs \NumRuns{}
independent runs of \NumRounds{} rounds. Every run starts from a version that
needs no validation of its own. In Example 2 it is the reference analysis. In
Example 1 it is the reference analysis with two deferral rules
(section~\ref{sec:exp-be-settings}). This version rejects only where the
reference analysis rejects, so its Type I error rate is at most that of the
reference analysis in every scenario (Corollary~\ref{cor:domination}).

A version consists of a configuration of the settings and at most
one new decision rule. Without such a rule the configured analysis decides,
and with one the rule decides. A rule may call the configured analysis and any
other analysis of the menu. A new rule replaces the earlier one, and an update
may also remove the rule. The interface is vectorized, so a rule receives a
batch of simulated datasets and returns one decision for each. A rule must
decide a batch of 10,000 datasets within ten seconds on one processor core,
including the analyses it requests.
An update is malformed when it cannot be parsed, misses this time limit, or
returns different decisions when it is run twice on the same batch. A
malformed update is not deployed under any method, the feedback names the
reason, and its round still counts.

The same five deployment methods run in both examples, and all of them receive
the same developer.
\begin{itemize}
\item No validation. Every well-formed update is deployed without a
  check.
\item Static benchmark. Before the first update, the experiment simulates one set
  of null datasets, of the size that validation uses, and reuses it for every
  update. An update is deployed when the one-sided 95 percent Clopper and
  Pearson upper bound on its Type I error rate is at most $\alpha'$ in every
  scenario of this set. This is today's practice of re-checking each version on
  a stored benchmark at a fixed confidence level.
\item Improvement test. An update is deployed when a one-sided lower confidence
  bound on its gain in power over the current version is positive. The error
  budget is spread over updates, as in recent checks for self-improving agents
  \cite{si:sengupta2026selfevolving,si:sun2026reuse}. Power is estimated on new
  simulated alternatives, and the Type I error rate is not checked.
\item Learn then Test \cite{si:angelopoulos2025learn}. Before the first update,
  every configuration in the menu of settings is tested once, on
  null datasets simulated once, with familywise error control at $\delta$ by
  fixed-sequence testing. A proposed update is deployed only if its
  configuration lies in the validated set. Updates outside the menu, such as
  new decision rules, are never deployed.
\item Continual validation. Algorithm~\ref{alg:protocol} runs with the lifetime
  risk $\delta=0.005$, the validation level $\alpha'=0.0452$, the reference
  level $\alpha_0$ of the example (0.030 in Example 1 and 0.040 in Example 2) and the shares $\delta_r=6(\delta-\delta_0)/(\pi^2r^2)$.
  In both examples the reference analysis receives $\delta_0=0.001$.
\end{itemize}

The static benchmark covers the base set, which holds the validation scenarios in
force before the first update. It does not grow with later updates, because a
stored benchmark is fixed before the updates it will judge are known.

In round $r$ the improvement test runs the update and the current version on
the same $m$ new alternatives. It counts the alternatives on which only the
update rejects the null hypothesis without deferral, $G^{+}$, and those on
which only the current version does, $G^{-}$. The update is deployed when
$L(G^{+};m,\epsilon_r/2)-U(G^{-};m,\epsilon_r/2)>0$. Here $U$ and $L$ are the
exact one-sided Clopper and Pearson bounds, and $\epsilon_r=6\delta/(\pi^2r^2)$.
The probability that a run ever deploys an update that does not raise power is
then at most $\delta$.

Learn then Test tests every configuration on the base set and on the added
scenarios that its own settings call for, with as many null datasets per
scenario as validation draws. The null hypothesis of a configuration is that
its Type I error rate exceeds $\alpha'$ in some of these scenarios. In each
scenario, the exact $p$ value is the probability that a binomial count with
rate $\alpha'$ is at most the observed number of rejections, and the $p$ value
of the configuration is the largest of these over its scenarios. The tests run
in the order of the largest Type I error rate of each configuration over its
scenarios, from the smallest upward, estimated on a separate pilot simulation
with one tenth as many datasets per scenario, so the order does not depend on
the datasets of the test. Fixed-sequence testing then validates the
configurations in this order at level $\delta$ and stops at the first
configuration that fails.

Continual validation freezes each update, draws new null datasets in every
scenario in force, and runs the update next to the reference analysis on the
same datasets. It deploys the update only if every combined bound is at most
$\alpha'$, and every update goes through this check. The validator adds
scenarios when an update reads information that the scenarios in force do not
vary (section~\ref{sec:extension}). The runs below draw no random scenarios
(Proposition~\ref{prop:random}), so the off-grid scenarios measure how far
validation on fixed scenarios reaches, and section~\ref{sec:exp-be-random}
adds random scenarios. By Theorem~\ref{thm:lifetime}, the
probability that a run ever deploys an update above $\alpha'$ in a scenario on
which it was validated is at most $\delta$. An update whose Type I error rate
lies far enough below $\alpha'$ passes validation with high probability
(Proposition~\ref{prop:power}).

The feedback reports what each method computes. Under the static benchmark and
under continual validation, it gives the largest upper bound of the update and
the scenario that attains it. Under the improvement test it gives the lower
bound on the gain in power. Under Learn then Test it gives the largest upper
bound at level $\delta$ that the configuration reached in the one-time test, or
notes that the update lies outside the menu. When every update is deployed, the
feedback states only that the update was deployed.

The full validation set of an example holds every validation scenario,
including every scenario that the extension can add. The oracle serves only as a
reference line, because it needs the true Type I
error rate and the true power of every update, which the evaluation stream of
section~\ref{sec:exp-metrics} estimates. It deploys an update when all the
estimated rates of the update on the full validation set are at most $\alpha'$
and its estimated power exceeds that of the current version. It does not deploy
an update that counts as indeterminate (section~\ref{sec:exp-metrics}). The oracle runs with the same developer,
also for \NumRuns{} runs of \NumRounds{} rounds, and its feedback reports the
estimates on which it judged each update.

Every run draws its own random streams for each simulation it uses. These
include the stored set of the static benchmark, the one-time test of Learn then
Test and the one-time validation of the reference analysis. They also include
the new datasets of every round, so the runs are independent.

\subsection{Metrics and analysis plan}
\label{sec:exp-metrics}

Five metrics are fixed in advance and are identical in both examples. Metric
(a) is the largest true Type I error rate of the deployed version over the full
validation set, after each round. Metric (b) is the share of runs that ever
deployed a version above $\alpha=0.05$, the level that every decision must hold,
with a Wilson 95 percent interval, together with the same share for $\alpha'$. Metric
(c) is the true power of the deployed version after each round, together with
the ratio of its mean over runs to the mean of the oracle after the same round.
Metric (d) is the largest true Type I error rate of the version that continual
validation deployed, on off-grid null scenarios that no deployment method ever simulates, after
each round.
Metric (e) is the number of simulated datasets and the computing time of each
validation.

True rates are estimated in two stages on an evaluation stream that no deployment method
uses, with common random numbers across versions. The first stage draws 100,000
datasets per scenario for every version. The scenario with the largest first-stage
estimate, and every scenario whose first-stage interval contains $\alpha'$, are
then estimated again from \EvalSims{} datasets on an independent stream. The
second stage removes the upward bias of a maximum taken over many scenarios.
Off-grid scenarios go through the same two stages, with their own largest
first-stage estimate. The rate of a scenario is its second estimate where it
has one and its first otherwise, and metrics (a) and (d) take the largest of
these rates. Versions that recur across runs, such as configurations of the
menu, are evaluated once. True power is estimated from 100,000
datasets in each cell of the power workload of each example
(sections~\ref{sec:exp-be-scenarios} and~\ref{sec:exp-launch}).

A version counts as above $\alpha'$ when, in some scenario of the full
validation set, its interval lies above $\alpha'$. It counts as within
$\alpha'$ when its intervals lie below $\alpha'$ in every scenario of that set,
and as indeterminate otherwise. The same rule with $\alpha$ in place of
$\alpha'$ decides whether a version counts as above $\alpha$. The intervals of
both stages are exact
two-sided Clopper and Pearson intervals at level $0.01/(2K)$. Here $K$ is fixed
in advance as the number of scenarios, off-grid scenarios included, times the
largest possible number of evaluated versions, so that all intervals hold
together with probability at least 0.99. The first stage selects scenarios with
these intervals, and a scenario estimated twice is judged by its second
interval. Indeterminate versions are reported separately in metric (b) and are
never counted as within $\alpha'$.

The main comparison, fixed in advance, reads metric (b) together with metric
(c). With \NumRuns{} runs, a method without any failure has a Wilson upper limit of
about 16 percent, so metric (b) separates methods that fail often from methods that
rarely fail. The main text reports metric (b) at $\alpha$ as a percentage of
runs and power as the ratio of means after the last round. For no validation and for
continual validation it reports the largest value of metric (a) over all rounds
and runs, which is the largest Type I error rate of any deployed version over
the full validation set. For continual validation it also reports the largest
value of metric (d) over all rounds and runs. Theorem~\ref{thm:lifetime} gives the lifetime
risk of continual validation, at most $\delta$, so the experiment does not
estimate it. Metric (b) counts every scenario of the full validation
set, including scenarios on which a deployed version was never validated. For
continual validation the experiment therefore also reports metric (b) on the
scenarios on which each deployed version was validated, the event that
Theorem~\ref{thm:lifetime} bounds by $\delta$. For each method the experiment
also reports how many of the updates it received were truly above $\alpha'$,
and by how much each of them exceeded $\alpha'$. It reports the number of
indeterminate updates beside them. For every deployed version above $\alpha'$,
it names the scenarios in which the version is above, and whether the method
simulated them, so that each failure is traced to a scenario the method checked
or to one it never simulated. Metrics (a) and (d) are summarized
round by round by the median and the range over runs, and metric (c) by the
mean. The run with the highest metric
(a) after the last round is shown separately for each method.

\subsection{Implementation requirements}

Validation of one update analyzes \SimBE{} simulated studies in each of up to
936 scenarios. The validator therefore calls the estimators through a vectorized
reimplementation that runs on GPUs, and every simulation of the experiment, the
evaluation stream included, uses it. The reimplementation decides on unrounded
bounds, as every version of the experiment does. It reproduces the published
Method A results for the two EMA reference datasets, the point estimate, the 90
percent interval and the coefficient of variation to two decimals in percent
and the two outlying subjects of data set I, and the results for all 30
reference datasets of replicateBE. On simulated studies of both designs and both
analysis sets it agrees with a direct least-squares computation for each study
to within $10^{-9}$, and every configuration of the menu reaches the same
decision for every study.

\subsection{Results}
\label{sec:exp-be-results}

Table~\ref{tab:exp-be} summarizes the five deployment methods and the oracle. Under
every method the developer followed a similar path. In
its first rounds it raised the level of the two one-sided tests toward the
regulatory 0.05, switched to ABEL, or removed a deferral rule, and in later
rounds it wrote new decision rules. The twenty runs of each method followed twenty different
sequences of updates. Under continual validation the developer proposed \BEcvDistinct{}
distinct versions, \BEcvCodeShare{} percent of its updates were new decision
rules, and \BEcvRunsFinalCode{} of the twenty runs ended with such a rule. Under
Learn then Test it proposed \BEltCode{} new decision rules, none of which that
method could deploy.

\begin{table}[htbp]
\centering
\setlength{\tabcolsep}{4pt}
\caption{Example 1. For each method, the runs that ever deployed a version above
$\alpha=0.05$ and above $\alpha'$, with Wilson 95 percent intervals in percent,
the largest true Type I error rate of any deployed version, the mean true power
after the last round and its share of the oracle's, the updates received and how
many of them created a version above $\alpha$, and for the methods that simulate
null datasets the datasets of one validation in millions and its median time in
seconds on one GPU.}
\label{tab:exp-be}
\scriptsize
\begin{adjustbox}{max width=\textwidth}\begin{tabular}{lcccccccrr}
\toprule
Method & Runs above $\alpha$ & Runs above $\alpha'$ & Largest rate & Power & \% of oracle & Updates & Above $\alpha$ & Datasets (millions) & Seconds \\
\midrule
No validation & 20/20 (84 to 100) & 20/20 (84 to 100) & 0.098 & 0.86 & 129 & 200 & 188 &  &  \\
Static benchmark & 4/20 (8 to 42) & 6/20 (15 to 52) & 0.060 & 0.60 & 90 & 200 & 68 & 180 & 569 \\
Improvement test & 20/20 (84 to 100) & 20/20 (84 to 100) & 0.098 & 0.90 & 134 & 200 & 188 &  &  \\
Learn then Test & 0/20 (0 to 16) & 0/20 (0 to 16) & 0.045 & 0.35 & 52 & 197 & 36 &  &  \\
Continual validation & 0/20 (0 to 16) & 0/20 (0 to 16) & 0.045 & 0.55 & 82 & 198 & 32 & 447 & 1101 \\
Oracle & 0/20 (0 to 16) & 0/20 (0 to 16) & 0.044 & 0.67 &  & 196 & 40 &  &  \\
\bottomrule
\end{tabular}\end{adjustbox}
\end{table}

One run of continual validation shows how the developer used the feedback. Its
first update raised the threshold of the deferral by variability to 0.5, which
passed. Validation then refused ABEL with an upper bound of 0.0600 in resampled
studies of dataset 11. The developer wrote that the inflation of ABEL comes from the
switch between the standard and the widened limits near a coefficient of variation
of 30 to 40 percent, and its third update was a rule that applies ABEL but defers
the studies whose observed coefficient of variation lies between 0.25 and 0.45,
which passed with a bound of 0.0375. Two later increases of the level, to 0.04
and 0.035, were refused at 0.0558 and 0.0504. The developer then narrowed the
deferral zone of its rule to 0.28 to 0.40, which passed with a bound of 0.0452,
and its last update kept the current settings, with the comment that it saw
no safe change left that would add decisions. Its
final version reached a true power of 0.79 with a largest true Type I error rate
of 0.043.

No validation and the improvement test deployed a version above 0.05 in every
run, in most runs within the first two rounds, when the developer switched
to ABEL or raised the level. These versions exceeded 0.05 in the Gaussian scenarios
with a coefficient of variation from 0.28 to 0.40, where the switch of ABEL inflates
the error, in the scenarios added for new decision rules, and in resampled studies of
datasets 10 and 11. The static benchmark refused the updates that exceeded 0.05 in
datasets 10 and 11, because its stored set contains these datasets. It still deployed a version above 0.05 in
\BEfixedEver{} percent of runs and above $\alpha'$ in \BEfixedEverPrime{} percent.
Every scenario in which such a version exceeded 0.05 was a Gaussian scenario at a
coefficient of variation of 0.30, with or without unequal allocation and period
effects, and the stored base set contains none of them. Continual validation adds
exactly these scenarios when an update reads the observed variability or adds
a new decision rule.

Continual validation deployed no version above $\alpha'$ in any run, neither over
the full validation set nor on the scenarios on which each version was validated,
which is the event that Theorem~\ref{thm:lifetime} bounds. It judged \BEcvJudged{}
updates and deployed \BEcvDeployedN{}. Of the \BEcvRejectedN{} it refused, the
evaluation stream placed \BEcvRejAbove{} above $\alpha'$ and \BEcvRejIndet{} within
the Monte Carlo resolution of $\alpha'$, so validation refused no update that was
clearly valid. The versions with the highest power, such as ABEL at a lowered level for
highly variable studies, have rates within a few thousandths of $\alpha'$, and
certifying them takes more simulated datasets in later rounds, where the shares
$\delta_r$ are small. With 100{,}000 datasets per scenario, continual validation
refused \BEcvRejValidSmall{} valid updates and attained \BEcvPowerShareSmall{}
percent of the oracle's power, and with \SimBE{} it attained \BEcvPowerShare{}
percent. On the 690 off-grid scenarios its deployed versions stayed at or below
\BEoffGridMax. Every off-grid rate above 0.05 came from dropout that depends on
the first-period response in datasets 22 and 28, which validation never uses, a mechanism that
the validation scenarios of this example do
not contain. Each validation simulated
about \BEcvDatasetsPerVal{} million studies and took a median of
\BEcvSecondsPerVal{} seconds on one GPU, and the one-time validation of the
reference analysis took a median of \BErefSeconds{} seconds per run.

Two open-weight models, Qwen 3.8 with 27 billion parameters and the
32-billion-parameter distilled version of DeepSeek-R1, served as developers through Purdue's GenAI Studio under three of the
methods, with the same instructions and 100{,}000 datasets per scenario. Their
replies were malformed in \OWqwenMalformed{} and \OWdeepseekMalformed{} percent of
rounds, and none of the methods deployed a malformed update. Table~\ref{tab:exp-be-open} shows
that no validation and the improvement test again deployed versions above
0.05 in every run, while continual validation kept every deployed version within
the level.

\begin{table}[htbp]
\centering
\caption{Example 1 with open-weight developer models and 100{,}000 datasets per
scenario. Columns as in Table~\ref{tab:exp-be}, for the methods run with each
model.}
\label{tab:exp-be-open}
\scriptsize
\begin{adjustbox}{max width=\textwidth}\begin{tabular}{llccccrr}
\toprule
Developer & Method & Runs above $\alpha$ & Runs above $\alpha'$ & Largest rate & Power & Updates & Above $\alpha$ \\
\midrule
Qwen 3.8 27B & No validation & 20/20 (84 to 100) & 20/20 (84 to 100) & 0.236 & 0.96 & 183 & 152 \\
 & Improvement test & 20/20 (84 to 100) & 20/20 (84 to 100) & 0.301 & 0.96 & 181 & 143 \\
 & Continual validation & 0/20 (0 to 16) & 0/20 (0 to 16) & 0.043 & 0.57 & 192 & 32 \\
DeepSeek-R1 32B & No validation & 20/20 (84 to 100) & 20/20 (84 to 100) & 0.098 & 0.74 & 147 & 104 \\
 & Improvement test & 20/20 (84 to 100) & 20/20 (84 to 100) & 0.098 & 0.74 & 150 & 104 \\
 & Continual validation & 0/20 (0 to 16) & 0/20 (0 to 16) & 0.043 & 0.26 & 175 & 49 \\
\bottomrule
\end{tabular}\end{adjustbox}
\end{table}

In a further experiment, every method received the same fixed list of ten
updates, written in advance, in the same order, and Fig.~\ref{fig:replay} shows which updates each method deployed and the
true rate of the version each update would have created. In Example 1 continual
validation refused the switch to ABEL, the level of 0.05 and the rule with scaled
limits above a coefficient of variation of 0.35, which would all have exceeded
0.05, and it deployed the rule for the partial replicate and the rule with scaled
limits above 0.40, whose rates stay within the level. The static benchmark deployed
the switch to ABEL, whose rate exceeds 0.05 only near a coefficient of variation of
0.30, outside its stored set. Learn then Test refused all three new decision rules. In Example 2 continual validation refused daily looks without sequential
testing and the Bayesian threshold of 0.95, and it deployed sequential testing,
post-stratification and the rule with two planned looks.

\begin{figure}[htbp]
\centering
\IfFileExists{fig_replay.pdf}{\includegraphics[width=\textwidth]{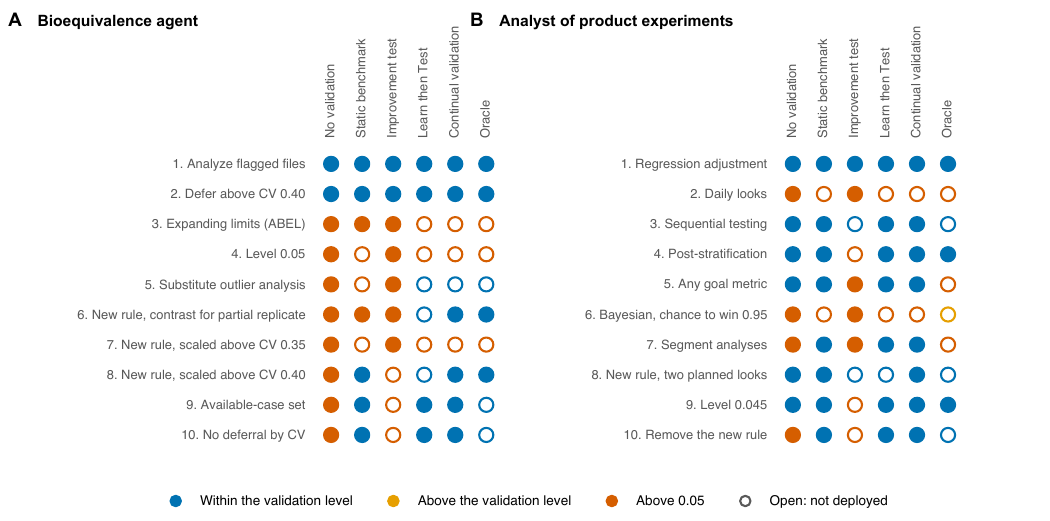}}{\figpending{The figure is generated from the runs with updates written in advance.}{4cm}}
\caption{Updates written in advance. The same ten updates, in the same order, judged by
every method in (\textit{A}) Example 1 and (\textit{B}) Example 2. A filled circle
marks an update the method deployed and an open circle one it did not. The color
gives the true Type I error rate of the version the update would create, within
$\alpha'$ or too close to $\alpha'$ to resolve (blue), above $\alpha'$
(orange), or also above $\alpha$ (vermilion).}
\label{fig:replay}
\end{figure}

\FloatBarrier

\subsection{Random scenarios cover the space between the grid values}
\label{sec:exp-be-random}

The runs above validated every update on fixed scenarios and added scenario
kinds when an update read new information (section~\ref{sec:extension}). The
third layer of continual validation draws new scenarios at random after every
freeze (section~\ref{sec:random}), so that no version can overfit the
scenarios that validation uses. We ran this layer in Example 1.

The distribution $\Pi$ was fixed before the first run, and its draws are
independent. With probability one half a draw resamples one of the ten
validation datasets, chosen uniformly, with the design and coefficient of
variation of that dataset. Otherwise it is a Gaussian model whose coefficient
of variation is uniform on 0.20 to 0.55 and whose design is the full or the
partial replicate with probability one half each. The number of volunteers is
uniform on 24, 30, 36, 42, 48, 54 and 60, volunteers drop out independently of
their responses at a rate uniform on 0 to 25 percent, and the true ratio lies
at the lower or the upper acceptance limit implied by the true variability or
beyond it by up to 0.05 on the log scale. The grid holds three of these sample
sizes, three dropout rates and two placements of the ratio, so most draws fall
between or beyond its values. The five datasets that validation never uses
stay outside $\Pi$, and the off-grid scenarios remain an independent test.
Dropout that depends on the response also stays outside $\Pi$, because it
moves the starting version of every run, the reference analysis with its two
deferral rules, above $\alpha$, to a true Type I error rate of 0.059 on
dataset 5 with 48 volunteers and 20 percent dropout.

In round $r$, after the freeze, continual validation with random scenarios
adds $K_r$ scenarios drawn from $\Pi$ on a random stream that the developer
never sees. A drawn scenario takes the reference bound where it is known
without simulation, zero on the partial replicate and $\alpha_0$ on a Gaussian
full replicate, and the trivial bound 1 otherwise, so that the direct bound
decides it. All scenarios of the round, fixed and drawn, must pass together.
With $\varepsilon_\Pi=0.05$ and $\omega=0.005$, $K_r$ is the smallest integer
with $(1-\varepsilon_\Pi)^{K_r}\le6\omega/(\pi^2r^2)$, which gives 113
scenarios in the first round and 203 in the tenth. By
Proposition~\ref{prop:random}, with probability at least $1-\delta-\omega=0.99$
every deployed update is then within $\alpha'$ on every fixed scenario of its
round and on all of $\Pi$ except a set of probability at most 0.05. Every drawn
scenario uses \SimBE{} simulated studies, as the fixed scenarios do.

In a second experiment with updates written in advance, continual validation on the fixed scenarios and
continual validation with random scenarios received the same ten updates in
the same order, each method in three runs with different validation streams
(Fig.~\ref{fig:random}\textit{A}). Six updates come from the runs above, a
deferral threshold of 0.50, two new decision rules that defer studies near
the switch to widened limits, the analysis of flagged files, the level 0.035
and the removal of the deferral threshold. The other four, the off-grid updates, copy the current
analysis at the sample sizes of the grid and switch to widened limits at level
0.05 without deferral at every other size, at sizes above 48 volunteers, or at
30 and 42 volunteers, the version that section~\ref{sec:random} describes.
Twenty further runs used continual validation with random scenarios and
Claude Sonnet 5.5 as developer, with the instructions and simulation sizes of
the runs above. The coverage of $\Pi$ by a deployed version was measured on 300
scenarios drawn from $\Pi$ on a separate evaluation stream, each with
100{,}000 simulated studies, as the share of scenarios in which its true Type I
error rate exceeds $\alpha'$.

Continual validation on the fixed scenarios deployed all four off-grid updates
in each of the three runs, because on the grid they decide exactly like the
version in use. With random scenarios it refused all twelve (Table~\ref{tab:exp-be-random}), each time on a
drawn scenario whose sample size lies outside the grid, with upper bounds from
0.055 to 0.088. The versions these updates create exceed
$\alpha'$ on at most 4 percent of $\Pi$, and the draws of every round found
that region. Both methods deployed the same five realistic updates in every run,
with nearly equal bounds, and both refused the level 0.035, whose excess the
grid itself detects. The versions that passed on the fixed scenarios reached a
true Type I error rate of 0.098 at sample sizes outside the grid, about twice
$\alpha$, while every version deployed with random scenarios stayed at or below
0.042 there (Figs.~\ref{fig:robust}\textit{A} and~\ref{fig:random}\textit{B}). The last versions of both methods
had the same power, 0.78, and the same largest rate on the validation
scenarios, 0.043, so the random scenarios removed only the behavior that the
grid cannot see. Under dropout that depends on the response, which lies
outside $\Pi$, the last versions of both methods reached 0.065, as in the runs
above.

The random scenarios cost nothing when the updates did not target the grid. None of the 69 versions that continual validation
deployed on the fixed scenarios in the runs above exceeded $\alpha'$ in any of
the 300 scenarios drawn from $\Pi$ for the evaluation, and their largest rate
there was 0.041. In the 20 runs with random scenarios, no run deployed a
version above $\alpha$, the largest rate on the validation scenarios was 0.045,
and the deployed versions reached \BErandPowerShare\ percent of the oracle's
power, against \BEcvPowerShare\ percent on the fixed scenarios, a difference in
final power of 0.021 with a Welch test $p$ value of 0.67
(Figs.~\ref{fig:robust}\textit{B} and~\ref{fig:random}\textit{C}). The fixed scenarios decided every one of the
85 updates that this method refused, and none of its 62 deployed versions
exceeded $\alpha'$ in a scenario drawn from $\Pi$ for the evaluation.

\begin{table}[htbp]
\centering
\caption{Random scenarios in Example 1. Updates written in advance, three runs per method with
the ten updates of Fig.~\ref{fig:random}\textit{A}. Live runs, 20 runs per method
with Claude Sonnet 5.5 as developer, the runs on fixed scenarios being those of
Table~\ref{tab:exp-be}. Rates are true Type I error rates of the deployed
versions, on the validation scenarios and at the sample sizes outside the grid
of the off-grid scenarios. The share of $\Pi$ is the largest, over the
deployed versions, share of 300 scenarios drawn from $\Pi$ on an evaluation
stream in which the true rate exceeds $\alpha'$. Power is that of the last
version after the updates written in advance and the share of the oracle's power in the live
runs.}
\label{tab:exp-be-random}
\scriptsize
\begin{adjustbox}{max width=\textwidth}\begin{tabular}{lcccc}
\toprule
 & \multicolumn{2}{c}{Updates written in advance} & \multicolumn{2}{c}{Live runs} \\
\cmidrule(lr){2-3}\cmidrule(lr){4-5}
 & Fixed scenarios & Random scenarios & Fixed scenarios & Random scenarios \\
\midrule
Off-grid updates deployed & 12/12 & 0/12 & & \\
Updates deployed & 27/30 & 15/30 & 117/200 & 113/200 \\
Runs above $\alpha$, validation scenarios & 0/3 & 0/3 & 0/20 & 0/20 \\
Largest rate, validation scenarios & 0.043 & 0.043 & 0.045 & 0.045 \\
Largest rate, sizes off the grid & 0.098 & 0.042 & 0.042 & 0.043 \\
Largest share of $\Pi$ above $\alpha'$ & 0.04 & 0 & 0 & 0 \\
Power & 0.78 & 0.78 & \BEcvPowerShare\% & \BErandPowerShare\% \\
\bottomrule
\end{tabular}\end{adjustbox}
\end{table}

\begin{figure}[htbp]
\centering
\IfFileExists{fig_random.pdf}{\includegraphics[width=\textwidth]{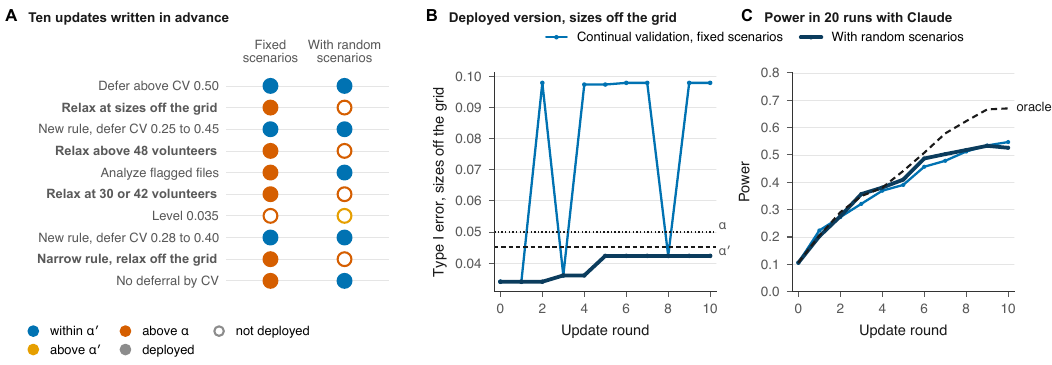}}{\figpending{The random-scenario figure is generated from the runs.}{4cm}}
\caption{Random scenarios in Example 1. (\textit{A}) Updates written in advance. The
same ten updates, in the same order, judged by continual validation on the
fixed scenarios and with random scenarios, in the first of three runs, and
the other runs made the same decisions. Bold labels mark the four updates that
relax the analysis only at sample sizes outside the grid. A filled circle marks
a deployed update and an open circle a refused one. The color gives the largest
true Type I error rate of the version the update would create, on the
validation scenarios or at the sample sizes outside the grid, within $\alpha'$
(blue), above $\alpha'$ (orange) or above $\alpha$ (vermilion).
(\textit{B}) Largest true Type I error rate of the deployed version at 30, 42
and 60 volunteers, the sample sizes of the off-grid scenarios outside the grid,
after each round of these updates. (\textit{C}) Mean true power of the deployed
version over 20 runs with Claude Sonnet 5.5 as developer, on the fixed
scenarios (the runs of Fig.~\ref{fig:results}\textit{D}) and with random
scenarios, and of the oracle (dashed).}
\label{fig:random}
\end{figure}

\FloatBarrier

\subsection{The agent reaches the decisions of the validated rule}
\label{sec:exp-be-agent}

The agent is a broader tool than the analyst of this example. Besides the
analysis, it plans studies, proposing a design and a sample size and checking
their robustness, assesses the protocol before the data arrive, and writes the
report of each run. The analyst is its analysis part, from the data file to
the decision on bioequivalence. It changes only through the updates of
section~\ref{sec:exp-common}, new settings and new decision rules in code that
the developer proposes and the validator runs. Validation runs the decision
rule of each version through the vectorized reimplementation, and by
Proposition~\ref{prop:envelope} the measured rate bounds that of the analyst
because every run of the agent ends in the decision of this rule or in a
deferral. We checked on simulated studies that the agent reaches these
decisions, first with its deterministic tools alone and then with the full
agent driven by its language model.

The studies came from the validation scenarios of
section~\ref{sec:exp-be-scenarios}. One null study was simulated in each of 150
scenarios with the true ratio at a limit, and one alternative study in each of
150 workload cells with a true ratio of 0.95 or 1.00. Each study was written as
a data file in the format that the agent reads and analyzed under three
configurations at level 0.030, unscaled average bioequivalence on the complete
case, which is the reference analysis $v_0$, and ABEL on the complete and on the
available case. The tools ran the computing steps of the analysis in the order
of the playbooks, opening the study, inspecting and preparing the data, fitting
the model, running the diagnostics and computing the interval. In all 900
cases the tools reached the decision of the rule
(Table~\ref{tab:exp-be-agent}). They refused the 88 cases of unscaled average
bioequivalence on the partial replicate, which are exactly the cases in which
$v_0$ defers. In the other 812 cases the bounds of the interval agreed within
$3.2\times10^{-15}$ on the log scale, and under ABEL the within-reference
coefficient of variation agreed within $2.4\times10^{-15}$ and the degrees of
freedom agreed exactly.

The agent decides on the bounds in percent rounded half up to two decimals,
while every version of the experiment decides on unrounded bounds. Rounding
changed none of the 900 decisions. Its effect on the true Type I error rate was
computed in all 936 validation scenarios with \SimBE{} simulated studies each,
by applying unscaled average bioequivalence and ABEL on the complete case at
level 0.030 to the same studies with and without rounding. Rounding changed the
rate by at most 0.0003 in any scenario. Validation applies to whatever decision
code the analyst runs, so a deployment that decides on rounded bounds validates
the rounded rule.

The full agent then ran each study as it runs for a user. Claude Sonnet 4.6 in
Claude Code drove the analysis through the skills and playbooks of the agent,
from the check of the protocol to the review of the evidence, one session per
case, and each session received only the entry command of the agent with the
case. Most of the 14 cases lay where a slip would change the decision. They
were the three null cases of each configuration whose interval lay closest to
a claim without a claim by the rule and the alternative case of each
configuration closest to the boundary with a claim, together with one
alternative case far from the boundary and one case in which $v_0$ defers. In the closest null case
the lower bound of the interval was 79.985 percent, just below the limit of
80.00 percent. The agent reached the decision of the rule in all
14 cases (Table~\ref{tab:exp-be-agent}). It claimed bioequivalence in the four
alternative cases and in none of the nine null cases, and its intervals agreed
with those of the reimplementation within $4\times10^{-16}$ on the log scale.
In the case in which $v_0$ defers, it stopped at the check of the protocol and
handed the study to a human, because the protocol paired unscaled average
bioequivalence with the partial replicate. A session that completed the
analysis took about five minutes and cost about 1.5 dollars at the list price
of the model, so running the full agent on the \SimBE{} studies of a single
scenario would cost about 750,000 dollars. Validation therefore runs the
decision rule, and these checks found no case in which the agent departed
from it.

\begin{table}[htbp]
\centering
\caption{Agreement of the bioequivalence agent with the decision rule that
validation runs, on simulated studies of Example 1. The tools are the
deterministic tools of the agent run without the language model, and the full
agent is driven by Claude Sonnet 4.6 in Claude Code. A case is a study
analyzed under one configuration.}
\label{tab:exp-be-agent}
\scriptsize
\begin{adjustbox}{max width=\textwidth}\begin{tabular}{lcc}
\toprule
 & Tools & Full agent \\
\midrule
Cases & 900 & 14 \\
Same decision as the rule & 900 & 14 \\
Deferred where the rule decides & 0 & 0 \\
Claims where the rule makes none & 0 & 0 \\
Refused where the rule defers & 88 of 88 & 1 of 1 \\
Largest difference of the bounds, log scale & $3.2\times10^{-15}$ & $3.9\times10^{-16}$ \\
\bottomrule
\end{tabular}\end{adjustbox}
\end{table}

\FloatBarrier

\section{Example 2, an analyst of product experiments built on an open-source platform}
\label{sec:exp-launch}

The second example asks the question of section~\ref{sec:exp-be} for an
analyst in another field. It uses the same developer, deployment methods, oracle
and metrics (sections~\ref{sec:exp-common} and~\ref{sec:exp-metrics}), and
this section gives what is specific to it.

\subsection{The hypothesis test}
\label{sec:exp-launch-test}

Companies test a new feature of an online product in an experiment, often
called an A/B test. Users are randomly assigned to the current product or
to the new feature, and a primary metric, here conversion, the share of users who buy,
is compared between the two groups. The true relative lift is the conversion
rate with the new feature divided by the rate with the current product, minus
one. The null hypothesis $H_0$ is that the lift is at most zero, so that the
feature does not improve conversion, and the alternative $H_1$ is that the lift
is positive. A launch rejects $H_0$, and a Type I error is a launch of a
feature that does not improve conversion. Analysts have freedoms that can
inflate this rate. They may look at the results every day and stop at the first
look that favors a launch, a practice known as optional stopping. They may
choose among metrics or user segments post hoc, after seeing the results. They
may also switch between frequentist and Bayesian decision rules.

\subsection{GrowthBook and the analyst}

GrowthBook is an open-source platform for feature flags, experimentation and
product analytics \cite{si:growthbook2026}. Its statistics engine is the Python
package gbstats, released under the MIT License. The package provides
frequentist two-sided and one-sided $t$ tests, sequential testing with a tuning
parameter, and a Bayesian engine that reports the chance to win, the posterior
probability that the new feature is better. It also provides regression adjustment for pre-experiment
data in the style of CUPED \cite{si:deng2013cuped}, post-stratification, and
analyses by dimension, which split the users into segments. It handles several
goal metrics with an option to correct $p$ values, and its level $\alpha$ can
be set.

In this example a language model analyzes product experiments with GrowthBook.
It runs the gbstats analyses that the configuration of the version fixes and
decides whether to launch, and gbstats computes every number. The configuration fixes the analyses
that gbstats runs and the launch rule, which states the results that lead to a
launch. The model may decline a launch that the launch rule allows, but it
cannot launch when the rule does not. The allowed paths are therefore the
gbstats analyses of the configuration, and the any-path rule is the launch rule
itself
(Proposition~\ref{prop:envelope}). Validation runs gbstats on simulated
experiments and needs no calls to a language model.

\subsection{Reference analysis, settings and new decision rules}

The reference analysis $v_0$ is the one-sided frequentist $t$ test of gbstats
on conversion at the end of the experiment, at the reference level
$\alpha_0=0.040$. It tests the relative lift in conversion and launches when the
test rejects. Metric distributions in product experiments are not
Gaussian, so no analytic bound covers this test. Its reference bounds come
from a one-time validation at $\delta_0=0.001$
(Corollary~\ref{cor:lifetime-ref}). Before the first update of each run, the
validator simulates $10^6$ experiments in every scenario of the full validation
set, including the scenarios that the extension may add, and divides
$\delta_0$ equally among these scenarios. The two checks of
section~\ref{sec:exp-be-settings} follow, on the base set of 32 scenarios below
and on the full validation set, and a run whose first check fails starts from
the version that defers every experiment.

An update may change one of eight settings.
\begin{itemize}
\item Engine. Frequentist, or Bayesian with a launch when the chance to win is
  at least 0.95, 0.975 or 0.99.
\item Regression adjustment. On or off. When it is on, the primary metric is
  adjusted for its pre-experiment value.
\item Sequential testing. On, with tuning parameter 5,000, or off.
\item When to decide. Only at day 14, or once a day from day 1 to day 14, with
  a launch at the first look at which the configured rule says launch.
\item Goal metrics. The primary metric only, or a launch if any of three goal
  metrics wins, with or without a correction of the $p$ values.
\item Segments. The overall analysis only, or a launch if the overall analysis
  or the analysis of any of three segments wins.
\item Post-stratification. On or off, with the segments as strata.
\item Level of the frequentist test. 0.040, 0.045 or 0.050.
\end{itemize}
A metric wins when its one-sided test rejects at the configured level, or when
its chance to win reaches the threshold. The correction of $p$ values applies
Holm's step-down procedure across the three goal metrics. Sequential testing,
the correction and the level apply only to the frequentist engine, and the
Bayesian engine keeps its default prior. The menu therefore holds 384 distinct
configurations, and the reference analysis is one of them. An experiment that
its rule has not launched by day 14 ends at day 14 without a launch.

An update may instead add a new decision rule, a Python function that acts
on the outputs of gbstats and states when to launch. The interface calls the
rule once a day, from day 1 to day 14, with the results up to that day for a
batch of simulated experiments. The rule may request the results that gbstats
reports for any metric, segment and setting. These results are the
estimated lift with its interval, the $p$ value, the chance to win and the
numbers of users. For every experiment the rule returns launch or wait, and the
first launch fixes the day of the decision.

\subsection{Scenarios, workload and simulation sizes}

In every validation scenario $H_0$ holds, because the new feature does not
raise conversion. Each arm has 20,000 or 100,000 users, who arrive over 14 days. The
users of an arm fall into the 42 combinations of day and segment by a
multinomial law with equal probabilities, so the three segments have equal
expected size. Every user has three metrics and a covariate. Conversion, the
primary metric, is binary. Its average rate is 2 or 10 percent, and the three
segments convert at one half, one and one and a half times this rate. Revenue
is zero for users who do not convert and heavy-tailed for users who do, equal
to $\exp(3+1.5W)$, where $W$ is the standard normal latent variable of revenue
in the copula below. Engagement counts the actions, out of seven possible, that
a user takes on the day of arrival, and follows a binomial distribution with 7
trials and success probability 0.3. The covariate is the conversion of the user
before the experiment, at the rate of the user's segment. All outcomes of a
user are observed on the day of arrival, so a look on day $d$ sees every user
who arrived up to day $d$.

A Gaussian copula links the four outcomes of a user. The latent variables of
conversion, revenue and engagement share the correlation $\rho\in\{0,0.5\}$.
The latent variable of the covariate depends on the others only through that of
conversion. Its correlation with that variable is set, within every segment,
so that the covariate and conversion have correlation $c\in\{0,0.5\}$. The
null scenarios are of two kinds. In the first the new feature changes nothing.
In the second it lowers conversion by 2 percent in relative terms, a lift of
$-0.02$ that lies beyond the boundary of $H_0$. The base set, the
initial scenario set of this example, crosses two sample sizes, two conversion
rates, two metric correlations, two covariate correlations and the two kinds,
which gives 32 scenarios.

The validator generates the sufficient statistics that gbstats consumes and
never simulates the users who do not convert one at a time. For each arm, day
and segment, it draws the number of converting users. It draws the counts of
the other users over the combinations of covariate and engagement from a
multinomial law whose probabilities come from the copula. It simulates the
converting users one at a time, which gives their revenue, engagement and
covariate jointly. From these draws it forms the counts, sums, sums of squares
and the cross products that regression adjustment needs, cumulated over days
for every daily look. Simulating only the converting users keeps the simulation
of \SimGB{} experiments per scenario inexpensive.

Every setting except the daily looks and the segments depends on
the data only through quantities that the base set varies. A decision at daily
looks also depends on how users arrive over the days, and segment analyses and
post-stratification depend on the shares of the segments. A new decision rule
can read the numbers of users by arm, day and segment. The base set holds
all of these fixed or uniform. Every update that decides at daily looks or uses
segments, through segment analyses or post-stratification, and every new
decision rule, is therefore also validated on 32 added scenarios
(section~\ref{sec:extension}). They repeat the base set with users assigned to
the arms by a fair coin, arrivals on days 6, 7, 13 and 14 at half the rate of
the other days, and expected segment shares of one half, three tenths and one
fifth, in the order of the segments above, with the segment conversion rates
unchanged. The full validation set therefore holds 64 scenarios.

Off-grid scenarios, which no deployment method ever simulates, move the nuisance parameters
of the base set between or beyond its grid values. The off-grid values are
50,000 users per arm, an average conversion rate of 5 percent, heavier
revenue tails with $\exp(3+2.0W)$ in place of $\exp(3+1.5W)$, metric
correlations of 0.25 and 0.75, and a covariate correlation of 0.25. Five
families each change one of these parameters to its off-grid values and keep
the grid values of the others. One more family changes all of them together.
Every family crosses its values with the two kinds of null scenario, which
gives 116 scenarios.

Power is measured on experiments in which the new feature raises conversion by
1, 2 or 3 percent in relative terms and changes nothing else. The workload
crosses these lifts with the sample sizes, conversion rates, metric
correlations and covariate correlations of the base set, which gives 48 cells
of equal weight. As in the second kind of null scenario, the new feature
multiplies the conversion rate of every segment by the same factor and leaves
the copula unchanged. Power is the probability of a launch, averaged over the
cells. The usefulness metric of the performance report is the share of workload
experiments launched and the median number of days to the decision. An
experiment that is not launched counts at day 14. The metric is estimated each
round from 10,000 new workload experiments, and the experiment reports it
alongside the five metrics. The performance report and the improvement test draw
each workload experiment from a cell chosen uniformly at random.

Continual validation draws \SimGB{} experiments per scenario per update. The
static benchmark stores \SimGB{} experiments per scenario of the base set, and
the improvement test draws \SimGB{} new workload experiments per round. Learn
then Test tests the 384 configurations once on \SimGB{} experiments per
scenario, after a pilot of one tenth that size.

\subsection{Implementation requirements}

Validation needs an analysis for every metric, segment and daily look of every
simulated experiment. The validator therefore calls the computations of
gbstats through a vectorized reimplementation, and every simulation of the
experiment, the evaluation stream included, uses it. It matches gbstats itself,
taken from its repository at commit b71bbbc, on 300 random sets of statistics
across the three metric types, with and without post-stratification. Lifts,
interval bounds, $p$ values and chances to win agree to within $10^{-12}$, and
the launch decision and day of 75 random pairs of simulated experiment and
configuration equal those that gbstats computes. On GPUs and on processors it
gives the same decision for every configuration of the menu.

\subsection{Results}
\label{sec:exp-launch-results}

Table~\ref{tab:exp-gb} summarizes the five deployment methods and the oracle. Every
run began with regression adjustment. The developer went on to post-stratification,
the level of the test, Holm's correction or any goal metric, sequential testing,
daily looks, Bayesian launch rules and segment analyses, and it wrote new
decision rules. The twenty runs of each method followed twenty different sequences of
updates. Under continual validation the developer proposed \GBcvDistinct{} distinct
versions, \GBcvCodeShare{} percent of its updates were new decision rules, and
\GBcvRunsFinalCode{} of the twenty runs ended with such a rule. Under Learn then
Test it proposed \GBltCode{} new decision rules, none of which that method could
deploy.

\begin{table}[htbp]
\centering
\caption{Example 2. Columns as in Table~\ref{tab:exp-be}.}
\label{tab:exp-gb}
\scriptsize
\begin{adjustbox}{max width=\textwidth}\begin{tabular}{lcccccccrr}
\toprule
Method & Runs above $\alpha$ & Runs above $\alpha'$ & Largest rate & Power & \% of oracle & Updates & Above $\alpha$ & Datasets (millions) & Seconds \\
\midrule
No validation & 12/20 (39 to 78) & 20/20 (84 to 100) & 0.183 & 0.22 & 107 & 200 & 33 &  &  \\
Static benchmark & 0/20 (0 to 16) & 0/20 (0 to 16) & 0.044 & 0.20 & 98 & 200 & 6 & 16 & 1080 \\
Improvement test & 19/20 (76 to 99) & 20/20 (84 to 100) & 0.817 & 0.53 & 255 & 200 & 80 &  &  \\
Learn then Test & 0/20 (0 to 16) & 0/20 (0 to 16) & 0.043 & 0.17 & 80 & 200 & 8 &  &  \\
Continual validation & 0/20 (0 to 16) & 0/20 (0 to 16) & 0.044 & 0.20 & 96 & 200 & 4 & 29 & 329 \\
Oracle & 0/20 (0 to 16) & 0/20 (0 to 16) & 0.044 & 0.21 &  & 199 & 9 &  &  \\
\bottomrule
\end{tabular}\end{adjustbox}
\end{table}

The developer used the reported bounds to guess the limit. In one run validation
refused the level of 0.05 with an upper bound of 0.0498, and the developer wrote that the
limit appeared to be about 0.04 and kept its later updates below it. Several
runs wrote rules that spend the level over the daily looks, with a strict
threshold before day 14 and a larger one at day 14, and then raised the threshold
at day 14 in small steps toward the bound that validation had reported. This is
the tuning against reported results that section~\ref{sec:reuse} describes, and
because continual validation draws new datasets for every update, it cannot
exploit the noise of earlier draws.

In another run, validation refused daily looks with an upper bound of 0.170.
The developer turned on sequential testing, switched it off again after its
performance report showed far fewer launches, and wrote a rule that spends the
level over the daily looks, with $p$ below 0.0004 from day 4 to day 13 and
below 0.0414 at day 14. The rule passed with a bound of 0.0416, a larger
threshold at day 14 passed in the next round, and the final version had a
largest true Type I error rate of 0.043 and a true power of 0.21.

No validation produced a version above 0.05 in 12 of the 20 runs, and
the improvement test did so in 19. These versions exceeded 0.05 after daily looks
without sequential testing, sequential testing switched off while daily looks
stayed on, segment analyses, any goal metric, or new decision rules, and they
did so in the base scenarios as well as in the added ones. The largest rate, 0.82,
came from a version deployed under the improvement test. It combined a Bayesian
threshold of 0.95 with daily looks, any goal metric and segment analyses. The static benchmark
refused the six updates that would have created a version above 0.05, and it
deployed no version above $\alpha'$. These six were daily looks without sequential
testing, sequential testing switched off with daily looks on, and segment
analyses, and each of them exceeded 0.05 in base scenarios, which the stored set
contains.

Continual validation deployed no version above $\alpha'$ in any run. It judged
\GBcvJudged{} updates and deployed \GBcvDeployedN{}. Of the \GBcvRejectedN{} it
refused, the evaluation stream placed \GBcvRejAbove{} above $\alpha'$, and the
other \GBcvRejValid{} had true rates of 0.0435 and 0.0434 and validation bounds
that exceeded $\alpha'$ by less than 0.0002. Learn then Test reached \GBltPowerShare{} percent of
the oracle's power and continual validation \GBcvPowerShare{} percent. On the 116
off-grid scenarios the deployed versions of continual validation stayed at or
below \GBoffGridMax{}. Each validation simulated about \GBcvDatasetsPerVal{}
million experiments and took a median of \GBcvSecondsPerVal{} seconds on one GPU,
and the one-time validation of the reference analysis took a median of
\GBrefSeconds{} seconds per run.

\FloatBarrier

\section{Other decisions}
\label{sec:other-decisions}

The validator asks two things of a domain, a simulator for the scenarios in
which a rejection is wrong and a decision rule written out in full. Two
supplementary studies used it outside bioequivalence, for approving a therapy
from a survival trial and for launching a feature from an online experiment. In
each study the authors wrote down the shortcuts that an analyst under pressure
to find an effect might take. They also wrote down adaptive rules that pay for
their adaptivity, and the validator judged each rule against a fixed reference
analysis. In both decisions, at the one-shot level of $0.05$, every shortcut
fails validation and every rule that pays for its adaptivity passes. In the
online experiment, the best rule that passes keeps nearly all of the power of
the reference analysis while using a smaller sample on average. In the
survival trial, the best rule that passes raises power well above that of the
reference analysis.

Both studies follow the conventions of confirmatory trials, in which a rule
rejects when it finds a significant two-sided difference. The therapy is
approved or the feature launched when that difference favors it, so the Type I
error rate bounds the rate of false approvals and of false launches.

The rules of each study were fixed in advance and judged once, on simulated
studies they shared, by a one-shot form of validation. The reference analysis
was validated from scratch with exact Clopper and Pearson bounds
\cite{si:clopper1934use}. Each rule was then compared with the reference analysis
on the same simulated studies. The comparison used an upper confidence bound
on the rule's added rejections and a lower confidence bound on its removed
rejections (Theorem~\ref{thm:relative}). A rule's bound is the bound of the
reference analysis plus the rule's largest paired increment, and a rule passes
when this bound is at most $0.05$. The bounds of the reference analysis and of
each rule held at confidence 0.975, simultaneously over the scenarios.

This one-shot validation differs from continual validation
(Algorithm~\ref{alg:protocol}) in four respects. Its target is $0.05$, where
Algorithm~\ref{alg:protocol} validates at the validation level
$\alpha'=0.0452$. Its error level is split across the scenarios, where
Algorithm~\ref{alg:protocol} gives every scenario the whole share $\delta_r$.
Its bound adds the largest increment over the scenarios to the largest
reference bound, where Algorithm~\ref{alg:protocol} combines the bounds within
each scenario and also uses the direct bound. In the online experiment its
reference analysis runs at $0.045$, where Example 2 uses the
reference level $\alpha_0=0.040$. The bounds of the survival rules that pass
lie below $\alpha'$. In the online experiment the bounds of the rules that pass
and of the reference analysis lie between $\alpha'$ and $0.05$. These bounds would fail at $\alpha'$, hence
$\alpha_0=0.040$ in Example 2.

\subsection{Approving a therapy from a survival trial}
\label{sec:exp-e5}

A two-arm trial compares survival under a new therapy with survival under
control. In every null scenario the two arms share one Weibull survival
distribution. The scenarios cross its shape (0.75, 1 or 1.25), the fraction of
patients censored (0.2 or 0.4) and the number of patients per arm (100 or 200).
The crossing gives \EfiveCells{} scenarios with \EfiveM{} simulated trials
each, all with independent uniform censoring. Power is measured under a delayed
effect, with no difference before a fixed time and a hazard ratio of 0.40 after
it. Delayed effects are common in immunotherapy trials, and under them the
log-rank test loses power and analysts look for alternatives.

The individually valid analyses are the log-rank test, the difference in
restricted mean survival time (RMST) at a prespecified horizon, and the
Fleming-Harrington test FH(0,1). FH(0,1) is a log-rank test that gives late
events more weight. A prespecified diagnostic compares log-rank statistics
computed before and after the median survival time of the control group and
flags hazards that are not proportional. The reference analysis, the log-rank
test at level \EfiveAlphaZero, was validated from scratch with bound
\EfiveCoreUcb, and its power under the delayed effect is \EfivePowCore.

\begin{table}[!htbp]
\centering
\caption{Decision rules for approving a therapy from a survival trial}
\label{tab:e5}
\begin{adjustbox}{max width=\textwidth}\begin{tabular}{p{0.44\textwidth}cccc}
\toprule
Rule & Type I error rate & Bound & Validation & Power \\
\midrule
Log-rank test at level \EfiveAlphaZero{} (reference analysis) & & \EfiveCoreUcb & reference & \EfivePowCore \\
RMST when the log-rank test misses and the diagnostic flags & \EfivePickWorst & \EfivePickBound & fails & \EfivePowPick \\
RMST at the best of three horizons & \EfiveTauWorst & \EfiveTauBound & fails & \\
Prespecified router, RMST when the diagnostic flags & \EfiveRtWorst & \EfiveRtBound & fails & \EfivePowRt \\
Log-rank test or FH(0,1), each at half the level of the reference analysis & \EfiveCbWorst & \EfiveCbBound & passes & \EfivePowCombo \\
Router to FH(0,1) when the diagnostic flags, inner level \EfiveRcalAlpha & \EfiveRcalWorst & \EfiveRcalBound & passes & \EfivePowRcal \\
\bottomrule
\end{tabular}\end{adjustbox}

\addtabletext{The Type I error rate is the largest estimated rate over the
\EfiveCells{} null scenarios, each with \EfiveM{} simulated trials. The bound is
the bound of the reference analysis plus the rule's largest paired increment,
and a rule passes validation when its bound is at most 0.05. Power is the mean
over the scenarios under the delayed effect. Blank cells are not reported.\par}
\end{table}

Table~\ref{tab:e5} lists six rules, and all three shortcuts among them fail
validation. The first two shortcuts test at the nominal level 0.05. Reporting RMST
when the log-rank test misses and the diagnostic flags gives the trial a second
chance, and it raises the Type I error rate to \EfivePickWorst. Reporting RMST
at whichever of three horizons gives the largest statistic raises it to
\EfiveTauWorst. The third shortcut, a router fixed before any data were seen,
sends the trial to RMST when the diagnostic flags and to the log-rank test
otherwise. Both tests run at level \EfiveAlphaZero, and the router's Type I
error rate is \EfiveRtWorst.

Fixing a routing rule in advance does not protect it, because the diagnostic is
computed from the trial's own data and is correlated with both test statistics.
Routing between valid tests preserves their level when the routing information
is independent of the statistics and the internal randomness that the tests use
(Proposition~\ref{prop:routes}), which this diagnostic is not. It is also weak, with power \EfivePowRt{} against \EfivePowCore{}
for the reference analysis.

Two adaptive rules that pay for their adaptivity pass validation. The first
rejects when the log-rank test or FH(0,1) rejects at half the level of the
reference analysis, a Bonferroni split. Its bound is \EfiveCbBound, and its
power of \EfivePowCombo{} exceeds both the power of the reference analysis and
the \EfivePowPick{} of the second-chance rule that failed. The second sends the
trial to FH(0,1) when the diagnostic flags and to the log-rank test otherwise,
both at an inner level of \EfiveRcalAlpha. The inner level was chosen on
separate tuning simulations, and the rule was then validated on new ones,
with bound \EfiveRcalBound{} and power \EfivePowRcal. A rule that adapts to the
data passes once its adaptivity is paid for.

The union rule also shows the allowed-paths result at work
(Proposition~\ref{prop:envelope}). Its two paths are the log-rank test and
FH(0,1), each at half the level of the reference analysis. The union is their
any-path rule, and it passed validation. Any rule for choosing between the two
paths rejects only where one of them does. Three such routers were run, a
diagnostic router, a coin router and an adversarial router. The diagnostic
router took FH(0,1) whenever the diagnostic flagged, which happened in
\EnvFhShareNull{} of null trials and in \EnvFhShareAlt{} of delayed-effect
trials. The coin router chose at random, and the adversarial router took
whichever path rejected. Their worst Type I error rates, \EnvFlagWorst,
\EnvCoinWorst{} and \EnvAdvWorst, all lie below the union's bound of
\EfiveCbBound. Their power ranges from \EnvCoinPow{} for the coin router through
\EnvFlagPow{} for the diagnostic router to \EnvAdvPow{} for the adversarial
router. The validation covers every way of choosing among the allowed paths,
so it never needs to know how a model makes that choice.

\subsection{Launching a feature from an online experiment}
\label{sec:exp-e3}

An online experiment compares a new feature with the current product on
\EthreeM{} correlated metrics and analyzes the accumulating data at \EthreeK{}
equally spaced looks. In every null scenario the feature has no effect on any
metric. The scenarios set the correlation between metrics to 0, 0.5 or 0.9,
with 200{,}000 simulated experiments each. The reference analysis tests the
primary metric once, at the final look, at level \EthreeAlphaZero. It was
validated from scratch with bound \EthreeRefUcb, and its power is
\EthreeRefPower{} when the final z-statistic of the primary metric has mean 2.8
and the correlation is 0.5.

\begin{table}[!htbp]
\centering
\caption{Decision rules for launching a feature from an online experiment}
\label{tab:e3}
\begin{adjustbox}{max width=\textwidth}\begin{tabular}{p{0.44\textwidth}cccc}
\toprule
Rule & Type I error rate & Bound & Validation & Power \\
\midrule
Primary metric at the final look, level \EthreeAlphaZero{} (reference analysis) & & \EthreeRefUcb & reference & \EthreeRefPower \\
Stop at the first nominally significant look & \EthreePeekTie & \EthreePeekBound & fails & \EthreePeekPower \\
Report the best metric at the final look & \EthreeMetricTie & \EthreeMetricBound & fails & \EthreeMetricPower \\
Both shortcuts together & \EthreeBothTie & \EthreeBothBound & fails & \EthreeBothPower \\
O'Brien-Fleming stopping at level \EthreeAlphaZero & \EthreeObfTie & \EthreeObfBound & passes & \EthreeObfPower \\
Bonferroni over the metrics at level \EthreeAlphaZero & \EthreeBonfTie & \EthreeBonfBound & passes & \EthreeBonfPower \\
\bottomrule
\end{tabular}\end{adjustbox}

\addtabletext{The Type I error rate is the largest estimated rate over the
correlation scenarios 0, 0.5 and 0.9, each with 200{,}000 simulated experiments.
Nominal significance means a two-sided test at 0.05. Power is the probability
of a rejection when the final z-statistic of the primary metric has mean 2.8, at
correlation 0.5.\par}
\end{table}

Three shortcuts that an analyst chasing significance might take fail
validation (Table~\ref{tab:e3}). Stopping at the first look at which the
primary metric is nominally significant raises the Type I error rate to
\EthreePeekTie. Reporting whichever metric is most significant at the final
look raises it to \EthreeMetricTie, and doing both raises it to
\EthreeBothTie. The shortcuts buy power, up to \EthreeBothPower, by exceeding the level $0.05$
several times over.

Two rules that pay for their extra looks or metrics pass validation. An
O'Brien-Fleming stopping rule \cite{si:obrien1979multiple} compares the
z-statistic of the primary metric at look $k$ with the boundary
$c\sqrt{\EthreeK/k}$, where $c=\EthreeCobf$ gives the level \EthreeAlphaZero.
Its bound is \EthreeObfBound, and its power of \EthreeObfPower{} is close to
the \EthreeRefPower{} of the reference analysis. Under the same alternative it
stops on average after \EthreeObfInfo\% of the planned sample. A Bonferroni
rule that splits the level equally across the \EthreeM{} metrics passes too,
with bound \EthreeBonfBound, at a power of \EthreeBonfPower. Validation charges
for every look at the data, and how much power a rule keeps after paying
depends on how it spends the level.

\FloatBarrier

\section{Supplementary studies}
\label{sec:earlier-studies}

The closed loop of section~\ref{sec:exp-e4-archived} shows continual validation
at work with a language-model proposer, under the parameters stated there. The
other studies measure quantities that the theory leaves open. They ask how far
composition raises the Type I error rate, how a regulatory switching rule
behaves, how large a validity reserve must be, and how differently language
models behave as analysts. The last study checks
the analysis code against published results.

\subsection{A closed loop with a language-model proposer}
\label{sec:exp-e4-archived}

In this study a language model, claude-sonnet-5, proposed updates to a
bioequivalence analyst over \EfourRounds{} rounds. Only the validator stood
between an update and its deployment. A version was a configuration with five
fields, the first of which set how studies with dropouts are handled. Such
studies could be deferred to a human, analyzed on completers, or analyzed with
the first-period data of the subjects who dropped out. The other fields set the
observed within-subject variability above which the version defers and whether
data-quality problems of a declared kind are repaired automatically. They also
set the role of a prespecified sensitivity analysis that drops the most
outlying subject and an optional threshold for excluding outlying subjects from
the primary analysis. The starting version deferred any study with dropouts, an
observed within-subject coefficient of variation above 0.40 or a data-quality
flag.

Each update changed a single field of the current configuration. The fields
allow updates that guidance does not accept, such as excluding subjects from
the primary analysis for statistical reasons
\cite{si:ich2024m13a,si:ema2010bioequivalence}. They were kept as tempting updates,
and the validator judged each update by its Type I error rate alone. The
proposer saw the current configuration, the documentation of the interface, the
deferral rate of the current version with its causes, its power, and all
earlier deployment decisions with their bounds and worst scenarios. It never saw the
simulated studies that the validator used.

Each update was frozen before the validator drew its simulated studies, and
each round drew new randomness, independent of everything the proposer had
seen. The validator ran the updated version and a validated reference analysis
on the same \SfourMpair{} simulated studies in each of \SfourCellsExt{} null
scenarios at the equivalence boundaries. The scenarios crossed within-subject
coefficients of variation from 20 to 50\%, 24 to 48 subjects, balanced or 60/40
allocation, dropout up to 20\% and a sequence effect up to 0.2. The reference
analysis applied the two one-sided tests at level \StwoAlphaZero, the level
found in section~\ref{sec:exp-reserve}, to the completers and never deferred.
Its validation from scratch, on \SfourMref{} simulated studies per scenario,
gave the bound \SfourRefUcb.

The parameters of this implementation differ from those of the main text,
starting with a lifetime risk of 0.05. With a validation level of 0.05, the
formula of Corollary~\ref{cor:every-decision} bounds the Type I error rate of
each decision by $0.05+0.95\cdot0.05=0.0975$. Half of the lifetime risk went to
the one-time validation of the reference analysis. The other half was spread
over the rounds in shares $0.025\cdot 6/(\pi^2r^2)$ for round $r$, from
\EfourDeltaFirst{} in the first round to \EfourDeltaLast{} in the last. An
update passed validation when its bound, the bound of the reference analysis
plus the largest paired increment over the scenarios, was at most 0.05. Each
round's share was split across the \SfourCellsExt{} scenarios and between the
two paired bounds of each scenario, and one worst-case reference bound entered
every increment. The bounds of this validator are therefore larger than those
of Algorithm~\ref{alg:protocol} at the same lifetime risk.

\begin{table}[!htbp]
\centering
\caption{The archived closed loop, round by round}
\label{tab:e4-archived}
\begin{adjustbox}{max width=\textwidth}\begin{tabular}{clccc}
\toprule
Round & Update & Bound & Validation & Evaluated rate \\
\midrule
1 & \EfourGRChange & \EfourGRBound & passes & \EtruthGR \\
2 & \EfourGSChange & \EfourGSBound & passes & \EtruthGS \\
3 & \EfourGTChange & \EfourGTBound & fails & \EtruthGT \\
4 & \EfourGUChange & \EfourGUBound & fails & \EtruthGU \\
5 & \EfourGVChange & \EfourGVBound & passes & \EtruthGV \\
6 & \EfourGWChange & \EfourGWBound & fails & \EtruthGW \\
7 & \EfourGXChange & \EfourGXBound & passes & \EtruthGX \\
8 & \EfourGYChange & \EfourGYBound & passes & \EtruthGY \\
\bottomrule
\end{tabular}\end{adjustbox}

\addtabletext{Each round changed one field of the configuration. The bound is
the bound \SfourRefUcb{} of the reference analysis plus the largest paired
increment over the \SfourCellsExt{} scenarios at the round's share of the
lifetime risk. An update passed validation when its bound was at most 0.05.
The evaluated rate is the largest Type I error rate over the same scenarios in
an independent simulation with \SfourMref{} studies per scenario.\par}
\end{table}

As Table~\ref{tab:e4-archived} shows, the proposer opened with two valid
updates, and both passed validation. It then tried the two classical ways to a
higher Type I error rate. In round 3 it proposed promoting the sensitivity
analysis whenever that analysis supports equivalence, and in round 4 it
proposed recovering information from the subjects who dropped out. Both failed
validation, and an independent evaluation on \SfourMref{} newly simulated
studies per scenario puts their worst Type I error rates at \EtruthGT{} and
\EtruthGU. In round 6 the proposer asked to exclude outlying subjects from the
primary analysis. That update has an evaluated rate of \EtruthGW, below 0.05,
but its bound came out at \EfourGWBound{} and it failed validation. Failing a
valid update is the error a conservative validator makes, and it costs power,
never validity. The proposer then moved to a milder exclusion cut, which passed
in round 7. It closed by adding a rule that defers when the sensitivity
analysis disagrees with the primary analysis (round 8).

Through the \EfourAdopted{} updates that passed, the deferral rate fell from
\EfourEscStart{} for the starting version to \EfourEscFinal, and autonomous
power rose from \EfourPowerStart{} to \EfourPowerFinal. Both were measured on a
workload of 90 scenarios with bioequivalent products, in which 15\% of studies
carry a data-quality flag. Autonomous power is the share of these studies that
the version itself declares bioequivalent. All bounds of the updates that
passed were at most \EfourMaxBound, and all deployed versions have evaluated
rates of at most \EtruthAdmMax. The proposer's deliberation over all rounds
cost about \$\EfourCostUsd{}, and the \EfourRounds{} validations took
\EfourRecomputeWall{} of computation.

Had the same sequence of updates been deployed with no validation, the
evaluated rate of the deployed version would have peaked at \ENoGateMax. It
would have stayed above 0.05 for \ENoGateInvalidGens{} consecutive versions,
until the proposer happened to replace the harmful sensitivity rule in round 8.

The protocol was then rerun for \ErepN{} further trajectories that differed
only in the proposer's own sampling randomness. Of \ErepRoundsPlanned{} planned
rounds, \ErepProposals{} produced an update that reached the validator, and
\ErepInterfaceRej{} updates failed the typed interface. In \ErepStops{}
rounds the proposer declared that no further change would help, which closed
its trajectory and canceled any rounds left in it. The independent evaluation
places \ErepInvalidProposed{} of the \ErepProposals{} updates above 0.05,
about one in four, and every one of them failed validation. A further
\ErepConservRefusals{} updates whose evaluated rate lies below 0.05 also
failed. The \ErepUniqueAdm{} distinct versions that passed have evaluated
rates of at most \ErepAdmTruthMax.

Per trajectory, the median number of updates that passed was \ErepAdoptMed{} (range
\ErepAdoptLo{} to \ErepAdoptHi). The median final deferral rate was
\ErepEscMed{} (range \ErepEscLo{} to \ErepEscHi), and the median final
autonomous power was \ErepPowMed{} (range \ErepPowLo{} to \ErepPowHi). The
original trajectory, run before the replication, is typical of them.

A greedy optimizer then played the adversary that Theorem~\ref{thm:lifetime}
allows. In each round it simulated all single-field changes on its own pilot
studies and submitted the one with the highest estimated power, coming back to
failed directions later. Over \EscrRounds{} rounds it tried information
recovery and sensitivity promotion twice each. It also tried exclusion cuts,
which it relaxed after each failure until one passed. It collected
\EscrRefused{} failures, and its \EscrAdopted{} versions that passed have
bounds of at most \EscrMaxAdm. Its final version, with power \EscrPow{} and
deferral rate \EscrEsc, matches the benign trajectories. Searching for the edge
of validation bought the adversary nothing.

\subsection{Valid analyses do not compose}
\label{sec:exp-nonclosure}

By Theorem~\ref{thm:nonclosure}, an analyst who reports whichever of $K$ exact
level-$\alpha$ analyses rejects can reach a Type I error rate of
$\min(1,K\alpha)$, and this constant is attained. Table~\ref{tab:nonclosure}
checks both parts with \SoneNrepA{} simulated replicates per row. Independent
analyses give $1-(1-\alpha)^K$, and analyses built to have disjoint rejection
regions attain $K\alpha$. Each simulated value lies within Monte Carlo error of
its theoretical value, with a largest deviation of $\SoneMaxDev\times10^{-4}$.

\begin{table}[!htbp]
\centering
\caption{Selecting among $K$ exact analyses after seeing the data}
\label{tab:nonclosure}
\begin{adjustbox}{max width=\textwidth}\begin{tabular}{ccccc}
\toprule
 & \multicolumn{2}{c}{Independent analyses} & \multicolumn{2}{c}{Disjoint rejection regions} \\
\cmidrule(lr){2-3}\cmidrule(lr){4-5}
$K$ & Simulated & $1-(1-\alpha)^K$ & Simulated & $\min(1,K\alpha)$ \\
\midrule
2 & \SoneTabSInd & \SoneTabSIndTh & \SoneTabSAdv & \SoneTabSAdvTh \\
3 & \SoneTabTInd & \SoneTabTIndTh & \SoneTabTAdv & \SoneTabTAdvTh \\
4 & \SoneTabUInd & \SoneTabUIndTh & \SoneTabUAdv & \SoneTabUAdvTh \\
5 & \SoneTabVInd & \SoneTabVIndTh & \SoneTabVAdv & \SoneTabVAdvTh \\
8 & \SoneTabYInd & \SoneTabYIndTh & \SoneTabYAdv & \SoneTabYAdvTh \\
10 & \SoneTabRQInd & \SoneTabRQIndTh & \SoneTabRQAdv & \SoneTabRQAdvTh \\
\bottomrule
\end{tabular}\end{adjustbox}

\addtabletext{Type I error rate of an analyst who reports whichever of $K$ exact
level-$0.05$ analyses rejects, with \SoneNrepA{} simulated replicates per row.\par}
\end{table}

The construction with disjoint rejection regions is extreme, so the same
question was put to two analyses that a bioequivalence statistician would
defend. The two one-sided tests \cite{si:schuirmann1987comparison} are exactly
level $\alpha$ at the equivalence boundary. The distribution-free procedure of
Hauschke and colleagues \cite{si:hauschke1990distribution} builds on the
Hodges-Lehmann interval for the shift between the period differences of the two
sequences. This procedure is conservative because its rank statistic is
discrete. Nine scenarios at the boundary cross 24, 36 or 48 subjects with
within-subject coefficients of variation of 20, 25 or 30\%. With \SoneNrepB{}
simulated studies per scenario, the largest estimated rates of the two analyses
are \SoneBeParMax{} and \SoneBeNparMax, and the first lies within Monte Carlo
error of $\alpha$.

An analyst who reports the parametric result and falls back on the
distribution-free one when the parametric analysis fails reaches \SoneBeWorst,
about \SoneBeRatio\% above the nominal level. The two analyses are strongly
dependent, which keeps the inflation moderate. A moderate inflation is exactly
what a check of each analysis on its own misses. A fallback analysis is
also the kind of update an analyst making itself more robust would propose,
which is why every version is validated as a whole.

\subsection{A regulatory switching rule}
\label{sec:exp-rsabe}

For highly variable drugs the FDA scales the equivalence limits with the
reference product's within-subject standard deviation $\sigma_{wR}$
\cite{si:fda2026statistical}. Its mixed procedure chooses its criterion from the
observed $s_{wR}$, and the switching value \SthreeSwitch{} corresponds to a
within-subject coefficient of variation of 30\%. At or above the switching
value, the procedure applies the scaled criterion
$\mu^2-k^2\sigma_{wR}^2\le 0$, with the regulatory constant $k=\log(1.25)/0.25=\SthreeTheta$.
This criterion is tested through the upper confidence bound of Hyslop and
colleagues \cite{si:hyslop2000small}, together with a constraint on the point
estimate. Below the switching value the procedure applies the usual criterion
$|\mu|<\log 1.25$. The switch chooses between two analyses after the data are
seen. It is the situation of Theorem~\ref{thm:nonclosure}, built into a rule in
regulatory use.

The size of a switching rule depends on which configurations count as not
bioequivalent, and this study uses the convention that the switch itself
defines. A product whose true $\sigma_{wR}$ is at most \SthreeSwitch{} is not
bioequivalent when $|\mu|\ge\log 1.25$, and a product whose true $\sigma_{wR}$
exceeds \SthreeSwitch{} is not bioequivalent when $|\mu|\ge k\sigma_{wR}$.
Each criterion is thus judged where the rule intends to apply it. The Type I
error rate is computed on the boundary of this null set, at $\mu=\log 1.25$
below the switching value and at $\mu=k\sigma_{wR}$ above it. The
computation assumes equal within-subject standard deviations for test and
reference. All numbers reported in this section refer to this convention.

Under this convention the rule rejects far more often than the nominal 0.05
just below the switching value. The design is a balanced four-period full
replicate (TRTR/RTRT) with 24, 36 or 48 subjects. The Type I error rate was
computed at the 23 values $\sigma_{wR}=0.16,0.18,\ldots,0.60$, with \SthreeM{}
simulated studies per value. The largest rates on this grid are \SthreeWorstA{}
for 24 subjects, \SthreeWorstB{} for 36 and \SthreeWorstC{} for 48, each at
$\sigma_{wR}=\SthreeAtC$, the last grid value below the switching value. At 48
subjects the rule rejects \SthreeFactor{} times as often as the nominal 0.05.
The rate climbs steeply as $\sigma_{wR}$ approaches the switching value from
below and falls under 0.05 just above it. These grid maxima are therefore lower
bounds, up to Monte Carlo error, on the largest rate over the null boundary.

The mechanism is a selection effect just below the switching value, where the
product lies in the null of the fixed criterion. Its observed $s_{wR}$
nevertheless crosses the switching value often enough that the wider scaled
limits are applied to a product that the fixed limits would not approve. Under
this convention the inflation grows with the number of subjects, because more
subjects raise the power of the scaled test whenever the switch is taken.
Studies of
reference-scaled rules have reported inflation of the same kind
\cite{si:wonnemann2015inflation,si:labes2016inflation,si:munoz2016consumer}. A rule in
regulatory use thus shows the failure that continual validation prevents, with no
language model involved. Each branch is valid on its own, and the composed
procedure is not.

\subsection{The validity reserve}
\label{sec:exp-reserve}

A simulation can show that a Type I error rate lies below the target, never that
it equals the target. A procedure that is exactly level $\alpha$ on the null
boundary therefore cannot be validated at $\alpha$
(Proposition~\ref{prop:slack}). This study measured the reserve that
validation needs, on a null grid of \StwoCells{} scenarios with \StwoM{}
simulated studies each. The scenarios cross both equivalence boundaries, seven
within-subject coefficients of variation from 15 to 50\% and 24, 36 or 48
subjects. They also cross balanced or 60/40 allocation with dropout of 0, 10 or
20\%. The two one-sided tests at $\alpha=0.05$ have a largest estimated rate of
\StwoMaxTie. Their largest Clopper and Pearson bound, simultaneous over the
scenarios at confidence 0.95 through a Sidak split
\cite{si:clopper1934use,si:sidak1967rectangular}, is \StwoMaxUcb, so the test fails
validation. Run at level \StwoAlphaZero, a reserve of \StwoGamma, the same
test has a largest bound of \StwoMaxUcbRes{} and passes.

The reserve is the smallest value, in steps of 0.0002, at which a test would
pass the simultaneous validation with probability 0.90 if every scenario
had the Type I error rate $0.05-\gamma$. It shrinks roughly like $m^{-1/2}$ as
the number $m$ of simulated studies per scenario grows, from \StwoResAVQ{} at
$m=50{,}000$ to \StwoResARQQQ{} at $m=10^6$. Its cost in power is modest,
\StwoPowerLossMean{} percentage points on average and \StwoPowerLossMax{} at
most. The mean and the maximum are taken over 18 scenarios with true ratios from
0.90 to 1.00, coefficients of variation of 20 or 30\% and 24 to 48 subjects.
For the same confidence, the intersection-union test of
Theorem~\ref{thm:iut} needs no division of the level among the scenarios,
which shrinks each bound and the needed reserve with it.

\subsection{Language models as analysts}
\label{sec:exp-e1}

An earlier experiment asked how language models use analytic freedom when they
act as the analyst. Each model received completed bioequivalence studies whose
true formulation difference lay at the equivalence boundary, so that declaring
bioequivalence is a false rejection. Each study came with the results of two
individually valid analyses, a primary and a backup. Charters, the written
instructions that define the analyst's role, differed only in the freedom they
granted. They ranged from a locked charter that fixes the primary analysis to
one that names no primary analysis, with a variant that adds pressure from the
sponsor. The same studies were presented under each charter, and studies in
which the two analyses disagree, the decisive studies, were oversampled. The
mid-tier model, claude-sonnet-5, analyzed \EoneSoneCQNamedN{} studies per
charter, 300 of them decisive. In a later wave the smallest and the largest
model, claude-haiku-4-5 and claude-opus-5, analyzed 250 studies under each
charter of Table~\ref{tab:e1}, 200 of them decisive.

\begin{table}[!htbp]
\centering
\caption{Deferrals to a human on decisive studies, by model and charter}
\label{tab:e1}
\begin{adjustbox}{max width=\textwidth}\begin{tabular}{lccc}
\toprule
Charter & claude-haiku-4-5 & claude-sonnet-5 & claude-opus-5 \\
\midrule
Locked & 5/200 & \EoneSoneCQNamedEsc/300 & 1/200 \\
Expert discretion & 78/200 & \EoneSoneCSNamedEsc/300 & 0/200 \\
Sponsor pressure & 99/200 & \EoneSoneCTNamedEsc/300 & 2/200 \\
No primary named & 35/200 & \EoneSoneCUNoneEsc/300 & 198/200 \\
\bottomrule
\end{tabular}\end{adjustbox}

\addtabletext{A decisive study is one in which the primary and the backup
analysis disagree. The advisory charter was run at full size for
claude-sonnet-5 only and is omitted. Pooled over all charters, a failed primary
analysis was replaced by its passing backup in \EwvHkSubX{} of \EwvHkSubN{}
presentations for claude-haiku-4-5. The corresponding counts are \EwvSnSubX{}
of \EwvSnSubN{} for claude-sonnet-5 and \EwvOpSubX{} of \EwvOpSubN{} for
claude-opus-5. Each denominator counts the presentations in which the primary
failed and the backup passed.\par}
\end{table}

Deferral to a human depended strongly on the model (Table~\ref{tab:e1}). No
model deferred on a study in which the two analyses agreed. All three models
rarely or never replaced a failed primary analysis by its passing backup (table
notes). On the decisive studies, however, the three models parted ways. Under
the charter that names no primary analysis, the smallest model deferred on 35
of 200 decisive studies. The mid-tier model deferred on \EoneSoneCUNoneEsc{} of
300, and the largest model on 198 of 200. Under sponsor pressure the order
reversed, and the smallest model deferred most often.

Deferral never raises the Type I error rate (Proposition~\ref{prop:escalation}),
so all three models erred toward safety, each in its own way and to its own
degree under identical instructions. A guarantee that rested on such behavior
would have to be established again for every model and every update. Continual
validation checks what the deployed version is able to do, so the guarantee
holds whichever model chooses among the allowed paths
(Proposition~\ref{prop:envelope}). A model swap that keeps those paths
therefore needs no new validation (Table~\ref{tab:modifications}).

\subsection{Checks of the analysis code}
\label{sec:oracle}

The analysis code was checked against published reference results for
two-period and replicate designs. Sch\"{u}tz and colleagues published reference
datasets for \EtwoNsets{} two-period bioequivalence studies together with the
results of EquivTest/PK, SAS, WinNonlin and R \cite{si:schutz2014reference}. Our
implementation reproduces every point estimate and 90\% confidence limit on all
of them to the two decimals reported, with no discrepancy
(Table~\ref{tab:oracle}). The datasets include an imbalanced one (C) and a large
imbalanced one with outliers (H), on which a further package in that study
disagreed.

For replicate designs, a fixed-effects implementation of the EMA's Method A
\cite{si:ema2015pkwpqa} reproduces the published point estimate and 90\% interval
of the EMA reference dataset. It gives a within-reference coefficient of
variation of \EtwoDSoneCVwR\% against the published \EtwoDSonePubCVwR\%. On a
second replicate reference dataset
\cite{si:schutz2020reference,si:schutz2022replicatebe} it gives a point estimate of
\EtwoDStwosixPE\% against the published \EtwoDStwosixPubPE\%.

\begin{table}[!htbp]
\centering
\caption{The analysis code on published two-period reference datasets}
\label{tab:oracle}
\begin{adjustbox}{max width=\textwidth}\begin{tabular}{cccc}
\toprule
Dataset & Subjects & This implementation & Published \\
\midrule
A & \EtwoAN & \EtwoAPe{} (\EtwoALo, \EtwoAHi) & \EtwoAPubpe{} (\EtwoAPublo, \EtwoAPubhi) \\
B & \EtwoBN & \EtwoBPe{} (\EtwoBLo, \EtwoBHi) & \EtwoBPubpe{} (\EtwoBPublo, \EtwoBPubhi) \\
C & \EtwoCN & \EtwoCPe{} (\EtwoCLo, \EtwoCHi) & \EtwoCPubpe{} (\EtwoCPublo, \EtwoCPubhi) \\
D & \EtwoDN & \EtwoDPe{} (\EtwoDLo, \EtwoDHi) & \EtwoDPubpe{} (\EtwoDPublo, \EtwoDPubhi) \\
E & \EtwoEN & \EtwoEPe{} (\EtwoELo, \EtwoEHi) & \EtwoEPubpe{} (\EtwoEPublo, \EtwoEPubhi) \\
F & \EtwoFN & \EtwoFPe{} (\EtwoFLo, \EtwoFHi) & \EtwoFPubpe{} (\EtwoFPublo, \EtwoFPubhi) \\
G & \EtwoGN & \EtwoGPe{} (\EtwoGLo, \EtwoGHi) & \EtwoGPubpe{} (\EtwoGPublo, \EtwoGPubhi) \\
H & \EtwoHN & \EtwoHPe{} (\EtwoHLo, \EtwoHHi) & \EtwoHPubpe{} (\EtwoHPublo, \EtwoHPubhi) \\
\bottomrule
\end{tabular}\end{adjustbox}

\addtabletext{Point estimate and 90\% confidence interval of the ratio of test to
reference, in percent. Published values are those reported by Sch\"{u}tz and
colleagues \cite{si:schutz2014reference}.\par}
\end{table}

\FloatBarrier

\section{Reproducibility and software}
\label{sec:repro}

The code and the construction of the null scenarios built from real studies
will be deposited in a public repository upon publication
\cite{si:archive2026lifetime}. The repository will also hold the validation
results and the outputs behind every number. Every figure is drawn by code from
the outputs of the same programs.

\subsection*{Software}
The computations of Fig.~\ref{fig:si-concept} use Python 3.9 with NumPy 2.0 and
SciPy 1.13. The supplementary studies of sections~\ref{sec:other-decisions}
and~\ref{sec:earlier-studies} were run with Python 3.12 and NumPy 2.5. Rerun
with Python 3.9 and NumPy 2.0, the study of section~\ref{sec:exp-nonclosure}
reproduces its archived outputs exactly. 
The examples of
sections~\ref{sec:exp-be} and~\ref{sec:exp-launch} run on vectorized
reimplementations of the experimental bioequivalence agent
and of gbstats at commit \texttt{b71bbbc}
under Python 3.11 with NumPy 1.26, SciPy 1.13 and CuPy 13.6.
Fig.~1 is an
editable PowerPoint slide written by a Python script with python-pptx and exported
to PDF by Microsoft PowerPoint, and Figs.~\ref{fig:results} and~\ref{fig:robust}, like the data panels of Fig.~\ref{fig:si-concept},
are drawn with matplotlib 3.9 in the Helvetica typeface of the captions. Exact one-sided Clopper and
Pearson bounds are computed from quantiles of the beta distribution.

\subsection*{Randomness}
Every simulation is reproducible from recorded random seeds, and the simulated
studies of a validation round are drawn after its version is frozen.

\subsection*{Data}
The ten bioequivalence studies are distributed with the R package replicateBE
\cite{si:schutz2022replicatebe} under the GNU General Public License, version 3,
and were first published in the sources it lists
\cite{si:ema2015pkwpqa,si:fda2017bestudies,si:patterson2012viewpoint,si:shumaker1998phenytoin,si:chow2008design,si:hauschke2007bioequivalence,si:patterson2017bioequivalence}.
The EMA datasets may be reproduced with acknowledgment of the source, and the FDA
files are works of the United States Government. The five datasets of
section~\ref{sec:exp-be-scenarios} that validation never uses, numbers 7, 8, 22, 25 and 28, ship with the
same package and were simulated by its authors. Section~\ref{sec:oracle} uses
the reference datasets for \EtwoNsets{} two-period bioequivalence studies
published by Sch\"{u}tz and colleagues \cite{si:schutz2014reference} and the
replicateBE datasets 1 and 26.

\subsection*{Language models}
Both examples use Claude Sonnet 5.5 (claude-sonnet-5-5) as proposer, sampled
at temperature 1. The supplementary studies used claude-sonnet-5,
claude-haiku-4-5 and claude-opus-5 with fixed instructions. All raw responses
of the models are kept with the code, so every analysis of them reruns without
access to the models.

\FloatBarrier
\renewcommand{\refname}{Appendix References}


\begin{thebibliography}{10}
\bibitem{lu2024aiscientist}
C. Lu, et~al., Towards end-to-end automation of {AI} research.
\newblock {\em\protect\JournalTitle{Nature}} \textbf{651}, 914--919 (2026).

\bibitem{swanson2025virtuallab}
K. Swanson, W. Wu, N.~L. Bulaong, J.~E. Pak, J. Zou, The {Virtual Lab} of {AI}
  agents designs new {SARS-CoV-2} nanobodies.
\newblock {\em\protect\JournalTitle{Nature}} \textbf{646}, 716--723 (2025).

\bibitem{alber2026cellvoyager}
S. Alber, et~al., {CellVoyager}: {AI} {CompBio} agent generates new insights by
  autonomously analyzing biological data.
\newblock {\em\protect\JournalTitle{Nature Methods}} \textbf{23}, 749--759
  (2026).

\bibitem{zhang2026virtualbiotech}
H.~G. Zhang, P. Eckmann, J. Miao, A.~B. Mahon, J. Zou, The {Virtual Biotech}: A
  multi-agent {AI} framework for therapeutic discovery and development.
\newblock {\em\protect\JournalTitle{Science}} (2026) \unskip.
  \url{https://doi.org/10.1126/science.aeg6779} (accessed 6 October 2026).

\bibitem{chen2025scienceagentbench}
Z. Chen, et~al., {ScienceAgentBench}: Toward rigorous assessment of language
  agents for data-driven scientific discovery in {\em International Conference
  on Learning Representations}.
\newblock (2025).

\bibitem{majumder2025discoverybench}
B.~P. Majumder, et~al., {DiscoveryBench}: Towards data-driven discovery with
  large language models in {\em International Conference on Learning
  Representations}.
\newblock (2025).

\bibitem{chan2025mlebench}
J.~S. Chan, et~al., {MLE}-bench: Evaluating machine learning agents on machine
  learning engineering in {\em International Conference on Learning
  Representations}.
\newblock (2025).

\bibitem{nie2026dsgym}
F. Nie, et~al., {DSGym}: A holistic framework for evaluating and training data
  science agents.
\newblock {\em\protect\JournalTitle{arXiv [Preprint]}} (2026) \unskip.
  \url{https://arxiv.org/abs/2601.16344} (accessed 6 October 2026).

\bibitem{huang2025popper}
K. Huang, et~al., Automated hypothesis validation with agentic sequential
  falsifications in {\em Proceedings of the 42nd International Conference on
  Machine Learning}, Proceedings of Machine Learning Research.
\newblock Vol.{} 267, pp. 25372--25437 (2025).

\bibitem{miao2026fisherr1}
J. Miao, J. Mu, G. Chen, J. Zou, {Fisher-R1}: Training {LLM} agents for
  reliable hypothesis testing.
\newblock {\em\protect\JournalTitle{arXiv [Preprint]}} (2026) \unskip.
  \url{https://arxiv.org/abs/2608.07437} (accessed 6 October 2026).

\bibitem{schuirmann1987comparison}
D.~J. Schuirmann, A comparison of the two one-sided tests procedure and the
  power approach for assessing the equivalence of average bioavailability.
\newblock {\em\protect\JournalTitle{Journal of Pharmacokinetics and
  Biopharmaceutics}} \textbf{15}, 657--680 (1987).

\bibitem{ich2024m13a}
{International Council for Harmonisation}, {ICH} {M13A} guideline on
  bioequivalence for immediate-release solid oral dosage forms, (ICH),
  Harmonised guideline (2024).

\bibitem{neyman1933problem}
J. Neyman, E.~S. Pearson, On the problem of the most efficient tests of
  statistical hypotheses.
\newblock {\em\protect\JournalTitle{Philos. Trans. R. Soc. Lond. A}}
  \textbf{231}, 289--337 (1933).

\bibitem{lehmann2022testing}
E.~L. Lehmann, J.~P. Romano, {\em Testing Statistical Hypotheses} (Springer,
  ed. 4, 2022).

\bibitem{bertran2026manyai}
M. Bertran, R. Fogliato, Z.~S. Wu, Many {AI} analysts, one dataset: Navigating
  the agentic data science multiverse.
\newblock {\em\protect\JournalTitle{Proc. Natl. Acad. Sci. U.S.A.}}
  \textbf{123}, e2606495123 (2026).

\bibitem{miao2026agentic}
J. Miao, J.~K. Pritchard, J. Zou, The agentic garden of forking paths.
\newblock {\em\protect\JournalTitle{arXiv [Preprint]}} (2026) \unskip.
  \url{https://arxiv.org/abs/2607.01507} (accessed 6 October 2026).

\bibitem{tocher1950extension}
K.~D. Tocher, Extension of the {Neyman-Pearson} theory of tests to
  discontinuous variates.
\newblock {\em\protect\JournalTitle{Biometrika}} \textbf{37}, 130--144 (1950).

\bibitem{chen2024chatgptdrift}
L. Chen, M. Zaharia, J. Zou, How is {ChatGPT}'s behavior changing over time?
\newblock {\em\protect\JournalTitle{Harvard Data Science Review}} \textbf{6}
  (2024).

\bibitem{schmidhuber2007godel}
J. Schmidhuber, {G{\"o}del} machines: Fully self-referential optimal universal
  self-improvers in {\em Artificial General Intelligence}, B. Goertzel,
  C. Pennachin, Eds. (Springer, 2007), pp. 199--226.

\bibitem{zelikman2024stop}
E. Zelikman, E. Lorch, L. Mackey, A.~T. Kalai, Self-taught optimizer ({STOP}):
  Recursively self-improving code generation in {\em Conference on Language
  Modeling}.
\newblock (2024).

\bibitem{zhang2026darwin}
J. Zhang, S. Hu, C. Lu, R. Lange, J. Clune, {Darwin G{\"o}del Machine}:
  Open-ended evolution of self-improving agents in {\em International
  Conference on Learning Representations}.
\newblock (2026).

\bibitem{shinn2023reflexion}
N. Shinn, F. Cassano, A. Gopinath, K. Narasimhan, S. Yao, {Reflexion}: Language
  agents with verbal reinforcement learning in {\em Advances in Neural
  Information Processing Systems}.
\newblock Vol.{}~36, pp. 8634--8652 (2023).

\bibitem{yuksekgonul2025textgrad}
M. Yuksekgonul, et~al., Optimizing generative {AI} by backpropagating language
  model feedback.
\newblock {\em\protect\JournalTitle{Nature}} \textbf{639}, 609--616 (2025).

\bibitem{zhang2026ace}
Q. Zhang, et~al., Agentic context engineering: Evolving contexts for
  self-improving language models in {\em International Conference on Learning
  Representations}.
\newblock (2026).

\bibitem{simmons2011false}
J.~P. Simmons, L.~D. Nelson, U. Simonsohn, False-positive psychology:
  Undisclosed flexibility in data collection and analysis allows presenting
  anything as significant.
\newblock {\em\protect\JournalTitle{Psychological Science}} \textbf{22},
  1359--1366 (2011).

\bibitem{silberzahn2018many}
R. Silberzahn, et~al., Many analysts, one data set: Making transparent how
  variations in analytic choices affect results.
\newblock {\em\protect\JournalTitle{Advances in Methods and Practices in
  Psychological Science}} \textbf{1}, 337--356 (2018).

\bibitem{angelopoulos2023ppi}
A.~N. Angelopoulos, S. Bates, C. Fannjiang, M.~I. Jordan, T. Zrnic,
  Prediction-powered inference.
\newblock {\em\protect\JournalTitle{Science}} \textbf{382}, 669--674 (2023).

\bibitem{jin2026txconformal}
Y. Jin, et~al., {TxConformal}: Controlling false discoveries in {AI}-driven
  therapeutic discovery.
\newblock {\em\protect\JournalTitle{bioRxiv [Preprint]}} (2026) \unskip.
  \url{https://doi.org/10.64898/2026.04.27.721076} (accessed 6 October 2026).

\bibitem{dwork2015reusable}
C. Dwork, et~al., The reusable holdout: Preserving validity in adaptive data
  analysis.
\newblock {\em\protect\JournalTitle{Science}} \textbf{349}, 636--638 (2015).

\bibitem{benjamini2020selective}
Y. Benjamini, Selective inference: The silent killer of replicability.
\newblock {\em\protect\JournalTitle{Harvard Data Science Review}} \textbf{2}
  (2020).

\bibitem{tian2021online}
J. Tian, A. Ramdas, Online control of the familywise error rate.
\newblock {\em\protect\JournalTitle{Statistical Methods in Medical Research}}
  \textbf{30}, 976--993 (2021).

\bibitem{bauer1994evaluation}
P. Bauer, K. K{\"o}hne, Evaluation of experiments with adaptive interim
  analyses.
\newblock {\em\protect\JournalTitle{Biometrics}} \textbf{50}, 1029--1041
  (1994).

\bibitem{thomas2019preventing}
P.~S. Thomas, et~al., Preventing undesirable behavior of intelligent machines.
\newblock {\em\protect\JournalTitle{Science}} \textbf{366}, 999--1004 (2019).

\bibitem{feng2021approval}
J. Feng, S. Emerson, N. Simon, Approval policies for modifications to machine
  learning-based software as a medical device: A study of bio-creep.
\newblock {\em\protect\JournalTitle{Biometrics}} \textbf{77}, 31--44 (2021).

\bibitem{sengupta2026selfevolving}
B. Sengupta, Self-evolving agents with anytime-valid certificates.
\newblock {\em\protect\JournalTitle{arXiv [Preprint]}} (2026) \unskip.
  \url{https://arxiv.org/abs/2607.00871} (accessed 6 October 2026).

\bibitem{sun2026reuse}
X. Sun, Y. Zeng, Z. She, X. Wang, Which self-improvements should we trust?
  {R}eliable self-improvement when agents reuse their benchmarks.
\newblock {\em\protect\JournalTitle{arXiv [Preprint]}} (2026) \unskip.
  \url{https://arxiv.org/abs/2609.33180} (accessed 6 October 2026).

\bibitem{angelopoulos2025learn}
A.~N. Angelopoulos, S. Bates, E.~J. Cand{\`e}s, M.~I. Jordan, L. Lei, Learn
  then test: Calibrating predictive algorithms to achieve risk control.
\newblock {\em\protect\JournalTitle{The Annals of Applied Statistics}}
  \textbf{19} (2025).

\bibitem{clopper1934use}
C.~J. Clopper, E.~S. Pearson, The use of confidence or fiducial limits
  illustrated in the case of the binomial.
\newblock {\em\protect\JournalTitle{Biometrika}} \textbf{26}, 404--413 (1934).

\bibitem{ema2010bioequivalence}
{European Medicines Agency}, Guideline on the investigation of bioequivalence,
  (Committee for Medicinal Products for Human Use), Guideline
  CPMP/EWP/QWP/1401/98 Rev. 1 (2010).

\bibitem{fda2026statistical}
{U.S. Food and Drug Administration}, Statistical approaches to establishing
  bioequivalence, (FDA), Guidance for industry (2026).

\bibitem{davit2009comparing}
B.~M. Davit, et~al., Comparing generic and innovator drugs: A review of 12 years
  of bioequivalence data from the {United States Food and Drug Administration}.
\newblock {\em\protect\JournalTitle{Annals of Pharmacotherapy}} \textbf{43},
  1583--1597 (2009).

\bibitem{chow2008design}
S.-C. Chow, J.-P. Liu, {\em Design and Analysis of Bioavailability and
  Bioequivalence Studies} (Chapman \& Hall/CRC, ed. 3, 2008).

\bibitem{hauschke2007bioequivalence}
D. Hauschke, V. Steinijans, I. Pigeot, {\em Bioequivalence Studies in Drug
  Development: Methods and Applications} (Wiley, 2007).

\bibitem{tothfalusi2009evaluation}
L. T{\'o}thfalusi, L. Endr{\'e}nyi, A. Garc{\'i}a-Arieta, Evaluation of
  bioequivalence for highly variable drugs with scaled average bioequivalence.
\newblock {\em\protect\JournalTitle{Clinical Pharmacokinetics}} \textbf{48},
  725--743 (2009).

\bibitem{davit2012implementation}
B.~M. Davit, et~al., Implementation of a reference-scaled average
  bioequivalence approach for highly variable generic drug products by the {US}
  {Food and Drug Administration}.
\newblock {\em\protect\JournalTitle{AAPS J.}} \textbf{14}, 915--924 (2012).

\bibitem{karalis2012bioequivalence}
V. Karalis, M. Symillides, P. Macheras, Bioequivalence of highly variable drugs:
  A comparison of the newly proposed regulatory approaches by {FDA} and {EMA}.
\newblock {\em\protect\JournalTitle{Pharmaceutical Research}} \textbf{29},
  1066--1077 (2012).

\bibitem{schutz2022replicatebe}
H. Sch{\"u}tz, M. Tomashevskiy, D. Labes, {replicateBE}: Average bioequivalence
  with expanding limits ({ABEL}) (R package version 1.1.3, CRAN) (2022)
  \unskip. \url{https://doi.org/10.32614/CRAN.package.replicateBE} (accessed 6
  October 2026).

\bibitem{schutz2020reference}
H. Sch{\"u}tz, et~al., Reference datasets for studies in a replicate design
  intended for average bioequivalence with expanding limits.
\newblock {\em\protect\JournalTitle{AAPS J.}} \textbf{22}, 44 (2020).

\bibitem{labes2016inflation}
D. Labes, H. Sch{\"u}tz, Inflation of type {I} error in the evaluation of
  scaled average bioequivalence, and a method for its control.
\newblock {\em\protect\JournalTitle{Pharmaceutical Research}} \textbf{33},
  2805--2814 (2016).

\bibitem{munoz2016consumer}
J. Mu{\~n}oz, D. Alcaide, J. Oca{\~n}a, Consumer's risk in the {EMA} and {FDA}
  regulatory approaches for bioequivalence in highly variable drugs.
\newblock {\em\protect\JournalTitle{Stat. Med.}} \textbf{35}, 1933--1943
  (2016).

\bibitem{kohavi2020trustworthy}
R. Kohavi, D. Tang, Y. Xu, {\em Trustworthy Online Controlled Experiments: A
  Practical Guide to {A/B} Testing} (Cambridge University Press, 2020).

\bibitem{larsen2024statistical}
N. Larsen, et~al., Statistical challenges in online controlled experiments: A
  review of {A/B} testing methodology.
\newblock {\em\protect\JournalTitle{The American Statistician}} \textbf{78},
  135--149 (2024).

\bibitem{johari2017peeking}
R. Johari, P. Koomen, L. Pekelis, D. Walsh, Peeking at {A/B} tests: Why it
  matters, and what to do about it in {\em Proceedings of the 23rd {ACM}
  {SIGKDD} International Conference on Knowledge Discovery and Data Mining}.
\newblock pp. 1517--1525 (2017).

\bibitem{growthbook2026}
{GrowthBook, Inc.}, {GrowthBook}: Open source feature flags, experimentation,
  and product analytics (GitHub repository, statistics engine gbstats under the
  MIT License) (2026) \unskip. \url{https://github.com/growthbook/growthbook}
  (accessed 7 October 2026).

\bibitem{fda2019adaptive}
{U.S. Food and Drug Administration}, Adaptive designs for clinical trials of
  drugs and biologics: Guidance for industry (Final guidance) (2019) \unskip.
  \url{https://www.fda.gov/media/78495/download} (accessed 6 October 2026).

\bibitem{fda2024pccp}
{U.S. Food and Drug Administration}, Marketing submission recommendations for a
  predetermined change control plan for artificial intelligence-enabled device
  software functions: Guidance for industry and {Food and Drug Administration}
  staff (Final guidance, originally issued 4 December 2024) (2025) \unskip.
  \url{https://www.fda.gov/media/166704/download} (accessed 6 October 2026).

\bibitem{berger1982multiparameter}
R.~L. Berger, Multiparameter hypothesis testing and acceptance sampling.
\newblock {\em\protect\JournalTitle{Technometrics}} \textbf{24}, 295--300
  (1982).

\end{thebibliography}

\begin{thebibliography}{10}

\bibitem{si:ema2010bioequivalence}
{European Medicines Agency}, Guideline on the investigation of bioequivalence,
  (Committee for Medicinal Products for Human Use), Guideline
  CPMP/EWP/QWP/1401/98 Rev. 1 (2010).

\bibitem{si:clopper1934use}
C.~J. Clopper, E.~S. Pearson, The use of confidence or fiducial limits
  illustrated in the case of the binomial.
\newblock {\em\protect\JournalTitle{Biometrika}} \textbf{26}, 404--413 (1934).

\bibitem{si:berger1982multiparameter}
R.~L. Berger, Multiparameter hypothesis testing and acceptance sampling.
\newblock {\em\protect\JournalTitle{Technometrics}} \textbf{24}, 295--300
  (1982).

\bibitem{si:schuirmann1987comparison}
D.~J. Schuirmann, A comparison of the two one-sided tests procedure and the
  power approach for assessing the equivalence of average bioavailability.
\newblock {\em\protect\JournalTitle{Journal of Pharmacokinetics and
  Biopharmaceutics}} \textbf{15}, 657--680 (1987).

\bibitem{si:foster2008alpha}
D.~P. Foster, R.~A. Stine, {$\alpha$}-investing: A procedure for sequential
  control of expected false discoveries.
\newblock {\em\protect\JournalTitle{J. R. Stat. Soc. Series B Stat. Methodol.}}
  \textbf{70}, 429--444 (2008).

\bibitem{si:tian2021online}
J. Tian, A. Ramdas, Online control of the familywise error rate.
\newblock {\em\protect\JournalTitle{Statistical Methods in Medical Research}}
  \textbf{30}, 976--993 (2021).

\bibitem{si:ich2024m13a}
{International Council for Harmonisation}, {ICH} {M13A} guideline on
  bioequivalence for immediate-release solid oral dosage forms, (ICH),
  Harmonised guideline (2024).

\bibitem{si:labes2016inflation}
D. Labes, H. Sch{\"u}tz, Inflation of type {I} error in the evaluation of
  scaled average bioequivalence, and a method for its control.
\newblock {\em\protect\JournalTitle{Pharmaceutical Research}} \textbf{33},
  2805--2814 (2016).

\bibitem{si:wonnemann2015inflation}
M. Wonnemann, C. Fr{\"o}mke, A. Koch, Inflation of the type {I} error:
  Investigations on regulatory recommendations for bioequivalence of highly
  variable drugs.
\newblock {\em\protect\JournalTitle{Pharmaceutical Research}} \textbf{32},
  135--143 (2015).

\bibitem{si:munoz2016consumer}
J. Mu{\~n}oz, D. Alcaide, J. Oca{\~n}a, Consumer's risk in the {EMA} and {FDA}
  regulatory approaches for bioequivalence in highly variable drugs.
\newblock {\em\protect\JournalTitle{Stat. Med.}} \textbf{35}, 1933--1943
  (2016).

\bibitem{si:schutz2022replicatebe}
H. Sch{\"u}tz, M. Tomashevskiy, D. Labes, {replicateBE}: Average bioequivalence
  with expanding limits ({ABEL}) (R package version 1.1.3, CRAN) (2022)
  \unskip. \url{https://doi.org/10.32614/CRAN.package.replicateBE} (accessed 6
  October 2026).

\bibitem{si:schutz2020reference}
H. Sch{\"u}tz, et~al., Reference datasets for studies in a replicate design
  intended for average bioequivalence with expanding limits.
\newblock {\em\protect\JournalTitle{AAPS J.}} \textbf{22}, 44 (2020).

\bibitem{si:sengupta2026selfevolving}
B. Sengupta, Self-evolving agents with anytime-valid certificates.
\newblock {\em\protect\JournalTitle{arXiv [Preprint]}} (2026) \unskip.
  \url{https://arxiv.org/abs/2607.00871} (accessed 6 October 2026).

\bibitem{si:sun2026reuse}
X. Sun, Y. Zeng, Z. She, X. Wang, Which self-improvements should we trust?
  {R}eliable self-improvement when agents reuse their benchmarks.
\newblock {\em\protect\JournalTitle{arXiv [Preprint]}} (2026) \unskip.
  \url{https://arxiv.org/abs/2609.33180} (accessed 6 October 2026).

\bibitem{si:angelopoulos2025learn}
A.~N. Angelopoulos, S. Bates, E.~J. Cand{\`e}s, M.~I. Jordan, L. Lei, Learn
  then test: Calibrating predictive algorithms to achieve risk control.
\newblock {\em\protect\JournalTitle{The Annals of Applied Statistics}}
  \textbf{19} (2025).

\bibitem{si:growthbook2026}
{GrowthBook, Inc.}, {GrowthBook}: Open source feature flags, experimentation,
  and product analytics (GitHub repository, statistics engine gbstats under the
  MIT License) (2026) \unskip. \url{https://github.com/growthbook/growthbook}
  (accessed 7 October 2026).

\bibitem{si:deng2013cuped}
A. Deng, Y. Xu, R. Kohavi, T. Walker, Improving the sensitivity of online
  controlled experiments by utilizing pre-experiment data in {\em Proceedings
  of the Sixth {ACM} International Conference on Web Search and Data Mining}.
\newblock pp. 123--132 (2013).

\bibitem{si:obrien1979multiple}
P.~C. O'Brien, T.~R. Fleming, A multiple testing procedure for clinical trials.
\newblock {\em\protect\JournalTitle{Biometrics}} \textbf{35}, 549--556 (1979).

\bibitem{si:hauschke1990distribution}
D. Hauschke, V.~W. Steinijans, E. Diletti, A distribution-free procedure for
  the statistical analysis of bioequivalence studies.
\newblock {\em\protect\JournalTitle{International Journal of Clinical
  Pharmacology, Therapy and Toxicology}} \textbf{28}, 72--78 (1990).

\bibitem{si:fda2026statistical}
{U.S. Food and Drug Administration}, Statistical approaches to establishing
  bioequivalence: {G}uidance for industry, (Center for Drug Evaluation and
  Research), Final guidance (2026).

\bibitem{si:hyslop2000small}
T. Hyslop, F. Hsuan, D.~J. Holder, A small sample confidence interval approach
  to assess individual bioequivalence.
\newblock {\em\protect\JournalTitle{Stat. Med.}} \textbf{19}, 2885--2897
  (2000).

\bibitem{si:sidak1967rectangular}
Z. {\v{S}}id{\'a}k, Rectangular confidence regions for the means of
  multivariate normal distributions.
\newblock {\em\protect\JournalTitle{Journal of the American Statistical
  Association}} \textbf{62}, 626--633 (1967).

\bibitem{si:schutz2014reference}
H. Sch{\"u}tz, D. Labes, A. Fuglsang, Reference datasets for 2-treatment,
  2-sequence, 2-period bioequivalence studies.
\newblock {\em\protect\JournalTitle{AAPS J.}} \textbf{16}, 1292--1297 (2014).

\bibitem{si:ema2015pkwpqa}
{European Medicines Agency}, Questions \& answers: Positions on specific
  questions addressed to the {Pharmacokinetics Working Party} ({PKWP})
  (EMA/618604/2008 Rev. 13, Committee for Medicinal Products for Human Use)
  (2015) \unskip.
  \url{https://www.ema.europa.eu/en/documents/scientific-guideline/questions-answers-positions-specific-questions-addressed-pharmacokinetics-working-party_en.pdf}
  (accessed 6 October 2026).

\bibitem{si:archive2026lifetime}
X. Sun, Y. Zeng, X. Wang, Reproducibility code and results for
lifetime-valid decisions by continually evolving {AI} analysts
(GitHub, 2026) \unskip.
\url{https://github.com/xiaojing777/lifetime-valid-ai-analysts}
(to be made public upon publication).

\bibitem{si:fda2017bestudies}
{U.S. Food and Drug Administration}, Bioequivalence studies (Archived web page)
  (2017) \unskip.
  \url{https://wayback.archive-it.org/7993/20171115165246/https://www.fda.gov/Drugs/ScienceResearch/ucm301277.htm}
  (archived 15 November 2017, accessed 6 October 2026).

\bibitem{si:patterson2012viewpoint}
S.~D. Patterson, B. Jones, Viewpoint: Observations on scaled average
  bioequivalence.
\newblock {\em\protect\JournalTitle{Pharmaceutical Statistics}} \textbf{11},
  1--7 (2012).

\bibitem{si:shumaker1998phenytoin}
R.~C. Shumaker, C.~M. Metzler, The phenytoin trial is a case study of
  ``individual'' bioequivalence.
\newblock {\em\protect\JournalTitle{Drug Information Journal}} \textbf{32},
  1063--1072 (1998).

\bibitem{si:chow2008design}
S.-C. Chow, J.-p. Liu, {\em Design and Analysis of Bioavailability and
  Bioequivalence Studies}.
\newblock (Chapman and Hall/CRC, Boca Raton, FL), Third edition, (2008).

\bibitem{si:hauschke2007bioequivalence}
D. Hauschke, V. Steinijans, I. Pigeot, {\em Bioequivalence Studies in Drug
  Development}.
\newblock (Wiley, Chichester, UK), (2007).

\bibitem{si:patterson2017bioequivalence}
S.~D. Patterson, B. Jones, {\em Bioequivalence and Statistics in Clinical
  Pharmacology}.
\newblock (Chapman and Hall/CRC, Boca Raton, FL), Second edition, (2017).

\end{thebibliography}
\end{document}